\documentclass[USenglish]{article}
\usepackage[a4paper]{geometry}

\pdfoutput=1
\makeatletter
\def\Hy@Warning#1{} % squelch hyperref
\makeatother

\date{}

\let\originalleft\left
\let\originalright\right
\renewcommand{\left}{\mathopen{}\mathclose\bgroup\originalleft}
\renewcommand{\right}{\aftergroup\egroup\originalright}

\usepackage[dvipsnames]{xcolor}
\usepackage{mathtools}
\renewcommand\MoveEqLeft[1][1.5]{\kern #1em  &   \kern -#1em}
\usepackage{amsmath}
\usepackage{amsthm}
\usepackage{amssymb}
\usepackage[hidelinks]{hyperref}
\usepackage{thmtools}
\usepackage{thm-restate}
\usepackage[capitalize]{cleveref}

\usepackage{float}

\theoremstyle{plain}
\newtheorem{theorem}{Theorem}[section]
\newtheorem{lemma}[theorem]{Lemma}
\newtheorem{observation}[theorem]{Observation}
\newtheorem{corollary}[theorem]{Corollary}

\theoremstyle{definition}
\newtheorem{definition}[theorem]{Theorem}

\theoremstyle{remark}
\newtheorem{claim}[theorem]{Claim}
\newtheorem*{claim*}{Claim}

\newenvironment{claimproof}{\begin{proof}}{\end{proof}}

\def\dense{\medmuskip=2.0mu plus 2.0mu minus 2.0mu
\thinmuskip=2.0mu
\thickmuskip=2.0mu plus 5.0mu}

\usepackage{booktabs}
\usepackage{tabularx}

\newcommand{\Oh}{{\operatorname{O}}}

\DeclareMathOperator{\E}{\mathbb{E}}
\DeclareMathOperator{\Prob}{\mathbb{P}}
\DeclareMathOperator{\Var}{\mathrm{Var}}
\DeclareMathOperator{\Cov}{\mathrm{Cov}}

\DeclareMathOperator{\R}{\mathbb{R}}
\DeclareMathOperator{\N}{\mathbb{N}}

\renewcommand*{\Pr}{{\Prob}}

\DeclareMathOperator{\poly}{\mathrm{poly}} 

\DeclareMathOperator{\1}{\mathbf{1}} 

\newcommand{\GlobalDivergence}{\Upsilon}

\newcommand{\SyncProcName}{\textsc{SBal}}
\newcommand{\AsyncProcName}{\textsc{ABal}}
\newcommand{\SyncProc}[3]{\SyncProcName(#1, #2,#3)}
\def\dcmhackI{}
\newcommand{\AsyncProc}[2]{\AsyncProcName\dcmhackI(#1, #2)}
\newcommand{\BalProcName}{\textsc{Bal}}
\newcommand{\BalProc}[2]{\BalProcName(#1, #2)}

\newcommand{\MatchDistr}{\mathcal{D}}
\newcommand{\randModel}{\MatchDistr_{\textsc{RM}}}
\newcommand{\RMDistr}{\randModel}
\newcommand{\singdModel}{\MatchDistr_{\textsc{A}}}
\newcommand{\ADistr}{\singdModel}
\newcommand{\baldModel}{\MatchDistr_{\textsc{BC}}}
\newcommand{\BCDistr}{\baldModel}

\newcommand{\AutoExp}[1]{{\E[#1]}}
\newcommand{\AutoExpCond}[2]{{\E[#1\mid #2]}}
\newcommand{\AutoVar}[1]{{\Var[#1]}}
\newcommand{\AutoVarCond}[2]{{\Var[#1\mid #2]}}
\newcommand{\AutoCov}[1]{{\Cov[#1]}}
\newcommand{\AutoProb}[1]{{\Prob[#1]}}

\newcommand{\BigAutoExp}[1]{{\E\left[#1\right]}}
\newcommand{\BigAutoExpCond}[2]{{\E\left[#1\,\middle\vert\,#2\right]}}
\newcommand{\BigAutoVar}[1]{{\Var\left[#1\right]}}
\newcommand{\BigAutoVarCond}[2]{{\Var\left[#1\,\middle\vert\,#2\right]}}

\newcommand{\BigAutoProb}[1]{{\Prob\left[#1\right]}}

\DeclarePairedDelimiter\abs{\lvert}{\rvert}
\DeclarePairedDelimiter\norm{\lVert}{\rVert}
\DeclarePairedDelimiterX{\inp}[2]{\langle}{\rangle}{#1, #2}

\DeclareMathOperator{\discr}{\mathrm{disc}}
\DeclareMathOperator{\disc}{\mathrm{disc}}
\DeclareMathOperator{\NodePotential}{\Phi}
\DeclareMathOperator{\EdgePotential}{\Psi}
\DeclareMathOperator{\HittingTime}{t_{\mathrm{hit}}(G)}
\DeclareMathOperator{\EdgeHittingTime}{t^*_{\mathrm{hit}}(G)}
\DeclareMathOperator{\SpectralGap}{\lambda}
\DeclareMathOperator{\Laplacian}{\mathbf{L}}
\DeclareMathOperator{\AdjacencyMat}{\mathbf{A}}
\DeclareMathOperator{\IdentityMat}{\mathbf{I}}

\newcommand{\ResistiveDistance}[2]{\mathrm{Res}(#1, #2)}
\newcommand{\ResistiveDiameter}{\mathrm{Res}(G)}
\newcommand{\MaxEdgeResistance}{\mathrm{Res}^*(G)}

\newcommand{\Adaptive}{\textsc{Adaptive}}
\newcommand{\AdaptiveParam}{\beta}

\newcommand{\BalancingSpeed}{\beta}

\newcommand{\pmax}{p_{\mathrm{max}}}

\newcommand{\CircuitSize}{\zeta}

\newcommand{\LoadSymb}{X}

\newcommand{\NodeLoadT}[2]{\LoadSymb_{#1}(#2)}
\newcommand{\LoadVec}{\vec{\LoadSymb}}
\newcommand{\LoadVecT}[1]{\LoadVec(#1)}

\newcommand{\AllocSymb}{\ell}

\newcommand{\NodeAllocT}[2]{\AllocSymb_{#1}(#2)}
\newcommand{\AllocVec}{\vec{\AllocSymb}}
\newcommand{\AllocVecT}[1]{\AllocVec(#1)}

\newcommand{\InitialContribSymb}{I}
\newcommand{\DynamicContribSymb}{D}
\newcommand{\RoundingContribSymb}{R}

\newcommand{\NodeInitialContribT}[2]{\InitialContribSymb_{#1}(#2)}
\newcommand{\InitialContribVec}{\vec{\InitialContribSymb}}
\newcommand{\InitialContribVecT}[1]{\InitialContribVec(#1)}

\newcommand{\NodeDynamicContribT}[2]{\DynamicContribSymb_{#1}(#2)}
\newcommand{\DynamicContribVec}{\vec{\DynamicContribSymb}}
\newcommand{\DynamicContribVecT}[1]{\DynamicContribVec(#1)}

\newcommand{\NodeRoundingContribT}[2]{\RoundingContribSymb_{#1}(#2)}
\newcommand{\RoundingContribVec}{\vec{\RoundingContribSymb}}
\newcommand{\RoundingContribVecT}[1]{\RoundingContribVec(#1)}

\newcommand{\RoundingErrSymb}{\varepsilon}

\newcommand{\RoundingErrT}[2]{\RoundingErrSymb_{#1}(#2)}
\newcommand{\RoundingErrVec}{\vec{\RoundingErrSymb}}
\newcommand{\RoundingErrVecT}[1]{\RoundingErrVec(#1)}

\newcommand{\FractionalSymb}{F}

\newcommand{\FractionalT}[2]{\FractionalSymb_{#1}(#2)}

\newcommand{\Filter}{\mathcal{F}}
\newcommand{\FilterT}[1]{\Filter(#1)}

\newcommand{\MixMat}{\mathbf{M}}
\newcommand{\MixMatBeta}[1]{\mathbf{M}^{#1}}
\newcommand{\MixMatT}[1]{\MixMat(#1)}
\newcommand{\MixMatTT}[2]{\MixMat^{[#1, #2]}}
\newcommand{\MixMatSeq}[1]{\MixMat^{[#1]}}

\newcommand{\mixMat}{\mathbf{m}}
\newcommand{\mixMatBeta}[1]{\mathbf{m}^{#1}}
\newcommand{\mixMatT}[1]{\mixMat(#1)}
\newcommand{\mixMatTT}[2]{\mixMat^{[#1, #2]}}
\newcommand{\mixMatSeq}[1]{\mixMat^{[#1]}}

\newcommand{\cMixMat}{\mathbf{R}}

\newcommand{\cMixMatTT}[2]{\cMixMat^{[#1, #2]}}

\newcommand{\RoundMat}{\mathbf{R}}

\newcommand{\dx}{{\mathrm{d}x}}

\newcommand{\dt}{{\mathrm{d}t}}

\newcommand{\dvarphi}{{\mathrm{d}\varphi}}

\newcommand{\BallsRoundNode}[3]{{B}({#1,#2,#3})}
\newcommand{\BallRoundContr}[3]{{C}_{#1}{(#2,#3)}}

\def\paragraph#1{\subparagraph*{#1.}}

\allowdisplaybreaks

\title{Dynamic Averaging Load Balancing on Arbitrary Graphs}

\makeatletter
\let\@fnsymbol\@arabic
\makeatother

\crefformat{footnote}{#2\footnotemark[#1]#3} %for author footnotes
\author{Petra Berenbrink\footnotemark[1], Lukas Hintze\cref{uhh}, Hamed Hosseinpour\cref{uhh},\\ Dominik Kaaser\footnotemark[2], Malin Rau\cref{uhh}}
\begin{document}

\maketitle

\footnotetext[1]{\label{uhh}Universität Hamburg, Germany}
\footnotetext[2]{TU Hamburg, Germany}
\footnotetext[3]{Petra Berenbrink, Hamed Hosseinpour, Malin Rau: Supported by DFG Research Group ADYN (FOR 2975) under grant DFG 411362735}
\footnotetext[4]{Petra Berenbrink, Hamed Hosseinpour: Supported by the DFG under grant 427756233}

%\footnotetext[1]{\label{author1}Yale University, USA, firstname.lastname@yale.edu}
%\footnotetext[2]{\label{author2}Universität Hamburg, Germany, firstname.lastname@uni-hamburg.de}
%\footnotetext[3]{Universität Hamburg, Germany, firstname.lastname-1@uni-hamburg.de}
%\footnotetext[4]{TU Hamburg, Germany, firstname.lastname@tuhh.de}
%\footnotetext[5]{Petra Berenbrink: Supported by DFG Research Group ADYN under grant DFG 411362735}

% \author{Petra Berenbrink}{Universität Hamburg, Germany}{petra.berenbrink@uni-hamburg.de}{https://orcid.org/0000-0002-6930-3259}{DFG FOR 2975.}
% \author{Lukas Hintze}{Universität Hamburg, Germany}{lukas.rasmus.hintze@uni-hamburg.de}{}{TODO}
% \author{Hamed Hosseinpour}{Universität Hamburg, Germany}{hamed.hosseinpour@uni-hamburg.de}{https://orcid.org/0000-0003-3625-5913}{TODO}
% \author{Dominik Kaaser}{TU Hamburg, Germany}{dominik.kaaser@tuhh.de}{https://orcid.org/0000-0002-2083-7145}{TODO}
% \author{Malin Rau}{University of Hamburg, Hamburg, Germany}{malin.rau@uni-hamburg.de}{https://orcid.org/0000-0002-5710-560X}{DFG FOR 2975.}

\begin{abstract}
In this paper we study dynamic averaging load balancing on general graphs.
We consider infinite time and dynamic processes, where in every step new load items are assigned to randomly chosen nodes.
A matching is chosen, and the load is averaged over the edges of that matching.
We analyze the discrete case where load items are indivisible, moreover our results also carry over to the continuous case where load items can be split arbitrarily.
For the choice of the matchings we consider three different models, random matchings of linear size, random matchings containing only single edges, and deterministic sequences of matchings covering the whole graph.
We bound the discrepancy, which is defined as the difference between the maximum and the minimum load.
Our results cover a broad range of graph classes and, to the best of our knowledge, our analysis is the first result for discrete and dynamic averaging load balancing processes.
As our main technical contribution we develop a drift result that allows us to apply techniques based on the effective resistance in an electrical network to the setting of dynamic load balancing.
\end{abstract}

\section{Introduction}

Parallel and distributed computing is ubiquitous in science, technology, and beyond. 
Key to the performance of a distributed system is the efficient utilization of resources: in order to obtain a substantial speed-up it is of utmost importance that all processors have to handle the same amount of work.
Unfortunately, many practical applications such as finite element simulations are highly ``irregular'', and the amount of load generated on some processors is much larger than the amount of load generated on others.
We therefore investigate \emph{load balancing} to redistribute the load.
Efficient load balancing schemes have a plenitude of applications, including high performance computing \cite{DBLP:journals/ijhpca/ZhengBMK11}, cloud computing \cite{DBLP:journals/access/MohammadianNHD22}, numerical simulations \cite{DBLP:conf/dimacs/Meyerhenke12}, and finite element simulations \cite{DBLP:journals/aes/PatzakR12}.

In this paper we consider \emph{neighborhood} load balancing on arbitrary graphs with $n$ nodes, where the nodes balance their load in each step only with their direct neighbors.
We assume \emph{discrete} load items as opposed to \emph{continuous} (or \emph{idealized}) load items which can be broken into arbitrarily small pieces.
We study \emph{infinite} and \emph{dynamic} processes where new load items are generated in every step.
We consider two different settings.
In the \emph{synchronous} setting $m$ load items are generated on randomly chosen nodes. 
Then a matching is chosen and the load of the nodes is balanced (via weighted averaging) over the edges of that matching.
Here we further distinguish between two matching models.
We consider the \emph{random matching model} where linear-size  matchings are randomly chosen, and the \emph{balancing circuit model} where the graph is divided deterministically into $d_{\max}$ many matchings.
Here $d_{\max}$ is the maximum degree of any node.
In the \emph{asynchronous model} exactly one load item is generated on a randomly chosen node.
In turn, the node chooses one of its edges at random and balances its load with the corresponding neighbor.
This model can be regarded as a variant of the synchronous model where the randomly chosen matching has size one.
It was introduced by \cite{DBLP:conf/icalp/AlistarhNS20} where the authors show results for cycles assuming continuous load. 
Our goal is to bound the so-called \emph{discrepancy}, which is defined as the maximal load of any node minus the minimal load of any node.

\paragraph{Results in a Nutshell}

In this paper we present, for the three models introduced above, bounds on the expected discrepancy and bounds that hold with high probability. 
Our bounds for the synchronous model with balancing circuits hold for arbitrary graphs $G$, the bounds for the asynchronous model and the synchronous model with random matchings hold for regular graphs $G$ only.
For the asynchronous model and the model with random matchings our bounds on the discrepancy are expressed in terms of hitting times of a standard random walk on $G$, as well as in terms of the spectral gap of the Laplacian of $G$. For the synchronous model with balancing circuits
we express our bounds in terms of the \emph{global divergence}. This can be thought of as a measure of the convergence speed of the Markov chains modeling a random walk on $G$. However, it does not directly measure the speed of convergence of the chain. It accounts for the time period in which the chain keeps a given distance from the stationary (and uniform) distribution. In physics terminology, it is a measure of total \emph{absement}, which is the time-integral of displacement.

For all three infinite processes our bounds on the discrepancy hold at an arbitrary point of time as long as the system is initially empty. 
Otherwise, the bounds hold after an initial time period, its length is a function of the initial discrepancy.    
In the following we give some exemplary results assuming that the system is initially empty and $m=n$. 
For the synchronous model with random matchings and the asynchronous model
we can bound the discrepancy by $\Oh(\sqrt{n}\log(n))$ for \emph{any} regular graph $G$. Our results show 
a polylogarithmic bound on the discrepancy for all regular graphs with a hitting time at most $\Oh(n \poly\log(n))$ (e.g., the two-dimensional torus or the hypercube).
In all models we can bound the discrepancy by $\Oh(\sqrt{n\log(n)})$ for arbitrary  constant-degree regular graphs.
For the full results we refer the reader to \cref{thm:main_sync_random}, \cref{thm:main_sync_circuit}, and \cref{thm:main_async}.
We give a detailed overview on the results on specific graph classes in \cref{table:disc:upperbound} in \cref{sec:conclusions}.

All bounds presented in this paper also hold for the corresponding continuous processes without rounding. 
The authors of \cite{DBLP:conf/icalp/AlistarhNS20} consider the asynchronous process on cycles in the continuous setting where the load items can be divided into arbitrary small pieces. 
They bound the expected discrepancy and show that $\disc(G)=O(\sqrt{n} \log(n))$ for a cycle $G$ with $n$ nodes. 
In contrast, we improve that bound for the cycle to  $\disc(G)= O(\sqrt{n \log(n)})$. 
Note that our result not only bounds the expected discrepancy but it also holds with high probability.

Our main analytical vehicle is a drift theorem that bounds the tail of the sum of a non-increasing sequence of random variables.
Our drift theorem adapts known drift results from the literature, similarly to the Variable Drift Theorem in \cite{DBLP:series/ncs/Lengler20}.

\subsection{Related Work}
There is a vast body of literature on iterative load balancing schemes on graphs where nodes are allowed to balance (or average) their load with neighbors only.
One distinguishes between \emph{diffusion} load balancing where the nodes balance their load with all neighbors at the same time and the \emph{matching model} (or \emph{dimension exchange}) model where the edges which are used for the balancing form a matching. 
In the latter model every resource is only involved in one balancing action per step, which greatly facilitates the analysis. 

In this overview we only consider theoretical results and, as it is beyond the scope of this work to provide a complete survey, we focus on results for discrete load balancing. 
For results about continuous load balancing see, for example,~\cite{DBLP:journals/pc/DiekmannFM99, DBLP:conf/focs/KempeDG03}. 
There are also many results in the context of balancing schemes where not the resources try to balance their load but the tokens (acting as selfish players) try to find a resource with minimum load. See~\cite{DBLP:journals/siamcomp/FischerRV10} for a comprehensive survey about selfish load balancing and~\cite{DBLP:journals/dc/AckermannFHS11,DBLP:journals/corr/HoeferS13,DBLP:conf/ipps/BerenbrinkKLM17} for some recent results.
Another related topic is token distribution where nodes do not balance their entire load with neighbors but send only single tokens over to neighboring nodes with a smaller load.
See \cite{DBLP:journals/siamcomp/GhoshLMMPRRTZ99, DBLP:journals/algorithmica/HeideOW96, DBLP:journals/siamcomp/PelegU89} for the static setting and 
\cite{DBLP:journals/siamcomp/AnshelevichKK08} for the dynamic setting.

\paragraph{Discrete Models}
The authors of~\cite{DBLP:journals/mst/MuthukrishnanGS98} give the first rigorous result for discrete load balancing in the diffusion model. 
They assume that the number of tokens sent along each edge is obtained by rounding down the amount of load that would be sent in the continuous case.
Using this approach they establish that the discrepancy 
is at most $O(n^2)$ after $O(\log(Kn))$ steps, where $K$ is the initial discrepancy.
Similar results for the matching model are shown in~\cite{DBLP:journals/jcss/GhoshM96}. 
While always rounding down may lead to quick stabilization, the discrepancy tends to be quite large, a function of the diameter of the graph. 
Therefore, the authors of~\cite{DBLP:conf/focs/RabaniSW98} suggest to use randomized rounding in order to get a better approximation of the continuous case. 
They show results for a wide class of diffusion and matching load balancing protocols and introduce the so-called \emph{local divergence}, which aggregates the sum of load differences over all edges in all rounds.  
The authors prove that the local divergence gives an upper bound on the maximum deviation between the continuous and discrete case of a protocol. 
In~\cite{DBLP:conf/stoc/FriedrichS09} the authors show several results for a randomized protocol with rounding in the matching model. 
For complete graphs their results show a discrepancy of $O(n\sqrt{\log n})$ after $\Theta(\log(Kn))$ steps. 
Later, \cite{DBLP:journals/jcss/BerenbrinkCFFS15} extended some of these results to the diffusion model. 
In~\cite{DBLP:conf/focs/SauerwaldS12} the authors show that the number of rounds needed to reach constant discrepancy is w.h.p.\ bounded by a function of the spectral gap of the relevant mixing matrix and the initial discrepancy. 
In~\cite{DBLP:journals/jpdc/BerenbrinkFH09} the authors propose a very simple potential function technique to analyze discrete diffusion load balancing schemes, both for discrete and continuous settings. 
In \cite{DBLP:conf/ipps/BerenbrinkFKK19} the authors investigate a load balancing process on complete graphs.
In each round a pair of nodes is selected uniformly at random and completely balance their loads up to a rounding error of $\pm 1$.

The authors of \cite{DBLP:conf/icalp/CaiS17} study load balancing via matchings assuming random placement of the load items.  
The initial load distribution is sampled from exponentially concentrated distributions (including the uniform, binomial, geometric, and Poisson distributions).
The authors show that in this setting the convergence time is smaller than in the worst case setting. 
Regardless of the graph's topology, the discrepancy decreases by a factor of $\sqrt[4]{t}$ within $t$ synchronous rounds. 
Their approach of using concentration inequalities to bound the discrepancy (in terms of the squared $2$-norm of the columns of the matrices underlying the mixing process) strongly influenced our approach.

\paragraph{Dynamic Models}
There are far less results for the dynamic setting where new load enters the system over time. 
In \cite{DBLP:conf/icalp/AlistarhNS20} the authors study a model similar to our asynchronous model.
In each step one load item is allocated to a chosen node.
In the same step the chosen node picks a random neighbor, and the two nodes balance their loads by averaging them (continuous model). The authors show that the expected discrepancy is bounded by $O( n\sqrt{n} \log n)$,  as well as a lower bound on the square of the discrepancy of $\Omega(n)$. 
The authors of \cite{DBLP:journals/siamcomp/AnagnostopoulosKU05} consider load balancing via matchings in a dynamic model where the load is, in every step, distributed by an adversary.
They show the system is stable for sufficiently limited adversaries. 
They also give some upper bounds on the maximum load for the somewhat more restricted adversary. 
The authors of \cite{DBLP:journals/algorithmica/BerenbrinkFM08} consider discrete dynamic diffusion load balancing on arbitrary graphs.  
In each step up to $n$ load items are generated on arbitrary nodes (the allocation is determined by an adversary).
Then the nodes balance their load with each neighbor and finally one load item is deleted from every non-empty node.
The authors show that the system is stable, which means that the total load remains bounded over time (as a function of $n$ alone and independently of the time $t$).

\section{Balancing Models and Notation} \label{sec:model}

We consider the following class of dynamic load balancing processes on $d$-regular graphs~$G$ with $n$ nodes $V(G) = [n]$.
Each process is modeled by a Markov chain $(\LoadVecT{t})_{t\in\N_0}$, where the \emph{load vector} $\LoadVecT{t} = (\NodeLoadT{i}{t})_{i \in [n]} \in \R^n$ is the \emph{state} of the process at the end of step $t$, and $\NodeLoadT{i}{t}$ is the load of node~$i$ at time~$t$.
We measure a load vector's imbalance by the discrepancy $\discr(\vec{x})$, which is the difference between the maximum load and the minimum load  $\discr(\vec{x}) \coloneqq \max_{i \in [n]} x_i - \min_{j \in [n]} x_j$.

We consider two balancing processes, the synchronous process $\SyncProcName$ and the asynchronous process $\AsyncProcName$.
Both processes are parameterized by a \emph{balancing parameter} $\BalancingSpeed$ determining the balancing speed and a matching distribution $\MatchDistr(G)$.
For $\SyncProcName$, $\MatchDistr(G)$ is a distribution over linear-sized matchings of $G$.
For $\AsyncProcName$, $\MatchDistr(G)$ is a distribution over edges of $G$.
$\SyncProcName$ is additionally parameterized by the number of load items $m \in \N^+$ allocated in each round.
$\AsyncProcName$ allocates only one new load item per step.

\paragraph{Synchronous Processes}
The synchronous process $\SyncProc{\MatchDistr(G)}{\BalancingSpeed}{m}$ works as follows.
The process first allocates $m$ items to randomly chosen nodes.
Then it uses the matching distribution $\MatchDistr(G)$ to determine the matching which is applied.
Finally it balances the load over the edges of the matching (see Process $\BalProc{\mixMat}{\BalancingSpeed}$ described below). 
The parameter $\BalancingSpeed \in (0, 1]$ controls the fraction of the load difference that is sent over an edge in a step.

For the synchronous process $\SyncProcName$ we consider two families of matching distributions, random matchings ($\RMDistr(G)$) and balancing circuits ($\BCDistr(G)$). 
$\RMDistr(G)$ is generated according to the following method described in~\cite{DBLP:journals/jcss/GhoshM96}.
First an edge set $S$ is formed by including each edge with probability $1/(4d) - 1/(16d^2) = \Theta(1/d)$, independently from all other edges. 
Then a linear-sized matching $\MixMatT{t} \subseteq S$ is computed locally.
We will use capital $\MixMat$ for randomly chosen matchings. The analysis for 
the random matching model can be found in \cref{sec:analysis:random:matching}.
In the \emph{balancing circuit model} we assume $G$ is covered by $\CircuitSize$ fixed matchings $\mixMatT{1},\ldots, \mixMatT{\CircuitSize}$.
 $\BCDistr(G)$ deterministically chooses matchings in periodic manner such that in step $t$ the matching $\mixMatT{t}=\mixMatT{t \bmod \CircuitSize}$ is chosen.
We will use small $\mixMat$ for deterministically chosen matchings.
The analysis for the balancing circuit model can be found in \cref{sec:analysis_balancing_circuit}.

\paragraph{Asynchronous Process}
The asynchronous process $\AsyncProc{\MatchDistr(G)}{\BalancingSpeed}$ works as follows.
The process first uses $\MatchDistr(G)$ to generate a matching, this time containing one edge only.
The distribution we consider, $\ADistr(G)$, first chooses a node $i$ uniformly at random and then it chooses one of the nodes' edges $(i,j)$ uniformly at random.
Finally one new token is assigned to either node $i$ or $j$ and then the edge $(i,j)$ is used for balancing (see $\BalProc{\mixMat}{\BalancingSpeed}$).
Note that for $\AsyncProc{\ADistr(G)}{\BalancingSpeed}$ the load allocation heavily depends on the edges which are used for balancing.
This makes the analysis for this model quite challenging.
In contrast, in $\SyncProc{\ADistr(G)}{\BalancingSpeed}{m}$ the load allocation and the balancing are independent.
Note that in the case of $d$-regular graphs $\ADistr(G)$ is equivalent to the uniform distribution over all edges or to choosing a random matching of size one.
We analyze the asynchronous model in \cref{sec:asynchronous}.

\medskip

\noindent\colorbox{black!10}{\begin{minipage}{\textwidth-2\fboxsep}
$\SyncProc{\MatchDistr(G)}{\BalancingSpeed}{m}$:
In each round $t \in \N^+$:
\begin{enumerate}
    \item Allocate $m$ discrete, unit-sized load items to the nodes uniformly and independently at random.
    Define $\ell_i(t)$ as the number of tokens assigned to node $i$.
    \item Sample a matching $\MixMatT{t}$ according to $\MatchDistr(G)$. 
   \item Balance with $\BalProc{\MixMatT{t}}{\BalancingSpeed}$ applied to $X_i(t):=X_i(t)+\NodeAllocT{i}{t}$, $i\in \{1,\ldots n\}$.
   \end{enumerate}%
\end{minipage}}

\medskip

\noindent\colorbox{black!10}{\begin{minipage}{\textwidth-2\fboxsep}
$\AsyncProc{\MatchDistr(G)}{\BalancingSpeed}$:
In each round $t \in \N^+$:
\begin{enumerate}
    \item Select an edge $\{i, j\}$ according to $\MatchDistr(G)$.
    \item Allocate a single unit-size load item to either node $i$ or $j$
    with a probability of $1/2$. 
    
    I.e., with prob. $1/2$ set $\ell_i(t)=1$ and $\ell_k=0$ for all $k\neq i$, otherwise 
   set $\ell_j(t)=1$ and $\ell_k=0$ for all $k\neq j$.
    \item Balance with $\BalProc{\MixMatT{t}}{\BalancingSpeed}$ applied to $X_i(t):=X_i(t)+\ell_i(t)$,
        where $\MixMatT{t}$ includes just the edge $\{i,j\}$.
\end{enumerate}
\end{minipage}}

\medskip

\noindent\colorbox{black!10}{\begin{minipage}{\textwidth-2\fboxsep}
$\BalProc{\mixMat}{\BalancingSpeed}$: For each edge $\{i,j\}$ in the matching $\mixMat$ balance loads of $i$ and $j$:
\begin{enumerate}
\item Assume w.l.o.g.\ that $\NodeLoadT{i}{t}\ge \NodeLoadT{j}{t}$. 
\item Let $p=\frac{\BalancingSpeed \cdot (\NodeLoadT{i}{t}-\NodeLoadT{j}{t})}{2}-
    \left\lfloor\frac{\BalancingSpeed \cdot (\NodeLoadT{i}{t}-\NodeLoadT{j}{t})}{2}\right\rfloor$.
\item Then, node $i$ sends $L_{i,j}$ load items to node $j$ where
\vspace{-1ex}
\[L_{i,j} \coloneqq \begin{cases}
    \left\lceil\frac{\BalancingSpeed \cdot(\NodeLoadT{i}{t}-\NodeLoadT{j}{t})}{2}\right\rceil,  & \text{with probability } p, \\[5pt]
    \left\lfloor\frac{\BalancingSpeed \cdot (\NodeLoadT{i}{t}-\NodeLoadT{j}{t})}{2}\right\rfloor,  & \text{with probability } 1-p.
\end{cases}\]
\end{enumerate}
\end{minipage}}

\medskip

\noindent In the idealized setting, where the load is continuously divisible, a load of 
${\BalancingSpeed (\NodeLoadT{i}{t}-\NodeLoadT{j}{t})}/{2}$ 
is sent from node $i$ to node $j$.

\subsection{Notation}

We are given an arbitrary graph $G=(V,E)$ with $n$ nodes. 
We mainly assume that $G$ is regular and write $d$ for the node degree. 
Recall that the process is modeled by a Markov chain $(\LoadVecT{t})_{t\in\N}$, where $\LoadVecT{t} = (\NodeLoadT{i}{t})_{i \in [n]} \in \R^n$ is the \emph{load vector} at the end of step $t$, and $\NodeLoadT{i}{t}$ is the load of node $i$ at time $t$.
We write $\NodeAllocT{i}{t}$ for the number of load items allocated to node $i$ in step $t$ and define $\AllocVecT{t} = (\NodeAllocT{i}{t})_{i \in [n]}$.
We will use upper case letters such as $\NodeLoadT{i}{t}$ and $\MixMatT{t}$ to denote random variables and random matrices and lower case letters (like $x_i(t)$, $\mixMatT{t}$) for fixed outcomes. 
If clear from the context we will omit $t$ from a random variable.

We model the idealized balancing step in round $t$ by multiplication with a matrix $\MixMatBeta{\BalancingSpeed}(t) \in \R^{n \times n}$ given by
\[\MixMatBeta{\BalancingSpeed}_{i,j}(t) \coloneqq \begin{cases}
    1,\quad&\textup{if $i=j$ and $i$ is not matched at time $t$,} \\
    1-\BalancingSpeed/2,\quad&\textup{if $i=j$ and $i$ is matched at time $t$,} \\
    \BalancingSpeed/2,\quad&\textup{if $i$ and $j$ are matched at time $t$,} \\
    0,\quad&\textup{otherwise.}
\end{cases}\]
We will omit the parameter $\BalancingSpeed$ if it is clear from context.
With slight abuse of notation we use the same symbol $\MixMatT{t}$ for the matching itself and the associated balancing matrix
    and refer to both as just ``matchings''. Furthermore, we write $E(\MixMatT{t})$ for their edges.
For the product of all matching matrices from time $t_1$ to time $t_2$ we write
    \[\MixMatTT{t_1}{t_2} \coloneqq \MixMatT{t_2} \cdot \MixMatT{t_2 - 1} \cdot \cdots \cdot \MixMatT{t_1 + 1} \cdot \MixMatT{t_1},\]
    where for $t_1 > t_2$ we consider this to be the identity matrix. 
We generally refer to these matrices as \emph{mixing matrices}. Moreover, we write $\MixMatSeq{t}$ for the sequence of matching matrices~$(\MixMatT{\tau})_{\tau \in [t]}$ and analogously $\mixMatSeq{t}$ for a fixed sequence of matching matrices ~$(\mixMatT{\tau})_{\tau \in [t]}$.
We will write $\MixMat_{k,\cdot}$ for the vector forming the $k$th row of the matrix $\MixMat$ (which we often treat as a column vector despite it being a row).

In the balancing circuit model we define the \emph{round matrix} $\RoundMat\coloneqq \mixMatTT{1}{\CircuitSize}$ as the product of the matching matrices forming a complete period of the balancing circuit.
Note that $\CircuitSize$ has no relation to the minimum or maximum degree,
    although we may assume w.l.o.g.\ that each edge is covered by at least one of the matchings.
We write $\SpectralGap(\RoundMat)$ for the spectral gap of the round matrix $\RoundMat$,
    i.e., for the difference between the largest two eigenvalues of $\RoundMat$.

We write $\RoundingErrVecT{t} \in \R^n$ for the vector of additive rounding errors in round $t$.
Then $\RoundingErrT{k}{t}$ is the difference between the load at node $k$ after step $t$ and the load at node $k$ after step $t$ in an idealized scheme where loads are arbitrarily divisible.

Putting all of this together we can express the load vector at the end of step $t \in \N^+$ as
\begin{equation}\label{eq:load_recurrence}
    \LoadVecT{t} = \MixMatT{t} \cdot\left(\LoadVecT{t-1} + \AllocVecT{t}\right) + \RoundingErrVecT{t}.
\end{equation}

We write $\HittingTime$ for the \emph{hitting time} of $G$,
    which is the maximum expected time it takes for a standard random walk on $G$ (i.e., the walk moves to a neighbor chosen uniformly at random in each step)
        to reach a given node $i$ from a given node $j$,
        with the maximum taken over all such pairs of nodes.
We write $\EdgeHittingTime$ for the \emph{edge hitting time} of $G$,
    which is defined like the hitting time,
        except that the maximum is taken over adjacent nodes only.
We write $\Laplacian(G)$ for the normalized Laplacian matrix of a graph~$G$.
For regular graphs it may be defined as
    $\Laplacian(G) \coloneqq \IdentityMat - \AdjacencyMat(G) /d$,
    where $\AdjacencyMat(G)$ is the adjacency matrix of $G$.
Writing $\lambda_0 \leq \lambda_1 \leq \ldots \leq \lambda_{n-1}$ for the real eigenvalues of $\Laplacian(G)$,
    we let $\SpectralGap(\Laplacian(G)) \coloneqq \lambda_1 - \lambda_0$ be the spectral gap of the Laplacian of $G$.

\section{Random Matching Model}
\label{sec:analysis:random:matching}

In this section we analyze the process $\SyncProc{\RMDistr(G)}{\BalancingSpeed}{m}$ for $d$-regular graphs $G$, where the matching distribution $\RMDistr(G)$ is generated by the algorithm given in~\cite{DBLP:journals/jcss/GhoshM96}.
Note that the result (as well as the results for the two other models) holds at any point of time $t$ if the system is initially empty.
Furthermore, we can show the same results in the idealized setting where load items can be divided into arbitrarily small pieces (see \cite{DBLP:conf/icalp/AlistarhNS20}). 
For more details we refer the reader to the paragraph directly after \cref{eq:disc_component_sum}.

\begin{theorem}
    \label{thm:main_sync_random}
Let $G$ be a $d$-regular graph and define
\(T(G) \coloneqq \min \Big\{\frac{\HittingTime }{n} \cdot \log(n),  \sqrt{\frac{d}{\SpectralGap(\Laplacian(G))}},\discretionary{}{}{} \frac{1}{ \SpectralGap(\Laplacian(G))} \Big\}\).
Let $\LoadVecT{t}$ be the state of process  $\SyncProc{\RMDistr(G)}{\BalancingSpeed}{m}$
at time $t$ with $\discr(\LoadVecT{0}) \eqqcolon K \ge 1$.
There exists a constant $c>0$ such that for all
$t \geq c \cdot
\log(K\cdot n)\discretionary{}{}{} /({\SpectralGap(\Laplacian(G)) \cdot \beta})$
it holds w.h.p.\footnote{The expression \emph{with high probability (w.h.p.)} denotes a probability of at least $1-n^{-\Omega(1)}$.} and in expectation
\[ {\discr(\LoadVecT{t}) = \Oh\left(\log(n) \cdot \left(1 + \sqrt{\frac{m}{n} \cdot \frac{\EdgeHittingTime}{n}}\right) + \sqrt{\frac{\log(n)}{\BalancingSpeed} \cdot \frac{m}{n} \cdot T(G)}\right). } \]
\end{theorem}

\begin{proof}
We first expand the recurrence of \cref{eq:load_recurrence} (cf.~\cite{DBLP:conf/focs/RabaniSW98}).
After one step we get
\begin{align*}
    \LoadVecT{t} &= \MixMatT{t} \cdot\left(\LoadVecT{t-1} + \AllocVecT{t}\right) + \RoundingErrVecT{t} 
    \\&= \MixMatT{t} \cdot\Big(\underbrace{\left(\MixMatT{t-1} \cdot\left(\LoadVecT{t-2} + \AllocVecT{t-1}\right) + \RoundingErrVecT{t-1}\right)}_{\LoadVecT{t-1}} + \AllocVecT{t}\Big) + \RoundingErrVecT{t} 
    \\& = 
    \MixMatTT{t-1}{t}\cdot \LoadVecT{t-2} + \sum_{\tau=t-1}^{t} \MixMatTT{\tau}{t}\cdot \AllocVecT{\tau} + \sum_{\tau=t-1}^{t} \MixMatTT{\tau+1}{t} \cdot \RoundingErrVecT{\tau}
\end{align*}
We repeatedly expand this form up to the beginning of the process and get
    \begin{equation}\label{eq:load:vector:equality}
        \LoadVecT{t}
        = \underbrace{ \vphantom{\sum_{\tau=1}^{t}} \MixMatTT{1}{t} \cdot \LoadVecT{0}}_{\InitialContribVecT{t}}
        + \underbrace{\sum_{\tau=1}^{t} \MixMatTT{\tau}{t}
                                    \cdot \AllocVecT{\tau}}_{\DynamicContribVecT{t}}
        + \underbrace{\sum_{\tau=1}^{t} \MixMatTT{\tau+1}{t} \cdot \RoundingErrVecT{\tau}}_{\RoundingContribVecT{t}}.
    \end{equation}
We write $\InitialContribVecT{t}$, $\DynamicContribVecT{t}$, and $\RoundingContribVecT{t}$ for the three terms as indicated.
Note that in general these terms are vectors of real numbers.
The sum $\InitialContribVecT{t} + \DynamicContribVecT{t}$ can be regarded as the contribution of an idealized process, where $\InitialContribVecT{t}$ is the contribution of the initial load and $\DynamicContribVecT{t}$ is the contribution of the dynamically allocated load.
Thus, $\RoundingContribVecT{t}$ is the deviation between the idealized process without rounding and the discrete process described in \cref{sec:model}.

To bound the discrepancy $\discr(\LoadVecT{t})$ of the load vector $\LoadVecT{t}$ at time $t$ we use the fact that the discrepancy is sub-additive such that  $\disc(\vec{x} + \vec{y}) \leq \disc(\vec{x}) + \disc(\vec{y})$ (see \cref{obs:disc_subadditive} in \cref{apx:omitted-proofs-3}).
Hence, to bound $\discr(\LoadVecT{t})$ we individually bound the discrepancies of the three terms in \cref{eq:load:vector:equality} 
    and get
    \begin{equation}\label{eq:disc_component_sum}
    \discr(\LoadVecT{t}) \leq \discr(\InitialContribVecT{t}) + \discr(\DynamicContribVecT{t}) + \discr(\RoundingContribVecT{t}) .
    \end{equation}
If the system is initially empty, then $\disc(\InitialContribVecT{t}) = 0$.
Moreover, in the idealized setting without rounding $\disc(\RoundingContribVecT{t}) = 0$.
Techniques to bound the first term $\discr(\InitialContribVecT{t})$ and the last term $\discr(\RoundingContribVecT{t})$ are well-established.
We state the corresponding results in  \cref{lem:initial:load:vanishes} and \cref{lem:rounding:errors:are:small} directly below the proof of our theorem. The main part of the proof is to bound 
$\discr(\DynamicContribVecT{t})$, which will be done in \cref{sec:bound:contribution:dynamcally:allocated:balls}.

Let now $\gamma > 1$. 
First, it follows from 
\cref{lem:initial:load:vanishes}
that for all $t \geq c \cdot
\log(K\cdot n) /({\SpectralGap(\Laplacian(G)) \cdot \beta})$ 
we have
$\discr(\InitialContribVecT{t}) \leq 1$
with probability at least $1-n^{-\gamma}$.
Second, it follows from
\cref{lem:disc:dyn}
that
$\discr(\RoundingContribVecT{t}) \leq 2 \sqrt{\gamma\log(n)/\BalancingSpeed}$ with probability at least $1-3\cdot n^{-\gamma +1}$.
Third, it follows from
\cref{lem:rounding:errors:are:small}
that
\[\discr(\DynamicContribVecT{t}) = \Oh\left(
            \gamma\log(n) \cdot \left(1 +
            \sqrt{\frac{m}{n} \cdot \frac{\EdgeHittingTime}{n}} \right) +
            \sqrt{\frac{\gamma \log(n)}{\BalancingSpeed} \cdot\frac{m}{n} \cdot T(G)}\right)\]
        with probability at least $1-2\cdot n^{-\gamma+1}$.
The statement of the theorem therefore follows from a union bound over the statements of \cref{lem:initial:load:vanishes}, \cref{lem:rounding:errors:are:small}, and \cref{lem:disc:dyn}.
The bound on expectation follows analogously from the linearity of expectation and the bounds on the expected discrepancies in the aforementioned lemmas.
\end{proof}

Intuitively, \cref{lem:initial:load:vanishes} states that the contribution of the initial load to the discrepancy is insignificant if $t$ is large enough.
We generalize the analysis of Theorem 1~\cite{DBLP:conf/focs/RabaniSW98} (or Theorem 2.9 in~\cite{DBLP:conf/focs/SauerwaldS12}) to establish a bound on the discrepancy of the initial load as a function of $\BalancingSpeed$. For the sake of completeness the proof of  \cref{lem:initial:load:vanishes} is given in \cref{proof:lem:initial:load:vanishes}.

\begin{lemma}[name=Memorylessness Property,restate=restateInitialLoadVanishes,label=lem:initial:load:vanishes] 
    Let $G$ be a $d$-regular graph. Let $K=\discr(\LoadVecT{0})$.
    Then there exists a constant $c>0$ 
    such that for all $\gamma > 0$
    and $t \in \N$ with $t \geq t_0(\gamma) \coloneqq c \cdot
        \max\left\{\gamma \log(n), \log(K\cdot n)\right\} \cdot \smash[b]{\frac{1}{\SpectralGap(\Laplacian(G)) \cdot \beta}}$ we get with probability at least $1 - n^{-\gamma}$ and in expectation
     \[ \discr(\InitialContribVecT{t}) \leq 1.\]
\end{lemma}

The next lemma bounds $\discr(\RoundingContribVecT{t})$, the discrepancy contribution of cumulative rounding errors.
Note that this result does not just hold for the random matching model, but for all the three models that we consider in this paper.
In the proof of the lemma we extend then results of Theorem 3.6 in~\cite{DBLP:conf/focs/SauerwaldS12} (which is based on work in~\cite{DBLP:journals/jcss/BerenbrinkCFFS15}) to establish a bound as a function of $\BalancingSpeed$. The proof is given in \cref{p:lem:rounding:errors:are:small}.
\begin{lemma}[name=Insignificance of Rounding Errors,restate=restateRoundingErrorsAreSmall,label=lem:rounding:errors:are:small]
    Let $G$ be an arbitrary graph.
    Then for all $\gamma >1$, $t \in \N$, and $k \in [n]$ we get with probability at least $1 - 2n^{-\gamma+1}$ and in expectation
\[ \discr(\RoundingContribVecT{t}) \leq 2\cdot\sqrt{{\gamma \log(n)}/{\beta}}. \]
\end{lemma}

To bound $\discr(\DynamicContribVecT{t})$, the discrepancy contribution of dynamically allocated load items we apply  the next lemma. It is in fact the core of our work. We prove it in \cref{sec:bound:contribution:dynamcally:allocated:balls}. 
\begin{lemma}[Contribution of Dynamically Allocated Load]\label{lem:disc:dyn}
Let $G$ be a $d$-regular graph.
Define $T(G) \coloneqq \min \left\{\HittingTime\cdot \log n/{n}, \sqrt{d/{\SpectralGap(\Laplacian(G))}}, 1/{\SpectralGap(\Laplacian(G))} \right\}$.
Then for all $\gamma > 1$ and $t \in \N$ we get with probability at least $1 - 3n^{-\gamma+1}$ and in expectation
\[
\discr(\DynamicContribVecT{t}) = \Oh\left(
\gamma\log(n) \cdot \left(1 +
\sqrt{\frac{m}{n} \cdot \frac{\EdgeHittingTime}{n}} \right) +
\sqrt{\frac{\gamma \log(n)}{\BalancingSpeed} \cdot\frac{m}{n} \cdot T(G)}\right).
\]
\end{lemma}

\subsection{Bounding the Contribution of Dynamically Allocated Load} \label{sec:bound:contribution:dynamcally:allocated:balls}

In this section we prove \cref{lem:disc:dyn}.
Some of the proofs are omitted and can be found in \cref{sec:Omitted:Proofs:31}.
As a first step, we bound $\discr(\DynamicContribVecT{t})$
using the \emph{global divergence} $\GlobalDivergence(\MixMatSeq{t})$, which is defined over a sequence of matching matrices $\MixMatSeq{t}$
        as
    \[\GlobalDivergence(\MixMatSeq{t}) \coloneqq \max_{k \in [n]} \GlobalDivergence_k(\MixMatSeq{t}),\quad\textup{where}\quad
    \GlobalDivergence_k(\MixMatSeq{t})
        \coloneqq \sqrt{\sum_{\tau=1}^t \norm*{\MixMatTT{\tau}{t}_{k,\cdot}- \frac{\vec{1}}{n}}_2^2}.
    \]
    The global divergence can be regarded as a measure of the convergence speed of a random walk
    that uses the matching matrices as transition probabilities. 
In~\cite{DBLP:conf/stoc/FriedrichS09,DBLP:conf/focs/SauerwaldS12,DBLP:journals/jcss/BerenbrinkCFFS15} the authors use a related notion which they call the \emph{local $p$-divergence}, also defined on a sequence of matchings $\mixMatSeq{t}$.
The difference lies in the fact that the global divergence, essentially, measures differences between nodes' values and a global average, while the local divergence measures differences between neighboring nodes.
To show \cref{lem:disc:dyn} we first observe the following. 
\begin{observation}\label{obs:disc:in_terms_of_one_viation}
It holds that $\discr(\DynamicContribVecT{t}) \leq 2 \cdot \max_{k\in[n]}\abs{\NodeDynamicContribT{k}{t} - t\cdot m/n} $.
\end{observation}

Next we consider a fixed node $k$ and show a concentration inequality on $\NodeDynamicContribT{k}{t}$ in terms of $\GlobalDivergence_k(\mixMatSeq{t})$, where $\mixMatSeq{t}$ is the sequence of matchings applied by our process (\cref{lem:mixing_well_means_balancing_well}).
Note that in the lemma we assume the matchings are fixed and the randomness is due to the random load placement only. Hence, the lemma directly applies to $\BCDistr(G)$.
Afterwards, we bound the global divergence of the random sequence of matchings, $\GlobalDivergence_k(\MixMatSeq{t})$ in terms of a notion of ``goodness'' of the used matching distribution $\MatchDistr$,
for the random sequence of matchings (\cref{lem:glob:div:bound:drift}),
and then bound the ``goodness'' of the distribution $\RMDistr(G)$ used in the random matching model (\cref{lem:rmdistr_is_good}).
We start with a bound on the deviation of $\NodeDynamicContribT{k}{t}$ from the average load $t \cdot m/n$ in terms of $\GlobalDivergence(\mixMatSeq{t})$.

\begin{lemma}[Load Concentration]\label{lem:mixing_well_means_balancing_well}
    Let $\mixMatSeq{t}$ be an arbitrary sequence of matchings.
    Then for all $\gamma>0$, $t \in \N$, and $k \in [n]$ we get with probability at most $2 \cdot n^{-\gamma}$
   \[  
        \abs*{\NodeDynamicContribT{k}{t} - t \cdot \frac{m}{n}} \geq \frac{4}{3} \cdot \gamma\log(n) + \sqrt{8\gamma\log(n) \cdot \frac{m}{n}} \cdot \GlobalDivergence_k(\mixMatSeq{t}).
        \]
   \end{lemma}

\begin{proof}
Our goal is to decompose $\NodeDynamicContribT{k}{t}$ into a sum of independent random variables. Recall that we assume that the matching matrices are fixed and all randomness is due to the random choices of the load items. 
This will enable us to apply a concentration inequality to this sum.
    For the decomposition observe that
\(\DynamicContribVecT{t} = \sum_{\tau=1}^t \mixMatTT{\tau}{t} \cdot \AllocVecT{\tau},\)
where $\AllocVecT{\tau}$ is the random load vector corresponding to the $m$ load items allocated at time $\tau$.
So the $k$th coordinate of $\DynamicContribVecT{t}$ is
\(
\NodeDynamicContribT{k}{t}
=\sum_{\tau=1}^t\sum_{w\in[n]} \mixMatTT{\tau}{t}_{k,w}\cdot \NodeAllocT{w}{\tau}.
\)
We define the indicator random variable $\BallsRoundNode{\tau}{j}{w}$ for $\tau\in[t], j\in[m]$ and $w\in[n]$ as 
\[\BallsRoundNode{\tau}{j}{w} \coloneqq
    \begin{cases}
       1,  & \text{if the $j$-th load item of step $\tau$ is allocated to node $w$, } \\
       0,  &  \mbox{otherwise.}
    \end{cases}
\]
Note that for fixed $\tau$ and $j$ we have $\sum_{w\in [n]} \BallsRoundNode{\tau}{j}{w} =1$, $\Pr\left[\BallsRoundNode{\tau}{j}{w}=1\right]=1/n$ and $\E[\BallsRoundNode{\tau}{j}{w}]=1/n$.
Observe that $\NodeAllocT{w}{\tau}$, the load allocated to node $w$ at step $\tau$, can be expressed as $\sum_{j\in [m]} \BallsRoundNode{\tau}{j}{w}$. 
Merging this with the value of $\NodeDynamicContribT{k}{t}$ gives
\begin{align*}
    \NodeDynamicContribT{k}{t} &= 
    \sum_{\tau=1}^{t}\sum_{w\in [n]} \mixMatTT{\tau}{t}_{k,w} \cdot \left(\sum_{j\in [m]} \BallsRoundNode{\tau}{j}{w} \right)
    =\sum_{\tau=1}^t \sum_{j\in[m] } \underbrace{\left(\sum_{w\in [n]} \left(  \mixMatTT{\tau}{t}_{k,w}\cdot \BallsRoundNode{\tau}{j}{w}\right)\right)}_{\eqqcolon\BallRoundContr{k}{\tau}{j}}.
\end{align*}

For a fixed $\tau\in [t]$ and $j\in[m]$ we define 
$\BallRoundContr{k}{\tau}{j}\coloneqq\sum_{w\in [n]} \mixMatTT{\tau}{t}_{k,w}\cdot \BallsRoundNode{\tau}{j}{w}$.
This random variable measures the contribution of $j$-th load item of round $\tau$ to $\NodeDynamicContribT{k}{t}$. 
Note that the load items are allocated independently from each other. 
Since $\mixMatTT{\tau}{t}$ are fixed matrices, then $\BallRoundContr{k}{\tau}{j}$ and $\BallRoundContr{k}{\tau'}{j'}$ are independent for  all $\tau$ and $\tau'$ and $j\neq j'$. 
To apply the concentration inequality from \cref{thm:upper:bound:on:sum} we need to show that $\BallRoundContr{k}{\tau}{j}\le 1$ and compute an upper bound on $\Var[\BallRoundContr{k}{\tau}{j}]$. Showing the first condition is easy since 
exactly one of the indicator random variables $\BallsRoundNode{\tau}{j}{w}$ is one and $\mixMatTT{\tau}{t}_{k,w}$ has a value between zero and one. 

It remains to consider the variance of $\BallRoundContr{k}{\tau}{j}$. First note that by linearity of expectation
\begin{align*}
\BigAutoExp{\BallRoundContr{k}{\tau}{j}} =\BigAutoExp{\!\sum_{w\in [n]\!\!\!} \left(  \mixMatTT{\tau}{t}_{k,w}\cdot \BallsRoundNode{\tau}{j}{w}\right)} \!\!=\!
\sum_{\!\!\!w\in[n]\!\!\!}\mixMatTT{\tau}{t}_{k,w}\cdot \BigAutoExp{\BallsRoundNode{\tau}{j}{w}}\!=\! \sum_{\!\!\!w\in[n]\!\!\!}\mixMatTT{\tau}{t}_{k,w}\cdot \frac{1}{n} \!=\! \frac{1}{n},
\end{align*}
where the last equality follows form the fact that $\mixMatTT{\tau}{k}$ is doubly stochastic.
Now we get
\begin{align*}
    \Var[\BallRoundContr{k}{\tau}{j}]
    &= \BigAutoExp{\left(\BallRoundContr{k}{\tau}{j} - \AutoExp{\BallRoundContr{k}{\tau}{j}}\right)^2}
= \BigAutoExp{\Big(\Big(\sum_{w\in [n]} \mixMatTT{\tau}{t}_{k,w}\cdot \BallsRoundNode{\tau}{j}{w}\Big) - \frac{1}{n}\Big)^2}
\\ &    =
    \sum_{w' \in [n]} \frac{1}{n} \cdot \left(\mixMatTT{\tau}{t}_{k,w'} - \frac{1}{n}\right)^2
= \frac{1}{n} \cdot \norm*{\mixMatTT{\tau}{t}_{k,\cdot} - \frac{\vec{1}}{n}}_2^2,
\end{align*}
where we used that for each $\tau$ and each $j$ exactly one of the $\BallsRoundNode{\tau}{j}{w}$ is one and all others are zero, and each of the $n$ possible cases has uniform probability.

Recall that $\BallRoundContr{k}{\tau}{j}$ and $\BallRoundContr{k}{\tau'}{j'}$ are independent for all $\tau, \tau'$ and $j\neq j'$. Hence we get
\begin{align*}
\BigAutoVar{\sum_{\tau=1}^{t}\sum_{j\in[m]} \BallRoundContr{k}{\tau}{j}}
   &= \sum_{\tau=1}^{t}\sum_{j\in[m]} \Var[\BallRoundContr{k}{\tau}{j}]
    = \frac{1}{n} \cdot \sum_{\tau=1}^t\sum_{j\in[m]} \norm*{\mixMatTT{\tau}{t}_{k,\cdot} - \frac{\vec{1}}{n}}_2^2
\\&= \frac{m}{n}\cdot \left(\GlobalDivergence_k(\mixMatSeq{t})\right)^2,
\end{align*}
    where the final equality uses the definition of the global divergence  $\GlobalDivergence_k(\mixMatSeq{t})$.
Applying \cref{thm:upper:bound:on:sum} with $M=1$ and $X=\NodeDynamicContribT{k}{t}=\sum_{\tau=1}^{t}\sum_{j\in[m]} \BallRoundContr{k}{\tau}{j}$ with $\lambda=2\gamma\log(n)/3 + \GlobalDivergence_k(\mixMatSeq{t})\cdot \sqrt{2\gamma m/n}$ results in 
\[\Pr\left[{ \NodeDynamicContribT{k}{t}- t\cdot \frac{m}{n}} \geq \frac{2}{3} \cdot \gamma\log(n) + \sqrt{2\gamma\log(n) \cdot \frac{m}{n}} \cdot \GlobalDivergence_k(\mixMatSeq{t})\right]\le n^{-\gamma}.
\]
The lower bound can be established using \cref{thm:lower:on:sum} (with $a_i=0$ and $M=1$) instead of \cref{thm:upper:bound:on:sum}. Via a union bound we get
\[
 \BigAutoProb{\abs*{\NodeDynamicContribT{k}{t} - t \cdot \frac{m}{n}} \geq \frac{4}{3} \cdot \gamma\log(n) + \sqrt{8\gamma\log(n) \cdot \frac{m}{n}} \cdot \GlobalDivergence_k(\mixMatSeq{t})}
            \leq 2 \cdot n^{-\gamma}. \qedhere \]
\end{proof}

To bound the global divergence of the matching sequence used by the process we use two potential functions. 
The \emph{quadratic node potential} $\NodePotential(\vec{x})$ is given by
\[\NodePotential(\vec{x}) \coloneqq \sum_{i \in [n]} \left(x_i - \overline{x}\right)^2,\quad \text{where} \quad \overline{x} \coloneqq \frac{1}{n} \cdot \sum_{j \in [n]} x_j.\]
 For a set of edges $S$ on the nodes $[n]$ and a vector $\vec{x} \in \R^n$, the \emph{quadratic edge potential} is
    \[\EdgePotential_S(\vec{x}) \coloneqq \sum_{\{i, j\} \in S} (x_i - x_j)^2.\]
We may also write $\EdgePotential_G \coloneqq \EdgePotential_{E(G)}$ whenever $G$ is a graph, and $\EdgePotential_\MixMat \coloneqq \EdgePotential_{E(\MixMat)}$ whenever $\MixMat$ is a matching matrix.
The following observation relates the drop of node potential to the edge potential in terms of $\BalancingSpeed$.
\begin{observation}[name=,label=obs:node_potential_change_exact,restate=restateObsPotentialRelation]
    Let $\MixMatBeta{\BalancingSpeed}$ be a matching matrix with parameter $\BalancingSpeed \in (0, 1]$.
    Then for any $\vec{x} \in \R^n$ we have
    $\NodePotential(\vec{x}) - \NodePotential(\MixMatBeta{\BalancingSpeed} \cdot \vec{x})
            = \frac{1 - (1-\BalancingSpeed)^2}{2} \cdot \EdgePotential_{E(\MixMatBeta{\BalancingSpeed})}(\vec{x})$.
\end{observation}

We now define a notion of a matching distribution being \emph{good}.
In \cref{lem:glob:div:bound:drift} below we show that the notion is sufficient for showing that matching sequences generated from such distributions have bounded global divergence.
Note that the ``goodness'' of a distribution  does not depend on $\BalancingSpeed$ but on graph properties and the random choices with which the matchings are chosen.
Hence, we assume $\BalancingSpeed=1$.

\begin{definition}\label{def:goodness}
Assume $G$ is an arbitrary $d$-regular graph. 
    Let $g\colon \R_0^+ \to \R^+$ be an increasing function and let $\sigma^2 > 1$.
    Then a matching distribution $\MatchDistr(G)$ is \emph{$(g,\sigma^2)$-good} if the following conditions hold for $\MixMatBeta{1} \sim \MatchDistr(G)$ and all stochastic vectors $\vec{x} \in \R^n$.
\begin{enumerate}
    \item \(\NodePotential(\vec{x}) -  \AutoExp{\NodePotential(\MixMatBeta{1}  \cdot \vec{x})} \geq g(\NodePotential(\vec{x})).\)
    \item \(\AutoVar{\NodePotential(\MixMatBeta{1} \cdot \vec{x})} \leq (\sigma^2 - 1) \cdot \left(\NodePotential(\vec{x}) -  \AutoExp{\NodePotential(\MixMatBeta{1} \cdot \vec{x})}\right)^2.\)
\end{enumerate}
\end{definition}

It remains to show two results.
First, assuming a matching distribution is $(g,\sigma^2)$-good, the global divergence of a matching sequence  generated by that distribution can be bounded in terms of $g$ and $\sigma$ (\cref{lem:glob:div:bound:drift}). Second, we have to calculate a function $g_G$ and the values of $\sigma_G$ for which the matching distribution $\RMDistr(G)$ is $(g_G,\sigma_G^2)$-good (see \cref{lem:rmdistr_is_good}).

\begin{lemma}[name=Global Divergence,label=lem:glob:div:bound:drift,restate=restateLemGlobalDivergence]
    Assume $G$ is an arbitrary graph. Let $g\colon \R_0^+ \to \R^+$ be an increasing function, $\sigma^2 > 1$, and $\BalancingSpeed \in (0,1]$.
    Let $\MixMatSeq{t} = (\MixMatBeta{\BalancingSpeed}(\tau))_{\tau=1}^t$ be an i.i.d.\ sequence of matching matrices generated by $\MatchDistr(G)$ and 
    assume $\MatchDistr(G)$ is a $(g,\sigma^2)$-good matching distribution.
    Then for all  $\gamma > 0$ and $k \in [n]$ we get with probability at least $1 - n^{-\gamma}$
        \[
        \left(\GlobalDivergence_k(\MixMatSeq{t})\right)^2 \leq 8 \sigma^2 (\gamma \log(n) + \log(8 \sigma^2)) + \frac{2}{\BalancingSpeed} \cdot \int_0^1 \frac{x}{g(x)}\,\dx.
        \]
\end{lemma}

\begin{lemma}
\label{lem:rmdistr_is_good}
    Assume $G$ is an arbitrary $d$-regular graph. Let
        \[g_G(x) \coloneqq \frac{1}{16 d} \cdot \max\left\{
            d \cdot \SpectralGap(\Laplacian(G)) \cdot x,
            \frac{x^2}{\ResistiveDiameter} ,
            \frac{4}{27} \cdot x^3\right\} \text{ and }
        \sigma_G^2 = 32 \cdot (\EdgeHittingTime / n) + 5.\]
    Then $\RMDistr(G)$ is $(g_G, \sigma_G^2)$-good.    
\end{lemma}

\begin{proof}
First, note that the function $g_G(x)$ is increasing in $x$. 
Applying the first part of \cref{prop:node_potential_change_statistics} (see below) we get that for any vector $\Vec{x}\in \R^n$ it holds that 
\[\NodePotential(\vec{x}) - \BigAutoExp{\NodePotential(\MixMatBeta{1} \cdot \vec{x})} 
\ge \frac{1}{16d} \cdot \EdgePotential_G(\vec{x}).\]
From the first two statements of \cref{lem:edge_potential_bounds}
(stated behind \cref{lem:edge_potential_bounds}) we see that for $\MixMatBeta{1} \sim \RMDistr(G)$ and all stochastic vectors $\vec{x} \in \R^n$
\[
\EdgePotential_G(\vec{x}) \geq \max\left\{d\cdot \SpectralGap(\Laplacian(G))\cdot \NodePotential(\vec{x}), \frac{\NodePotential(\vec{x})^2}{\ResistiveDiameter} , \frac{4}{27} \cdot \NodePotential(\vec{x})^3 \right\}.
\]
Hence,
\[\NodePotential(\vec{x}) - \BigAutoExp{\NodePotential(\MixMatBeta{1} \cdot \vec{x})} 
\ge \frac{1}{16d} \cdot \max\left\{d \cdot \SpectralGap(\Laplacian(G)) \cdot \NodePotential(\vec{x}),
\frac{\NodePotential(\vec{x})^2}{\ResistiveDiameter},
\frac{4}{27} \cdot \NodePotential(\vec{x})^3\right\},\]
and as a consequence, $\NodePotential(\vec{x}) -  \AutoExp{\NodePotential(\MixMatBeta{1}  \cdot \vec{x})} \geq g_G(\NodePotential(\vec{x}))$ by the definition of $g_G$.

It remains  to check the second condition of \cref{def:goodness} with our claimed value~$\sigma_G^2$.
Inserting its value as stated in the lemma, the condition requires that
\[\AutoVar{\NodePotential(\MixMatBeta{1} \cdot \vec{x})} \leq (32 (\EdgeHittingTime / n) + 5 - 1) \cdot \left(\NodePotential(\vec{x}) -  \AutoExp{\NodePotential(\MixMatBeta{1} \cdot \vec{x})}\right)^2,\]
which is given in the second part of \cref{prop:node_potential_change_statistics} (see below).
\end{proof}

In \cref{prop:node_potential_change_statistics} we first relate the drop of $\NodePotential$ to the quadratic edge potential $\EdgePotential$. In the second part we bound the variance of the potential drop as a function of the edge hitting time.

\begin{lemma}[label=prop:node_potential_change_statistics,restate=restateLemNodePotentialChangeStatistics]
Let $G$ be a $d$-regular graph, let $\MixMat^1 \sim \RMDistr(G)$, and let $\vec{x} \in \R^n$,
then
\begin{enumerate}
\item \(\NodePotential(\vec{x}) - \BigAutoExp{\NodePotential(\MixMatBeta{1} \cdot \vec{x})} 
            \ge \frac{1}{16d} \cdot \EdgePotential_G(\vec{x}).\)
\item \(\BigAutoVar{\NodePotential(\MixMatBeta{1} \cdot \vec{x})}
            \leq (32 \cdot (\EdgeHittingTime / n) + 4) \cdot\left(\NodePotential(\vec{x}) - \BigAutoExp{\NodePotential(\MixMatBeta{1} \cdot \vec{x})}\right)^2.\)
\end{enumerate}
\end{lemma}

In \cref{lem:edge_potential_bounds} we relate the size of the quadratic edge potential $\EdgePotential_G$ to 
the second-largest eigenvalue of $\Laplacian(G)$, the effective resistance of $G$ and node potential. 
To state it,
    we need some additional definitions.
For any two nodes $i$ and $j$ of the graph $G$
    $\ResistiveDistance{i}{j}$ is the \emph{effective resistance} (or \emph{resistive distance}) between $i$ and $j$ in $G$  (for a detailed definition see \cref{apx:aux}).
Furthermore, we write $\ResistiveDiameter$ for the \emph{resistive diameter} of $G$,
    i.e., the largest resistive distance between any pair of nodes in $G$,
    and write $\MaxEdgeResistance$ for the maximum effective resistance between any pair of nodes adjacent in $G$.
I.e., $\ResistiveDiameter \coloneqq \max_{i,j \in [n]} \ResistiveDistance{i}{j}$ and $\MaxEdgeResistance \coloneqq \max_{\{i,j\} \in E(G)} \ResistiveDistance{i}{j}$.
The first part of the following lemma was previously shown in~\cite{DBLP:journals/jcss/GhoshM96,DBLP:conf/focs/SauerwaldS12}.

\begin{lemma}
[label=lem:edge_potential_bounds,restate=restateEdgePotentialBounds]
    Let $\vec{x} \in \R^n$, and let $G$ be a connected $d$-regular graph.
    \begin{enumerate}
    \item \(\EdgePotential_G(\vec{x}) \geq d \cdot \SpectralGap(\Laplacian(G)) \cdot \NodePotential(\vec{x})\).
    \item If $\vec{x}$ is stochastic, then
        $\EdgePotential_G(\vec{x}) \geq \max\left\{\frac{1}{\ResistiveDiameter} \cdot \NodePotential(\vec{x})^2, \frac{4}{27} \cdot \NodePotential(\vec{x})^3\right\}$
    \item \(\max_{\{i,j\} \in E(G)} (x_i - x_j)^2 \leq \MaxEdgeResistance \cdot \EdgePotential_G(\vec{x}).\)
    \end{enumerate}
\end{lemma}

\subsubsection*{Proof of \cref{lem:disc:dyn}}
\begin{proof}
Define $g_G(x) = \frac{1}{16d} \cdot \max\left\{d \cdot \SpectralGap(\Laplacian(G)) \cdot x, x^2 / \ResistiveDiameter,    4 x^3 /27\right\}$ and let $\sigma_G^2 \coloneqq 32 \cdot (\EdgeHittingTime/n) + 5$.
Then by \cref{lem:rmdistr_is_good} the matching distribution $\RMDistr(G)$ is $(g_G, \sigma_G^2)$-good.
By \cref{lem:glob:div:bound:drift} we have for all $t \in \N$, $k \in [n]$
\[\BigAutoProb{\left(\GlobalDivergence_k(\MixMatSeq{t})\right)^2 \leq 8 \sigma_G^2 ((\gamma+1) \log(n) + \log(8 \sigma_G^2)) + \frac{1}{\BalancingSpeed} \cdot \int_0^1 \frac{x}{g_G(x)}\,\dx} \geq 1 - n^{-(\gamma+1)}.\]
To bound $\GlobalDivergence_k(\MixMatSeq{t})$ we use the following two claims (see \cref{apx:proof_discr_dyn} for the proof).
\begin{claim}\label{claim:integral_bound}
    It holds that $\displaystyle \int_0^1 {x}/{g_G(x)}\,\dx = \Oh(T(G))$.
\end{claim}
\begin{claim}\label{claim:edge_hitting_time_lower_bound}
For any $d$-regular graph $G$ it holds that $\EdgeHittingTime / n \geq 1/2$.
\end{claim}
Together we get from \cref{claim:integral_bound} and \cref{claim:edge_hitting_time_lower_bound} that with probability at least $1 - n^{-(\gamma+1)}$
\begin{equation}\label{eqn:concrete_global_div_bound}
\left(\GlobalDivergence_k(\MixMatSeq{t})\right)^2 = \Oh\left(\frac{\EdgeHittingTime}{n} \cdot \left(\gamma\log(n) + \log\left(\frac{\EdgeHittingTime}{n}\right)\right) + \frac{T(G)}{\BalancingSpeed}\right).
\end{equation}
Since $\EdgeHittingTime = \Oh(n^3)$ (Proposition 10.16 in \cite{LevinPeresbook}),
 $\log(\EdgeHittingTime / n) = \Oh(\log n)$, and $\gamma  > 1$,
\[
\GlobalDivergence_k(\MixMatSeq{t})
= \Oh\left(\sqrt{\gamma\log(n) \cdot \frac{\EdgeHittingTime}{n} + \frac{T(G)}{\BalancingSpeed}}\right) = \Oh\left(\sqrt{\gamma \log(n) \cdot \frac{\EdgeHittingTime}{n}} + \sqrt{\frac{T(G)}{\BalancingSpeed}}\right).
\]
Now \cref{lem:mixing_well_means_balancing_well} states that for any fixed sequence of matching matrices $\mixMatSeq{t}$, with probability at least $1 - 2n^{-(\gamma+1)}$ it holds that
\begin{equation}\label{eqn:concrete_dynamic_concentration_bound}
    \abs*{\NodeDynamicContribT{k}{t} - t \cdot \frac{m}{n}} = \Oh\left( \gamma\log(n) + \sqrt{\gamma\log(n) \cdot \frac{m}{n}} \cdot \GlobalDivergence_k(\mixMatSeq{t})\right).
\end{equation}
Applying a union bound over all $k \in [n]$, \cref{eqn:concrete_global_div_bound} and \cref{eqn:concrete_dynamic_concentration_bound} hold for all $k$ with probability at least $1 - 3n^{-\gamma}$. Hence, for all $k \in [n]$
\begin{equation*}\begin{aligned}
\abs*{\NodeDynamicContribT{k}{t} - t \cdot \frac{m}{n}}  
&=\Oh\left( \gamma\log(n) + \sqrt{\gamma\log(n) \cdot \frac{m}{n}} \cdot \left(\sqrt{\gamma \log(n) \cdot \frac{\EdgeHittingTime}{n}} + \sqrt{\frac{T(G)}{\BalancingSpeed}}\right)\right)
\\ 
&=\Oh\left(\gamma\log(n)\cdot \left(1 + \sqrt{\frac{m}{n} \cdot \frac{\EdgeHittingTime}{n}}\right) + \sqrt{\frac{(\gamma+1)\log(n)}{\BalancingSpeed} \cdot \frac{m}{n} \cdot T(G)}\right).
\end{aligned}\end{equation*}
The high-probability bound now follows from \cref{obs:disc:in_terms_of_one_viation}.
The corresponding bound on $\AutoExp{\discr(\DynamicContribVecT{t}}$ follows readily; see \cref{lem:tail_bound_to_expectation_bound} in \cref{apx:known-results-probability-theory} for the details.
\end{proof}

\section{Balancing Circuit Model}\label{sec:analysis_balancing_circuit}

Here we assume $\BalancingSpeed=1$.
Recall that we assume $G$ is covered by $\CircuitSize$ fixed matchings $\mixMatT{1},\ldots, \mixMatT{\CircuitSize}$.
The matching distribution $\BCDistr(G)$ then deterministically chooses the matching  $\mixMatT{t}=\mixMatT{t \bmod \CircuitSize}$ in step $t$.
The round matrix is defined as $\RoundMat \coloneqq \mixMatTT{1}{\CircuitSize}$ and the mixing matrices are fixed in this model. 
Thus, for a sequence of matchings $\mixMatSeq{t}$ the global divergence is 
$\GlobalDivergence(\mixMatSeq{t}) \coloneqq \max_{k\in [n]}\sqrt{\sum_{\tau=1}^t \norm*{\mixMatTT{\tau}{t}_{k,\cdot} - 1/n}_2^2}$.
The next theorem provides an upper bound on the discrepancy for this model.
Note that the following theorem holds for arbitrary graphs, while \cref{thm:main_sync_random} only holds for $d$-regular graphs.

\begin{theorem}\label{thm:main_sync_circuit}
Let $G$ be an arbitrary graph and $\LoadVecT{t}$ be the state of process $\SyncProc{\BCDistr(G)}{1}{m}$ at time $t$ with $\discr(\LoadVecT{0})\eqqcolon K$.
    For all $t \in \N$ with $t \ge  \frac{\CircuitSize}{\SpectralGap{(\RoundMat)}}\cdot\left(\ln(K\cdot n) \right)$
    it holds w.h.p.\ and in expectation
    \[\discr(\LoadVecT{t})=\Oh\left(\log (n)+ \sqrt{ m/n}\cdot\GlobalDivergence(\mixMatSeq{t})\cdot \sqrt{\log (n)}   \right).\]
\end{theorem}

\begin{proof}
The proof follows the same line as the proof \cref{thm:main_sync_random}, which is proved via \cref{lem:initial:load:vanishes}, 
\cref{lem:disc:dyn}, and \cref{lem:rounding:errors:are:small} bounding 
$\InitialContribVecT{t}, \DynamicContribVecT{t}$, and $\RoundingContribVecT{t}$, respectively.
\Cref{lem:initial:load:vanishes} is replaced by \cref{lem:erro_bound} below. \Cref{lem:initial:load:vanishes} can also be applied to the balancing circuit model since it only requires that the subgraph used for balancing is a matching. 

It remains to replace \cref{lem:rounding:errors:are:small}. 
Since the matching matrices are fixed this time the proof is much simpler. 
The proof of \cref{lem:mixing_well_means_balancing_well} carries to over to this model giving us a bound on $\abs{\NodeDynamicContribT{k}{t}-tm/n}$ for $k\in[n]$  with probability at least $1-2\cdot n^{-\gamma}$.
Applying the union bound over all nodes $k\in[n]$, together with  \cref{obs:disc:in_terms_of_one_viation} (stating that  $\discr(\DynamicContribVecT{t}) \leq 2 \cdot \max_{k\in[n]}\abs{\NodeDynamicContribT{k}{t} - t\cdot m/n} $), gives a bound on $\discr(\DynamicContribVecT{t})$ which holds with probability at least $1-2\cdot n^{\gamma+1}$.
\end{proof}

\begin{lemma}[Memorylessness Property] \label{lem:erro_bound}

For all $t\in \N$ with $t \ge {\CircuitSize}/{\SpectralGap{(\RoundMat)}}\cdot \left(\ln(K\cdot n)\right)$ it holds that $\discr(\InitialContribVecT{t})\le 2$.
    \end{lemma}
    \begin{proof}

Since $\NodePotential(\vec{x}) \le K^2\cdot n$ it follows from Lemma 2 in \cite{DBLP:conf/spaa/GhoshMS96} that
   \begin{equation*}
       \NodePotential\left(\mixMatTT{1}{t} \cdot \vec{x}\right) \leq (1-\SpectralGap{(\cMixMat)})^{2\lfloor t \rfloor/\CircuitSize}\cdot \NodePotential(\vec{x}) \le (1-\SpectralGap{(\cMixMat)})^{2\lfloor t \rfloor/\CircuitSize}\cdot K^2\cdot n \le e^{-2 \lfloor t \rfloor \cdot \SpectralGap{(\cMixMat)}/\CircuitSize+ 2\ln(Kn)}.
   \end{equation*}
Setting $t\ge (\CircuitSize/\SpectralGap{(\RoundMat)})\cdot \left(\ln(Kn)\right)$ gives
$\NodePotential\left(\mixMatTT{1}{ t} \cdot \vec{x}\right) \leq 1$ which implies that $\discr(\InitialContribVecT{t}\le 2$.
    \end{proof}
Note that a similar statement was shown in~\cite{DBLP:conf/focs/RabaniSW98,DBLP:conf/focs/SauerwaldS12,DBLP:journals/jcss/BerenbrinkCFFS15}.

The next theorem provides a lower bound on the discrepancy for this model. The proof can be found in \cref{apx:analysis_balancing_circuit}.
\begin{theorem}\label{thm:main_sync:lower}
Let $G$ be an arbitrary graph and $\LoadVecT{t}$ be the state of process $\SyncProc{\BCDistr(G)}{1}{m}$ at time $t$. 
Then for all  $t\in \N$  and $m\ge 4n\cdot \log(n)/ \GlobalDivergence(\mixMatSeq{t})$ it holds with constant probability
\[ \discr(\LoadVecT{t})=\Omega\left(\sqrt{m/n}\cdot \GlobalDivergence(\mixMatSeq{t})\right).\]
\end{theorem}

\section{Asynchronous Model}
\label{sec:asynchronous}
The following is our main theorem for the asynchronous model. 
The bounds provided by \cref{thm:main_async} for the asynchronous model differ from those in \cref{thm:main_sync_random} for the random matching model in two details.
First, the lower bound on the balancing time is larger by a factor of $n$. This is due to the fact that the asynchronous model balances across just one edge per round in contrast to $\Theta(n)$ edges in the random matching model.
Second, the upper bound on $\discr(\LoadVecT{t})$ is much simpler.
Note, however that setting $m=n$ in \cref{thm:main_sync_random} and further simplifying the result by using $\EdgeHittingTime / n = \Omega(1)$ (see also \cref{claim:edge_hitting_time_lower_bound} in the proof of \cref{lem:disc:dyn}) results in the same asymptotic bound as in \cref{thm:main_async}.

\begin{theorem}
\label{thm:main_async}
Let $G$ be a $d$-regular graph and 
define \((T(G) \coloneqq \min \Big\{\frac{\HittingTime}{n} \cdot \log(n), \sqrt{\frac{d}{\SpectralGap(\Laplacian(G))}},\discretionary{}{}{} \frac{1}{\SpectralGap(\Laplacian(G))} \Big\} \).
Let $\LoadVecT{t}$ be the state of process $\AsyncProc{\ADistr(G)}{\BalancingSpeed}$
at time $t$ with $\discr(\LoadVecT{0}) \eqqcolon K \ge 1$.
There exists a constant $c>0$ such that for all 
$t  \geq c \cdot n \cdot \log(K\cdot n) / (\SpectralGap(\Laplacian(G)) \cdot \BalancingSpeed)$
it holds w.h.p.\ and in expectation
\[ {\discr(\LoadVecT{t}) = \Oh\left(\log(n) \sqrt{\frac{\EdgeHittingTime}{n}} + \sqrt{\frac{\log(n)}{\BalancingSpeed} \cdot T(G)}\right).}\]
\end{theorem}

\begin{proof}[Proof Sketch of \cref{thm:main_async}]
The proof of the theorem follows along the same lines at the proof of \cref{thm:main_sync_random}. 
However, there are some major differences. 
Most importantly, the proof of \cref{lem:mixing_well_means_balancing_well} (giving a concentration bound on $\NodeDynamicContribT{k}{t}$ in terms of the global divergence of the sequence of matching matrices) can not be applied for $\AsyncProcName$. 
The proof heavily relies on the fact that the load allocation and the matching edges are chosen independently from each other, which is certainly not the case for  $\AsyncProcName$. 
Our new lemma (\cref{lem:mixing_well_means_balancing_well_async} in \cref{apx:asynchronous}) carefully analyses the dependency, and it uses a stronger concentration inequality.
In addition, we also have to re-calculate the function $g_G$ and $\sigma_G$ to show that the matching distribution used by $\ADistr$ is $(g_G, \sigma_G^2)$-good (see \cref{lem:asdistr_is_good} in \cref{apx:asynchronous}).
\end{proof}

\section{Drift Result} \label{sec:drift}

In our analysis we use the following tail bound for the sum of a non-increasing sequence of random variables with variable negative drift.
The proof uses established methods from drift analysis.
In particular, it relies one techniques found in the proof of the Variable Drift Theorem in \cite{DBLP:series/ncs/Lengler20}.
The full technical proof can be found in \cref{apx:drift_proof}.

\begin{theorem}[name=,restate=restateLemDrift,label=lem:drift]
    Let $(X(t))_{t\ge 0}$ be a non-increasing sequence of discrete random variables with $X(t)\in \R^+_0$ for all $t$ with fixed $X(0) = x_0$.
    Assume there exists an increasing function $h\colon \R^+_0 \to \R^+$ and a constant $\sigma > 0$ such that the following holds. For all $t \in \N$ and all $x > 0$ with $\AutoProb{X(t) = x} > 0$
    \begin{enumerate}
        \item \(\AutoExpCond{X(t+1)}{X(t) = x} \leq x - h(x),\)\label{cond:drift:1}
        \item \(\AutoVarCond{X(t+1)}{X(t) = x} \leq \sigma \cdot \left(\AutoExpCond{X(t+1)}{X(t) = x} - x\right)^2.\) \label{cond:drift:2}
    \end{enumerate}
   Then the following statements hold.
    \begin{enumerate}
        \item For all $\delta \in (0, 1)$ and any arbitrary but fixed $t$
        \[\BigAutoProb{\int_{X(t)}^{x_0} \frac{1}{h(\varphi)}\,\dvarphi \leq (1-\delta)t} \leq \exp\left(-\,\frac{\delta^2 t}{2(\sigma + 1)}\right).\]
        \item  For all $\delta \in (0, 1)$ and $p \in (0,1)$ we define $t_0 \coloneqq \frac{2(\sigma + 1)}{\delta^2} \left(-\log(p) + \log\left(\frac{2(\sigma + 1)}{\delta^2}\right)\right)$. Then \[\BigAutoProb{\sum_{t=t_0+1}^{\infty} X(t) \leq \frac{1}{1-\delta} \cdot \int_0^{x_0} \frac{\varphi}{h(\varphi)} \dvarphi} \geq 1 - p.\]
    \end{enumerate}
\end{theorem}

\section{Conclusions and Open Problems} \label{sec:conclusions}

In this paper we analyze discrete load balancing processes on graphs. 
As our main contribution we bound the discrepancy that arises in dynamic load balancing in three models, the random matching model, the balancing circuit model, and the asynchronous model.
Our results for the random matching model and the asynchronous model hold for $d$-regular graphs, while our analysis for the balancing circuit model applies to arbitrary graphs.

To the best of our knowledge our results constitute the first bounds for discrete, dynamic balancing processes on graphs.
Furthermore, our results improve the work by Alistarh et al.~\cite{DBLP:conf/icalp/AlistarhNS20} who prove that the expected discrepancy is bounded by $\sqrt{n}\log(n)$ in the (arguably simpler) continuous asynchronous process {\def\dcmhackI{^{\text{(cont)}}}$\AsyncProc{\ADistr(G)}{1}$}.
We improve their bound to $\sqrt{n \log(n)}$ and additionally show that it holds with high probability.
We conjecture that our results are tight up to polylogarithmic factors.
However, showing tight upper and lower bounds remains an open problem.

\paragraph{Results for Specific Graph Classes}

We show an overview of our bounds on the discrepancy for specific graph classes in \cref{table:disc:upperbound}. The corresponding results are formally derived in \cref{apx:hitting-time_spectral-gap} for the random matching model,  \cref{apx:bounds-specific-graphs-C} for the balancing circuit model, and  \cref{apx:async-bounds-graph-classes} for the asynchronous model.

\begin{table}[ht]
\caption {Asymptotic upper bounds on the discrepancy in specific graph classes.}
\label{table:disc:upperbound}
\def\hx#1{#1&}
\def\arraystretch{1.5}
\begin{tabularx}{\textwidth}{ Xccc }
\toprule
Graph & $\SyncProc{\RMDistr(G)}{1}{m}$ & $\SyncProc{\BCDistr(G)}{1}{m}$ &  $\AsyncProc{\ADistr(G)}{1}$\\[-1ex]
 & \cref{apx:hitting-time_spectral-gap} & \cref{apx:bounds-specific-graphs-C} & \cref{apx:async-bounds-graph-classes} \\

\midrule

\hx{$d$-regular graph\newline\small (const. $d$)}
$\log(n) + \sqrt{m\cdot\log(n)}$ 
&$\log(n) +\sqrt{m\cdot\log(n)}$
& $\sqrt{n\cdot\log(n)}$
\\

\hx{cycle $C_n$}
$\log(n) + \sqrt{m\cdot\log(n)}$ 
&$\log(n) +\sqrt{m\cdot\log(n)}$
& $\sqrt{n\cdot\log(n)}$
\\

\hx{2-D torus}
$\log(n) + \sqrt{m/n}\cdot \log^{3/2}(n)$ 
& $(1 + \sqrt{m/n})\cdot\log(n)$ 
& $\log^{3/2}(n)$\\

\hx{$r$-D torus\newline\small (const. $r \geq 3$)}
$(1 + \sqrt{m/n})\cdot\log(n)$ 
& $\log(n) + \sqrt{m/n\cdot\log(n)}$ 
& $\log(n)$ \\

\hx{hypercube}
$(1 + \sqrt{m/n})\cdot\log(n)$ 
& $(1 + \sqrt{m/n})\cdot\log(n)$ 
& $\log(n)$ \\

\bottomrule

\end{tabularx}
\end{table}

\paragraph{Open Problems}
We are confident that our results carry over to arbitrary graphs (as opposed to regular graphs), provided that there exists a lower bound on the probability $p_{min}$
with which an edge is used for balancing. 
However, to show bounds on the discrepancy one has to overcome fundamental problems such as the bias introduced by high-degree nodes. 
Another interesting open question is whether the results carry over to a model where the amount of load that may transmitted over an edge in each step is bounded by a constant. 
If only a single load item can be transferred per edge and step the problem is similar to the token distribution problem (see, for example, \cite{DBLP:journals/algorithmica/HeideOW96}).

Finally, we believe that one can also adapt our analysis to variant of a graphical balls-into-bins process.
The process works as follows.
In each step an edge $(i,j)$ is sampled uniformly at random. W.l.o.g.\ assume that the load of $i$ is smaller than the load of $j$ by an additive term~$\Delta$.
Then a biased coin is tossed showing heads with probability $p \coloneqq \min\{1, (1 + \AdaptiveParam \cdot \Delta) / 2\}$ and tails otherwise, where $\beta$ is a suitably chosen and non-constant parameter.
If the coin hits heads one item is allocated to $i$ and otherwise to $j$. 
A formal analysis of this allocation process (as well as of other, related balls-into-bins processes) is beyond the scope of our paper and remains an open problem.

\bibliography{references}

\begin{thebibliography}{10}

\bibitem{abramowitz+stegun}
Milton Abramowitz and Irene~A. Stegun.
\newblock {\em Handbook of Mathematical Functions with Formulas, Graphs, and
  Mathematical Tables}.
\newblock Dover, New York, ninth dover printing, tenth gpo printing edition,
  1964.

\bibitem{DBLP:journals/dc/AckermannFHS11}
Heiner Ackermann, Simon Fischer, Martin Hoefer, and Marcel Sch{\"{o}}ngens.
\newblock Distributed algorithms for {QoS} load balancing.
\newblock {\em Distributed Comput.}, 23(5-6):321--330, 2011.
\newblock \href {https://doi.org/10.1007/s00446-010-0125-1}
  {\path{doi:10.1007/s00446-010-0125-1}}.

\bibitem{DBLP:journals/aam/AksoyCTT18}
Sinan~G. Aksoy, Fan Chung, Michael Tait, and Josh Tobin.
\newblock The maximum relaxation time of a random walk.
\newblock {\em Adv. Appl. Math.}, 101:1--14, 2018.
\newblock \href {https://doi.org/10.1016/j.aam.2018.07.002}
  {\path{doi:10.1016/j.aam.2018.07.002}}.

\bibitem{DBLP:conf/icalp/AlistarhNS20}
Dan Alistarh, Giorgi Nadiradze, and Amirmojtaba Sabour.
\newblock Dynamic averaging load balancing on cycles.
\newblock In {\em 47th International Colloquium on Automata, Languages, and
  Programming, {ICALP} 2020}, volume 168 of {\em LIPIcs}, pages 7:1--7:16.
  Schloss Dagstuhl - Leibniz-Zentrum f{\"{u}}r Informatik, 2020.
\newblock \href {https://doi.org/10.4230/LIPIcs.ICALP.2020.7}
  {\path{doi:10.4230/LIPIcs.ICALP.2020.7}}.

\bibitem{DBLP:journals/siamcomp/AnagnostopoulosKU05}
Aris Anagnostopoulos, Adam Kirsch, and Eli Upfal.
\newblock Load balancing in arbitrary network topologies with stochastic
  adversarial input.
\newblock {\em SIAM Journal on Computing}, 34(3):616--639, 2005.
\newblock \href {https://doi.org/10.1137/S0097539703437831}
  {\path{doi:10.1137/S0097539703437831}}.

\bibitem{DBLP:journals/siamcomp/AnshelevichKK08}
Elliot Anshelevich, David Kempe, and Jon~M. Kleinberg.
\newblock Stability of load balancing algorithms in dynamic adversarial
  systems.
\newblock {\em {SIAM} J. Comput.}, 37(5):1656--1673, 2008.
\newblock \href {https://doi.org/10.1137/050639272}
  {\path{doi:10.1137/050639272}}.

\bibitem{DBLP:journals/algorithmica/HeideOW96}
Friedhelm~Meyer auf~der Heide, Brigitte Oesterdiekhoff, and Rolf Wanka.
\newblock Strongly adaptive token distribution.
\newblock {\em Algorithmica}, 15(5):413--427, 1996.
\newblock \href {https://doi.org/10.1007/BF01955042}
  {\path{doi:10.1007/BF01955042}}.

\bibitem{DBLP:journals/jcss/BerenbrinkCFFS15}
Petra Berenbrink, Colin Cooper, Tom Friedetzky, Tobias Friedrich, and Thomas
  Sauerwald.
\newblock Randomized diffusion for indivisible loads.
\newblock {\em J. Comput. Syst. Sci.}, 81(1):159--185, 2015.
\newblock \href {https://doi.org/10.1016/j.jcss.2014.04.027}
  {\path{doi:10.1016/j.jcss.2014.04.027}}.

\bibitem{DBLP:journals/jpdc/BerenbrinkFH09}
Petra Berenbrink, Tom Friedetzky, and Zengjian Hu.
\newblock A new analytical method for parallel, diffusion-type load balancing.
\newblock {\em J. Parallel Distributed Comput.}, 69(1):54--61, 2009.
\newblock \href {https://doi.org/10.1016/j.jpdc.2008.05.005}
  {\path{doi:10.1016/j.jpdc.2008.05.005}}.

\bibitem{DBLP:conf/ipps/BerenbrinkFKK19}
Petra Berenbrink, Tom Friedetzky, Dominik Kaaser, and Peter Kling.
\newblock Tight {\&} simple load balancing.
\newblock In {\em 2019 {IEEE} International Parallel and Distributed Processing
  Symposium, {IPDPS} 2019}, pages 718--726. {IEEE}, 2019.
\newblock \href {https://doi.org/10.1109/IPDPS.2019.00080}
  {\path{doi:10.1109/IPDPS.2019.00080}}.

\bibitem{DBLP:journals/algorithmica/BerenbrinkFM08}
Petra Berenbrink, Tom Friedetzky, and Russell~A. Martin.
\newblock On the stability of dynamic diffusion load balancing.
\newblock {\em Algorithmica}, 50(3):329--350, 2008.
\newblock \href {https://doi.org/10.1007/s00453-007-9081-y}
  {\path{doi:10.1007/s00453-007-9081-y}}.

\bibitem{DBLP:conf/ipps/BerenbrinkKLM17}
Petra Berenbrink, Peter Kling, Christopher Liaw, and Abbas Mehrabian.
\newblock Tight load balancing via randomized local search.
\newblock In {\em 2017 {IEEE} International Parallel and Distributed Processing
  Symposium, {IPDPS} 2017}, pages 192--201. {IEEE} Computer Society, 2017.
\newblock \href {https://doi.org/10.1109/IPDPS.2017.52}
  {\path{doi:10.1109/IPDPS.2017.52}}.

\bibitem{Berry90}
Andrew~C. Berry.
\newblock The accuracy of the gaussian approximation to the sum of independent
  variates.
\newblock {\em Transactions of the American Mathematical Society},
  49(1):122--136, 1941.

\bibitem{DBLP:journals/tamm/BhatiaD00}
Rajendra Bhatia and Chandler Davis.
\newblock A better bound on the variance.
\newblock {\em Am. Math. Mon.}, 107(4):353--357, 2000.

\bibitem{DBLP:conf/icalp/CaiS17}
Leran Cai and Thomas Sauerwald.
\newblock Randomized load balancing on networks with stochastic inputs.
\newblock In {\em 44th International Colloquium on Automata, Languages, and
  Programming, {ICALP} 2017}, volume~80 of {\em LIPIcs}, pages 139:1--139:14,
  2017.
\newblock \href {https://doi.org/10.4230/LIPIcs.ICALP.2017.139}
  {\path{doi:10.4230/LIPIcs.ICALP.2017.139}}.

\bibitem{DBLP:journals/cc/ChandraRRST97}
Ashok~K. Chandra, Prabhakar Raghavan, Walter~L. Ruzzo, Roman Smolensky, and
  Prasoon Tiwari.
\newblock The electrical resistance of a graph captures its commute and cover
  times.
\newblock {\em Comput. Complex.}, 6(4):312--340, 1997.
\newblock \href {https://doi.org/10.1007/BF01270385}
  {\path{doi:10.1007/BF01270385}}.

\bibitem{DBLP:journals/im/ChungL06}
Fan R.~K. Chung and Lincoln Lu.
\newblock Survey: Concentration inequalities and martingale inequalities: {A}
  survey.
\newblock {\em Internet Math.}, 3(1):79--127, 2006.
\newblock \href {https://doi.org/10.1080/15427951.2006.10129115}
  {\path{doi:10.1080/15427951.2006.10129115}}.

\bibitem{DBLP:journals/pc/DiekmannFM99}
Ralf Diekmann, Andreas Frommer, and Burkhard Monien.
\newblock Efficient schemes for nearest neighbor load balancing.
\newblock {\em Parallel Comput.}, 25(7):789--812, 1999.
\newblock \href {https://doi.org/10.1016/S0167-8191(99)00018-6}
  {\path{doi:10.1016/S0167-8191(99)00018-6}}.

\bibitem{doyle84}
Peter~G. Doyle and J.~Laurie Snell.
\newblock {\em Random Walks and Electric Networks}.
\newblock Number Book 22 in Carus Mathematical Monographs. Mathematical
  Association of America, Washington, DC, 1984.

\bibitem{esseen1942liapounoff}
Carl-Gustav Esseen.
\newblock {\em On the Liapounoff Limit of Error in the Theory of Probability}.
\newblock Arkiv f{\"o}r matematik, astronomi och fysik. Almqvist \& Wiksell,
  1942.

\bibitem{fanHoeffdingInequalitySupermartingales2012}
Xiequan Fan, Ion Grama, and Quansheng Liu.
\newblock Hoeffding’s inequality for supermartingales.
\newblock {\em Stochastic Processes and their Applications},
  122(10):3545--3559, 2012.
\newblock \href {https://doi.org/10.1016/j.spa.2012.06.009}
  {\path{doi:10.1016/j.spa.2012.06.009}}.

\bibitem{DBLP:journals/siamcomp/FischerRV10}
Simon Fischer, Harald R{\"{a}}cke, and Berthold V{\"{o}}cking.
\newblock Fast convergence to wardrop equilibria by adaptive sampling methods.
\newblock {\em {SIAM} J. Comput.}, 39(8):3700--3735, 2010.
\newblock \href {https://doi.org/10.1137/090746720}
  {\path{doi:10.1137/090746720}}.

\bibitem{DBLP:conf/stoc/FriedrichS09}
Tobias Friedrich and Thomas Sauerwald.
\newblock Near-perfect load balancing by randomized rounding.
\newblock In {\em Proceedings of the 41st Annual {ACM} Symposium on Theory of
  Computing, {STOC} 2009}, pages 121--130. {ACM}, 2009.
\newblock \href {https://doi.org/10.1145/1536414.1536433}
  {\path{doi:10.1145/1536414.1536433}}.

\bibitem{DBLP:journals/siamcomp/GhoshLMMPRRTZ99}
Bhaskar Ghosh, Frank~Thomson Leighton, Bruce~M. Maggs, S.~Muthukrishnan,
  C.~Greg Plaxton, Rajmohan Rajaraman, Andr{\'{e}}a~W. Richa, Robert~Endre
  Tarjan, and David Zuckerman.
\newblock Tight analyses of two local load balancing algorithms.
\newblock {\em {SIAM} J. Comput.}, 29(1):29--64, 1999.
\newblock \href {https://doi.org/10.1137/S0097539795292208}
  {\path{doi:10.1137/S0097539795292208}}.

\bibitem{DBLP:journals/jcss/GhoshM96}
Bhaskar Ghosh and S.~Muthukrishnan.
\newblock Dynamic load balancing by random matchings.
\newblock {\em J. Comput. Syst. Sci.}, 53(3):357--370, 1996.
\newblock \href {https://doi.org/10.1006/jcss.1996.0075}
  {\path{doi:10.1006/jcss.1996.0075}}.

\bibitem{DBLP:conf/spaa/GhoshMS96}
Bhaskar Ghosh, S.~Muthukrishnan, and Martin~H. Schultz.
\newblock First and second order diffusive methods for rapid, coarse,
  distributed load balancing (extended abstract).
\newblock In {\em Proceedings of the 8th Annual {ACM} Symposium on Parallel
  Algorithms and Architectures, {SPAA} '96}, pages 72--81. {ACM}, 1996.
\newblock \href {https://doi.org/10.1145/237502.237509}
  {\path{doi:10.1145/237502.237509}}.

\bibitem{DBLP:journals/corr/HoeferS13}
Martin Hoefer and Thomas Sauerwald.
\newblock Threshold load balancing in networks.
\newblock {\em CoRR}, abs/1306.1402, 2013.
\newblock URL: \url{http://arxiv.org/abs/1306.1402}, \href
  {http://arxiv.org/abs/1306.1402} {\path{arXiv:1306.1402}}.

\bibitem{Keilson1979}
Julian Keilson.
\newblock {\em Markov Chain Models --- Rarity and Exponentiality}.
\newblock Springer New York, New York, NY, 1979.
\newblock \href {https://doi.org/10.1007/978-1-4612-6200-8_1}
  {\path{doi:10.1007/978-1-4612-6200-8_1}}.

\bibitem{DBLP:conf/focs/KempeDG03}
David Kempe, Alin Dobra, and Johannes Gehrke.
\newblock Gossip-based computation of aggregate information.
\newblock In {\em 44th Symposium on Foundations of Computer Science {(FOCS}
  2003)}, pages 482--491. {IEEE} Computer Society, 2003.
\newblock \href {https://doi.org/10.1109/SFCS.2003.1238221}
  {\path{doi:10.1109/SFCS.2003.1238221}}.

\bibitem{LANDAU19815}
H.J. Landau and A.M. Odlyzko.
\newblock Bounds for eigenvalues of certain stochastic matrices.
\newblock {\em Linear Algebra and its Applications}, 38:5--15, 1981.
\newblock \href {https://doi.org/10.1016/0024-3795(81)90003-3}
  {\path{doi:10.1016/0024-3795(81)90003-3}}.

\bibitem{DBLP:series/ncs/Lengler20}
Johannes Lengler.
\newblock Drift analysis.
\newblock In Benjamin Doerr and Frank Neumann, editors, {\em Theory of
  Evolutionary Computation - Recent Developments in Discrete Optimization},
  Natural Computing Series, pages 89--131. Springer, 2020.
\newblock \href {https://doi.org/10.1007/978-3-030-29414-4\_2}
  {\path{doi:10.1007/978-3-030-29414-4\_2}}.

\bibitem{LevinPeresbook}
David Levin and Yuval Peres.
\newblock {\em Markov Chains and Mixing Times}.
\newblock AMS, 2017.
\newblock \href {https://doi.org/10.1090/mbk/107} {\path{doi:10.1090/mbk/107}}.

\bibitem{liggett1985interacting}
Thomas~M. Liggett.
\newblock {\em Interacting Particle Systems}.
\newblock Springer, 1985.
\newblock \href {https://doi.org/10.1007/b138374} {\path{doi:10.1007/b138374}}.

\bibitem{lovasz1993random}
L{\'a}szl{\'o} Lov{\'a}sz.
\newblock Random walks on graphs.
\newblock {\em Combinatorics, Paul Erd\H{o}s is Eighty}, 2:1--46, 1993.

\bibitem{lyons_peres_2017}
Russell Lyons and Yuval Peres.
\newblock {\em Probability on Trees and Networks}.
\newblock Cambridge Series in Statistical and Probabilistic Mathematics.
  Cambridge University Press, 2017.
\newblock \href {https://doi.org/10.1017/9781316672815}
  {\path{doi:10.1017/9781316672815}}.

\bibitem{McDiarmid1998}
Colin McDiarmid.
\newblock Concentration.
\newblock In {\em Probabilistic Methods for Algorithmic Discrete Mathematics},
  pages 195--248. Springer Berlin Heidelberg, 1998.
\newblock \href {https://doi.org/10.1007/978-3-662-12788-9_6}
  {\path{doi:10.1007/978-3-662-12788-9_6}}.

\bibitem{DBLP:conf/dimacs/Meyerhenke12}
Henning Meyerhenke.
\newblock Shape optimizing load balancing for mpi-parallel adaptive numerical
  simulations.
\newblock In {\em Graph Partitioning and Graph Clustering, 10th {DIMACS}
  Implementation Challenge Workshop}, volume 588 of {\em Contemporary
  Mathematics}, pages 67--82. American Mathematical Society, 2012.

\bibitem{DBLP:books/daglib/0012859}
Michael Mitzenmacher and Eli Upfal.
\newblock {\em Probability and Computing: Randomized Algorithms and
  Probabilistic Analysis}.
\newblock Cambridge University Press, 2005.
\newblock \href {https://doi.org/10.1017/CBO9780511813603}
  {\path{doi:10.1017/CBO9780511813603}}.

\bibitem{DBLP:journals/access/MohammadianNHD22}
Vahid Mohammadian, Nima~Jafari Navimipour, Mehdi Hosseinzadeh, and Aso~Mohammad
  Darwesh.
\newblock Fault-tolerant load balancing in cloud computing: {A} systematic
  literature review.
\newblock {\em {IEEE} Access}, 10:12714--12731, 2022.
\newblock \href {https://doi.org/10.1109/ACCESS.2021.3139730}
  {\path{doi:10.1109/ACCESS.2021.3139730}}.

\bibitem{DBLP:journals/mst/MuthukrishnanGS98}
S.~Muthukrishnan, Bhaskar Ghosh, and Martin~H. Schultz.
\newblock First- and second-order diffusive methods for rapid, coarse,
  distributed load balancing.
\newblock {\em Theory Comput. Syst.}, 31(4):331--354, 1998.
\newblock \href {https://doi.org/10.1007/s002240000092}
  {\path{doi:10.1007/s002240000092}}.

\bibitem{DBLP:journals/aes/PatzakR12}
Borek Patz{\'{a}}k and Daniel Rypl.
\newblock Object-oriented, parallel finite element framework with dynamic load
  balancing.
\newblock {\em Adv. Eng. Softw.}, 47(1):35--50, 2012.
\newblock \href {https://doi.org/10.1016/j.advengsoft.2011.12.008}
  {\path{doi:10.1016/j.advengsoft.2011.12.008}}.

\bibitem{DBLP:journals/siamcomp/PelegU89}
David Peleg and Eli Upfal.
\newblock The token distribution problem.
\newblock {\em {SIAM} J. Comput.}, 18(2):229--243, 1989.
\newblock \href {https://doi.org/10.1137/0218015} {\path{doi:10.1137/0218015}}.

\bibitem{DBLP:conf/focs/RabaniSW98}
Yuval Rabani, Alistair Sinclair, and Rolf Wanka.
\newblock Local divergence of markov chains and the analysis of iterative load
  balancing schemes.
\newblock In {\em 39th Annual Symposium on Foundations of Computer Science,
  {FOCS} '98}, pages 694--705. {IEEE} Computer Society, 1998.
\newblock \href {https://doi.org/10.1109/SFCS.1998.743520}
  {\path{doi:10.1109/SFCS.1998.743520}}.

\bibitem{DBLP:conf/focs/SauerwaldS12}
Thomas Sauerwald and He~Sun.
\newblock Tight bounds for randomized load balancing on arbitrary network
  topologies.
\newblock In {\em 53rd Annual {IEEE} Symposium on Foundations of Computer
  Science, {FOCS} 2012}, pages 341--350. {IEEE} Computer Society, 2012.
\newblock \href {https://doi.org/10.1109/FOCS.2012.86}
  {\path{doi:10.1109/FOCS.2012.86}}.

\bibitem{DBLP:journals/ijhpca/ZhengBMK11}
Gengbin Zheng, Abhinav Bhatele, Esteban Meneses, and Laxmikant~V. Kal{\'{e}}.
\newblock Periodic hierarchical load balancing for large supercomputers.
\newblock {\em Int. J. High Perform. Comput. Appl.}, 25(4):371--385, 2011.
\newblock \href {https://doi.org/10.1177/1094342010394383}
  {\path{doi:10.1177/1094342010394383}}.

\end{thebibliography}

\appendix

%\section*{APPENDIX}

\section{Auxiliary Results}

\subsection{Random Walks, Hitting Times, and Effective Resistance}
\label{apx:aux}

In this appendix we present for completeness fundamental definitions and relations concerning random walks, hitting times, and the effective resistance.
We start with a definition of the effective resistance of a network in \cref{def:effective_resistance}. For a motivation of the definition see \cite[Chapter 9]{LevinPeresbook}.
Further details and properties can also be found in~\cite{doyle84} and~\cite[Section 4]{lovasz1993random}.

\begin{definition}[Harmonic Functions and Effective Resistance]\label{def:effective_resistance}
Let $G$ be a graph and let $i,j \in [n]$ be nodes of the graph.
Then a \emph{harmonic function on $G$ with the poles $i$ and $j$}  (for unit edge weights) is a function $f : [n] \to \R$ such that
    for all $k \in [n] \setminus \{i,j\}$ we have
    $f(k) = \frac{1}{d(k)} \cdot \sum_{l \in N_G(k)} f(l)$,
    where $N_G(k)$ is the set of $k$'s neighbors in $G$.

Given a harmonic function $f$ on $G$ with the poles $i$ and $j$
    (with arbitrary boundary values $f(i) \neq f(j)$),
    the \emph{effective resistance} (or \emph{resistive distance} between $i$ and $j$ in $G$
    is given by \[\ResistiveDistance{i}{j} \coloneqq \frac{f(i) - f(j)}{\sum_{k \in N_G(i)} \abs{f(k) - f(i)}}.\]
Note that the value is not dependent on the boundary values of the harmonic function.
\end{definition}

Note that for boundary values $f(i)$ and $f(j)$ the harmonic function is unique \cite[Proposition 9.1]{LevinPeresbook}.

The following is a well-known property of effective resistances; it is a direct consequence of, e.g., Corollary 9.13 in \cite{LevinPeresbook}.
\begin{lemma}\label{lem:res_dist_leq_dist}
    Let $G$ be a graph, and write $\mathrm{d}(i,j)$ for the (standard) distance between $i$ and $j$ in $G$.
    Then $\ResistiveDistance{i}{j} \leq d(i,j)$.
\end{lemma}

For a graph $G$, and nodes $i,j \in V(G)$,
    let $H(i,j)$ be the \emph{hitting time from $i$ to $j$}, i.e., the expected time for a random walk on $G$ starting at $i$ to reach $j$ for the first time.

\begin{theorem}[Theorem 4.1 (i) in~\cite{lovasz1993random}]
    \label{thm:commute_time_resistance_relation}
    Let $G$ be a graph.
    Then for any $i,j \in V(G)$,
        \[H(i,j) + H(j,i) = 2 \cdot \abs{E} \cdot \ResistiveDistance{i}{j}.\]
\end{theorem}

\begin{corollary}
    \label{cor:hitting_time_resistance_relation_simple}
    Let $G$ be a graph.
    Then for any $i,j \in V(G)$,
        \[\max\{H(i,j), H(j,i)\} \leq 2 \cdot \abs{E(G)} \cdot \ResistiveDistance{i}{j} \leq 2\cdot \max\{H(i,j), H(j,i)\}.\]
\end{corollary}

\begin{proof}
    For the first inequality,
        since one of $H(i,j)$ and $H(j,i)$ is at least the maximum of the two,
            we have, by \cref{thm:commute_time_resistance_relation}:
    \[\max\{H(i,j), H(j,i)\} \leq H(i,j) + H(j,i) = 2 \cdot \abs{E(G)} \cdot \ResistiveDistance{i}{j}.\]
    And for the second inequality,
        since both $H(i,j)$ and $H(j,i)$ are at most the maximum of the two,
            we have, again by \cref{thm:commute_time_resistance_relation}
    \[2 \cdot \abs{E(G)} \cdot \ResistiveDistance{i}{j} = H(i,j) + H(j,i) \leq 2 \cdot \max\{H(i,j), H(j,i)\},\]
    as claimed.
\end{proof}

\begin{theorem}[Dirichlet's principle, see Exercise 2.13 in \cite{lyons_peres_2017}; or Exercise 9.9 in \cite{LevinPeresbook}, referencing Theorem 6.1 in  \cite{liggett1985interacting}]
    \label{thm:dirichlet_principle}
    Let $u, v$ be distinct nodes of a graph $G$.
    Then
    \[\min_{\substack{\vec{a} \in \R^n \\ a_v = 1 \\ a_u = 0}} \EdgePotential_G(\vec{a})
        = \frac{1}{\ResistiveDistance{u}{v}}.\]
\end{theorem}

\begin{theorem}[Corollary 3.3 in \cite{lovasz1993random}, applied to $d$-regular graphs]\label{thm:spectral_bound_on_commute_time}
    Let $G$ be an arbitrary graph on $n$ nodes.
    Then \[n \leq H(i,j) + H(j,i) \leq \frac{n}{\SpectralGap(\Laplacian(G))}.\]
\end{theorem}

\subsection{Tail Bounds}
\label{apx:known-results-probability-theory}

The following lemma allows us to turn a high-probability bound into a bound on the expected value. 
We consider this result folklore.
For completeness we give a formal proof below.

\begin{lemma}\label{lem:tail_bound_to_expectation_bound}
    Let $X$ be a non-negative real random variable, and let $n \in \N$.
    Then if there are $c, C > 0$ such that for all $\gamma > 0$,
        \[\AutoProb{X \geq (\gamma + 1)C} \leq c n^{-\gamma},\]
    then
        \[\AutoExp{X} \leq C \cdot \left(1 + \frac{c}{\log(n)}\right).\]
\end{lemma}

\begin{proof}
    Observe that when $x = (\gamma + 1)C$ we have $\gamma = \frac{x}{C} - 1$, so that for all $x \geq C$ we have
        \[\AutoProb{X \geq x} \leq c\cdot n^{-\frac{x}{C}+1}.\]
    Thus,
    \begin{align*}
    \AutoExp{X}
       &= \int_0^\infty \AutoProb{X \geq x}\,\dx
        = \int_0^C \AutoProb{X \geq x}\,\dx + \int_C^\infty \AutoProb{X \geq x}\,\dx
    \\ &\leq \int_0^C 1\,\dx + \int_C^\infty c \cdot n^{-\frac{x}{C}+1}\,\dx
        = C + \left[-\,\frac{c C n^{1 - \frac{x}{C}}}{\log(n)}\right]_{x=C}^\infty
        = C + \left[0 + \frac{cCn^{1-1}}{\log(n)}\right]
    \\ &= C\left(1 + \frac{c}{\log(n)}\right),
    \end{align*}
        as claimed.
\end{proof}

\begin{theorem}[Bhatia-Davis inequality \cite{DBLP:journals/tamm/BhatiaD00}]
    \label{thm:bhatia-davis}
    Let $X$ be a real random variable with $X \in [m, M]$.
    Then \(\AutoVar{X} \leq (M - \AutoExp{X})(\AutoExp{X} - m).\)
\end{theorem}

\begin{theorem}[Azuma--Hoeffding inequality Theorem 13.6 in \cite{DBLP:books/daglib/0012859}]
    \label{thm:azuma_hoeffding}
    Let $(X(t))_{t=0}^n$ be a martingale associated with the filter $(\mathcal{F}(t))_{t=0}^n$,
        where there exist non-negative sequences $(a_t)_{t=1}^n$, $(b_t)_{t=1}^n$ and $(\sigma_t)_{t=1}^n$
        such that for all $t \in [n]$,
             \[-b_t \leq X(t) - X(t-1) \leq a_t.\]
    Then for all $\varepsilon > 0$,
        \[\AutoProb{\abs*{X(n)- \AutoExp{X(n)}} \ge \varepsilon} \leq 2\exp\left(-\,\frac{2 \varepsilon^2}{\sum_{i=1}^n (a_t + b_t)^2}\right).\]
\end{theorem}

\begin{theorem}[Adapted from Theorem 6.6 in~\cite{DBLP:journals/im/ChungL06}]
    \label{thm:chung66}
    Let $(X(t))_{t=0}^n$ be a martingale associated with the filter $(\mathcal{F}(t))_{t=0}^n$,
        where there exist $(a_t)_{t=1}^n$ and $(\sigma_t)_{t=1}^n$
        such that for all $t \in [n]$,
        \begin{enumerate}
            \item $X(t) - X(t-1) \geq a_t$;
            \item $\AutoVarCond{X(t)}{\mathcal{F}(t-1)} \leq \sigma_t^2$.
        \end{enumerate}
    Then for all $\varepsilon > 0$,
        \[\AutoProb{X(n) \leq \AutoExp{X(n)} - \varepsilon} \leq \smash[t]{\exp\left(-\,\frac{\varepsilon^2}{2 \sum_{i=1}^n (a_t^2 + \sigma_t^2)}\right)}.\]
\end{theorem}

\begin{theorem}[Adapted from Theorem 2.1 and combined with Remark 2.1 and Equation 18 in~\cite{fanHoeffdingInequalitySupermartingales2012}]
    \label{thm:supermartingale_quad_characteristic}
    Let $(X(t))_{t=0}^n$ be a supermartingale associated with the filter $(\mathcal{F}(t))_{t=0}^n$,
        where $X(t) - X(t-1) \leq 1$ for all $t \in [n]$.
    Let $\langle X \rangle$ be the quadratic characteristic of $X$, i.e., let
        \[\langle X \rangle_0 = 0,\quad \langle X \rangle_t = \sum_{\tau=1}^t \AutoExpCond{(X(\tau) - X(\tau-1))^2}{\mathcal{F}(\tau - 1)},\quad\forall t \in [n].\]
    Then, for any $\varepsilon \geq 0$ and $\sigma > 0$,
        \[\BigAutoProb{\exists t \in [n]: X(t) - X(0) \geq \frac{\varepsilon}{3} + v\sqrt{2\varepsilon}\wedge \langle X \rangle_t \leq \sigma^2} \leq e^{-\varepsilon}.\]
\end{theorem}

\begin{corollary}
    \label{corr:martingale_quad_characteristic}
    Let $(X(t))_{t=0}^n$ be a martingale associated with the filter $(\mathcal{F}(t))_{t=0}^n$,
        where $\abs{X(t) - X(t-1)} \leq 1$ for all $t \in [n]$.
    Then with $\langle X \rangle$ as in \cref{thm:supermartingale_quad_characteristic},
        for any $\varepsilon \geq 0$ and $\sigma > 0$,
    \[\BigAutoProb{\abs{X(n) - X(0)} \geq \frac{\varepsilon}{3} + v\sqrt{2\varepsilon}} \leq 2(e^{-\varepsilon} + \AutoProb{\langle X \rangle_n > v^2}).\]
\end{corollary}

\begin{proof}
    As $(X(t))_{t=0}^n$ is a martingale, it is also a supermartingale,
        and it fulfills the conditions of \cref{thm:supermartingale_quad_characteristic} by the assumptions of the claim.
    So way may use \cref{thm:supermartingale_quad_characteristic} to see that
    \begin{align*}&\BigAutoProb{X(n) - X(0) \geq \frac{\varepsilon}{3} + \sigma\sqrt{2\varepsilon} \wedge \langle X \rangle_n \leq \sigma^2}
    \\ &\quad\quad\quad\quad\quad\quad\leq \BigAutoProb{\exists t \in [n]: X(t) - X(0) \geq \frac{\varepsilon}{3} + \sigma\sqrt{2\varepsilon}\wedge \langle X \rangle_t \leq \sigma^2}
        \leq e^{-\varepsilon}.
    \end{align*}
    As $\AutoProb{A} \leq \AutoProb{(A \wedge B) \vee B} \leq \AutoProb{A \wedge B} + \AutoProb{B},$
        this implies that 
    \[\BigAutoProb{X(n) - X(0) \geq \frac{\varepsilon}{3} + \sigma\sqrt{2\varepsilon}}
        \leq e^{-\varepsilon} + \AutoProb{\langle X \rangle_n \leq \sigma^2}.\]

    The claim follows from applying the same argument to the supermartingale $(-X(t))_{t=0}^n$ and a union bound.
\end{proof}

 \begin{theorem}[Berry-Esseen Theorem~\cite{Berry90,esseen1942liapounoff} for Non-identical Random Variables]\label{thm:BerryEssen}
        Let $Y_1,Y_2,\cdots, Y_k$ be independently distributed with $\E[Y_i]=0$, $\E[Y_i^2]=\Var[Y_i]=\sigma_i^2$ and $\E[|Y_i|^3]=\rho_i<\infty$. If $F_k(x)$ is the distribution of $\frac{Y_1+Y_2+\cdots+Y_k}{\sqrt{\sigma_1^2+\sigma_2^2+\cdots+\sigma_k^2}} $ and $\Phi_N(x)$ is the standard normal distribution, then 
\[|F_k(x)-\Phi_N(x)|\le C_0\cdot \psi_0,\]
    where $\psi_0=\frac{\sum_{i=1}^k \rho_i}{\left(\sum_{i=1}^k \sigma_i^2\right)^{3/2}}$ and $C_0$ is a constant.
\end{theorem}

\begin{theorem}[Theorem 3.4 of~\cite{DBLP:journals/im/ChungL06}, \cite{McDiarmid1998}]\label{thm:upper:bound:on:sum}
    let $X_i$ ($1\le i\le n$) be independent random variables satisfying $X_i\le\E[X_i]+M$, for $1\le i\le n$. We consider the sum $X=\sum_{i=1}^n X_i$ with expectation $\E[X]=\sum_{i=1}^{n}\E[X_i]$ and variance $\Var[X]=\sum_{i=1}^n \Var[X_i]$. Then we have
    \[
    \Pr\left[X\ge \E[X]+\lambda\right]\le \exp\left(-\frac{\lambda^2}{2\cdot(\Var[X]+M\lambda/3)}\right).
    \]
\end{theorem}
    
\begin{theorem}[Theorem 4.1 of~\cite{DBLP:journals/im/ChungL06}]\label{thm:lower:on:sum}
    Let $X_i$ denote independent random variable satisfying $X_i\ge \E[X_i]-a_i-M$ for $0\le i\le n$. For $X=\sum_{i=1}^n X_i$ we have
\[
\Pr\left[X\le\E[X]-\lambda\right]\le \exp\left(-\frac{\lambda^2}{2\cdot \left(\Var[X]+\sum_{i=1}^n a_i^2 + M\lambda/3\right)} \right).
    \]
\end{theorem}

\section{Omitted Proofs from \cref{sec:analysis:random:matching}}
\label{apx:omitted-proofs-3}

In this appendix we present the omitted proofs from \cref{sec:analysis:random:matching}.
We first formally prove that the discrepancy is sub-additive.

\begin{observation}\label{obs:disc_subadditive}
    For two vectors $\vec{x}, \vec{y} \in \R^n$,
        \[\disc(\vec{x} + \vec{y}) \leq \disc(\vec{x}) + \disc(\vec{y}).\]
\end{observation}

\begin{proof}
    For any $\vec{a}, \vec{b} \in \R^n$,
        \[\max_{i \in [n]} (a_i + b_i) \leq \max_{i \in [n]} a_i + \max_{i \in [n]} b_i,\]
        and thus
    \begin{align*}
    \disc(\vec{x} + \vec{y})
       &= \max_{i \in [n]} (x_i + y_i) - \min_{i \in [n]} (x_i + y_i)
        = \max_{i \in [n]} (x_i + y_i) + \max_{i \in [n]} ((-x_i) + (-y_i))
    \\ &\leq \max_{i \in [n]} x_i + \max_{i \in [n]} y_i + \max_{i \in [n]} (-x_i) + \max_{i \in [n]} (-y_i)
    \\ &= \left(\max_{i \in [n]} x_i - \min_{i \in [n]} x_i\right) + \left(\max_{i \in [n]} y_i - \min_{i \in [n]} y_i\right)
    \\ &= \disc(\vec{x}) + \disc(\vec{y}),
    \end{align*}
        as claimed.
\end{proof}

\subsection{Proof of \cref{lem:initial:load:vanishes}}

\restateInitialLoadVanishes*

\label{proof:lem:initial:load:vanishes}
\begin{proof}
 To bound $\discr(\InitialContribVecT{t})$,
        we use the following claim:
    \begin{claim*}
        If $t \geq t_0(0)$, then
        \(\AutoExp{\NodePotential(\InitialContribVecT{t})} \leq 1/4,\)
        and if $t \geq t_0(\gamma)$, then
        \(\AutoProb{\NodePotential(\InitialContribVecT{t}) \leq \frac{1}{4}} \geq 1 - n^{-\gamma}.\)
    \end{claim*}
    First, note that $\max_{i \in [n]} \abs{x_i - \overline{x}} \leq \sqrt{\NodePotential(\vec{x})}$ by definition of $\NodePotential$.
    Hence, $\discr(\vec{x}) \leq 2\sqrt{\NodePotential(\vec{x})}$.
    By the claim, if $t \geq t_0(\gamma)$, then $\NodePotential(\InitialContribVecT{t}) \leq 1/4$ with probability at least $1 - n^{-\gamma}$,
        and hence $\discr(\InitialContribVecT{t}) \leq 2\sqrt{\NodePotential(\InitialContribVecT{t})} \leq 2\sqrt{1/4} = 1$.
    Also by the claim,
        if $t \geq t_0(0)$, then $\AutoExp{\InitialContribVecT{t}} \leq 1/4$,
        and then by Jensen's inequality,
        \[\AutoExp{\discr(\InitialContribVecT{t})} \leq \BigAutoExp{2 \sqrt{\NodePotential(\InitialContribVecT{t})}}
            \leq 2 \sqrt{\AutoExp{\NodePotential(\InitialContribVecT{t})}} \leq 2 \sqrt{\frac{1}{4}} = 1.\]

    \begin{claimproof}[Proof of the claim]
    We aim to use the first statement of \cref{lem:drift} on $\NodePotential(\InitialContribVecT{t})$ and therefore need to check its preconditions.
    By the definition of $\InitialContribVecT{t}$, for all $t \geq 1$,
        \[\InitialContribVecT{t} = \MixMatTT{1}{t} \cdot \LoadVecT{0} = \MixMatBeta{\BalancingSpeed}(t) \cdot \MixMatTT{1}{t-1} \cdot \LoadVecT{0} = \MixMatBeta{\BalancingSpeed}(t) \cdot \InitialContribVecT{t-1}.\]
    Entirely analogous to the calculations in the proof of \cref{lem:glob:div:bound:drift} (\cref{eq:seven,eq:eight}),
        we have, writing $\vec{V} = \InitialContribVecT{t-1}$ (so that $\InitialContribVecT{t} = \MixMatBeta{\BalancingSpeed} \cdot \vec{V}$),
    \begin{align}
            \AutoExp{\NodePotential(\vec{V}) - \NodePotential(\MixMatBeta{\BalancingSpeed}(t) \cdot \vec{V})}
                &\geq \BalancingSpeed \cdot \AutoExp{\NodePotential(\vec{V}) - \NodePotential(\MixMatBeta{1}(t) \cdot \vec{V})}, \label{eq:seven_prime}
        \intertext{and}  \AutoVar{\NodePotential(\vec{V}) - \NodePotential(\MixMatBeta{\BalancingSpeed}(t) \cdot \vec{V})}
                &\leq 4 \BalancingSpeed^2 \cdot \AutoVar{\NodePotential(\vec{V}) - \NodePotential(\MixMatBeta{1}(t) \cdot \vec{V})}, \notag
    \end{align}
        and from the latter it immediately follows that for all $\varphi$
    \begin{align*}
    \AutoVarCond{\NodePotential(\InitialContribVecT{t})}{\NodePotential(\vec{V}) = \varphi}
       &= \AutoVarCond{\NodePotential(\MixMatBeta{\BalancingSpeed}(t) \cdot \vec{V})}{\NodePotential(\vec{V}) = \varphi}
    \\ &= \AutoVarCond{\varphi - \NodePotential(\MixMatBeta{\BalancingSpeed}(t) \cdot \vec{V})}{\NodePotential(\vec{V}) = \varphi}
    \\ &\leq 4 \BalancingSpeed^2 \cdot \AutoVarCond{\varphi - \NodePotential(\MixMatBeta{1}(t) \cdot \vec{V})}{\NodePotential(\vec{V}) = \varphi},
    \\ &= 4 \BalancingSpeed^2 \cdot \AutoVarCond{\NodePotential(\MixMatBeta{1}(t) \cdot \vec{V})}{\NodePotential(\vec{V}) = \varphi}.
    \end{align*}
    Combining the first statement of \cref{lem:edge_potential_bounds} and the first statement of \cref{prop:node_potential_change_statistics} gives us, for all $\vec{x} \in \R^n$,
    \[\NodePotential(\vec{x}) - \AutoExp{\NodePotential(\MixMatBeta{1}(t) \cdot \vec{x})} \geq \frac{\SpectralGap(\Laplacian(G))}{16} \cdot \NodePotential(\vec{x}),\]
        so that, for all $\varphi$,
    \[\AutoExpCond{\NodePotential(\InitialContribVecT{t})}{\NodePotential(\vec{V}) = \varphi}
        = \AutoExpCond{\NodePotential(\MixMatBeta{\BalancingSpeed}(t) \cdot \vec{V})}{\NodePotential(\vec{V}) = \varphi}
        \leq \varphi - \BalancingSpeed \cdot \frac{\SpectralGap(\Laplacian(G))}{16} \cdot \varphi.\]

    By the second statement of \cref{prop:node_potential_change_statistics}, for all $\vec{x} \in \R^n$:
    \[\BigAutoVar{\NodePotential(\MixMatBeta{1}(t) \cdot \vec{x})}
            \leq (32 \cdot (\EdgeHittingTime / n) + 4) \cdot\left(\NodePotential(\vec{x}) - \BigAutoExp{\NodePotential(\MixMatBeta{1}(t) \cdot \vec{x})}\right)^2.\]
    And so,
    \begin{align*}\MoveEqLeft
    \BigAutoVarCond{\NodePotential(\InitialContribVecT{t})}{\NodePotential(\vec{V}) = \varphi}
    \\ &\leq 4 \BalancingSpeed^2 \cdot \AutoVarCond{\NodePotential(\MixMatBeta{1}(t) \cdot \vec{V})}{\NodePotential(\vec{V}) = \varphi}
    \\ &\leq 4 \BalancingSpeed^2 \cdot \left(32 \cdot \frac{\EdgeHittingTime}{n} + 4\right) \cdot\left(\varphi - \BigAutoExpCond{\NodePotential(\MixMatBeta{1}(t) \cdot \vec{V})}{\NodePotential(\vec{V}) = \varphi}\right)^2
    \\ &= \left(128 \cdot \frac{\EdgeHittingTime}{n} + 16\right) \cdot \left(\BalancingSpeed \cdot \BigAutoExpCond{\varphi - \NodePotential\left(\MixMatBeta{1}(t) \cdot \vec{V}\right)}{\NodePotential(\vec{V}) = \varphi}\right)^2
    \\ &\overset{(\ref{eq:seven_prime})}{\leq} \left(128 \cdot \frac{\EdgeHittingTime}{n} + 16\right) \cdot \left(\BigAutoExpCond{\varphi - \NodePotential\left(\MixMatBeta{\BalancingSpeed}(t) \cdot \vec{V}\right)}{\NodePotential(\vec{V}) = \varphi}\right)^2
    \\ &= \left(128 \cdot \frac{\EdgeHittingTime}{n} + 16\right) \cdot \left( \BigAutoExpCond{\NodePotential\left(\MixMatBeta{\BalancingSpeed}(t) \cdot \vec{V}\right)}{\NodePotential(\vec{V}) = \varphi} -  \varphi\right)^2.
    \end{align*}

    So we can now apply \cref{lem:drift} with
        \[h(x) \coloneqq \BalancingSpeed \cdot \frac{\SpectralGap(\Laplacian(G))}{16} \cdot x;\quad\quad
        \sigma \coloneqq 128 \cdot \frac{\EdgeHittingTime}{n} + 16.\]
    With these values and $\delta = 1/2$, the first statement of \cref{lem:drift} gives us 
   \[\BigAutoProb{\int_{\NodePotential(\InitialContribVecT{t})}^{\NodePotential(\InitialContribVecT{0})} \frac{1}{h(\varphi)}\,\dvarphi \leq t/2} \leq \exp\left(-\,\frac{t}{8(\sigma + 1)}\right).\]
    The integral evaluates to
    \begin{align*}
        \int_{\NodePotential(\InitialContribVecT{t})}^{\NodePotential(\InitialContribVecT{0})} \frac{1}{h(\varphi)}\,\dvarphi
           &= \frac{16}{\BalancingSpeed \SpectralGap(\Laplacian(G))} \cdot \int_{\NodePotential(\InitialContribVecT{t})}^{\NodePotential(\InitialContribVecT{0})} \frac{1}{\varphi}\,\dvarphi
            = \log\left(\frac{\NodePotential(\InitialContribVecT{0})}{\NodePotential(\InitialContribVecT{t})}\right) \cdot \frac{16}{\BalancingSpeed \cdot \SpectralGap(\Laplacian(G))}.
    \end{align*}
    This is at least $t/2$ if and only if
        \[\NodePotential(\InitialContribVecT{t}) \leq \NodePotential(\InitialContribVecT{0}) \cdot \exp\left(-\,\frac{\BalancingSpeed \cdot \SpectralGap(\Laplacian(G))}{32} \cdot t\right),\]
        which follows after rearranging the initial inequality and exponentiation.
    So
        \begin{equation}\label{eq:initial_contrib_raw_tail_bound}\BigAutoProb{\NodePotential(\InitialContribVecT{t}) \leq \NodePotential(\InitialContribVecT{0}) \cdot \exp\left(-\,\frac{\BalancingSpeed \cdot \SpectralGap(\Laplacian(G))}{32} \cdot t\right)} \geq 1 - \exp\left(-\,\frac{t}{8(\sigma + 1)}\right).\end{equation}

    Now, let $K \coloneqq \discr(\InitialContribVecT{0}) = \discr(\LoadVecT{0})$.
    Then in particular, $\NodePotential(\InitialContribVecT{0}) \leq n \cdot K^2$,
        so that $\log(\NodePotential(\InitialContribVecT{0})) \leq 2\log(K \cdot n)$.
    Furthermore, it is the case that $0.5 \leq \EdgeHittingTime/n \leq 1 / \SpectralGap(\Laplacian(G))$ (by \cref{thm:spectral_bound_on_commute_time}) and that $\BalancingSpeed \in (0,1]$.

    Therefore, there is a sufficiently large constant $c > 0$
        such that if $t \geq t_0(\gamma) = c \cdot \max\{\gamma \log(n), \log(K \cdot n)\} / (\BalancingSpeed \cdot \SpectralGap(\Laplacian(G)))$,
        then \[t \geq \frac{\BalancingSpeed \cdot \SpectralGap(\Laplacian(G))}{32} \cdot \log(8 \cdot \NodePotential(\InitialContribVecT{0})),\]
        as well as
        \begin{align*}
        t  &\geq \max\{\gamma \log(n), \log(\NodePotential(\InitialContribVecT{0}))\} \cdot 8 \cdot \left(128 \cdot \frac{\EdgeHittingTime}{n} + 33\right)
        \\ &= \max\{\gamma \log(n), \log(\NodePotential(\InitialContribVecT{0}))\} \cdot 8(\sigma + 1).
        \end{align*}
    From $t \geq \frac{\BalancingSpeed \cdot \SpectralGap(\Laplacian(G))}{32} \cdot \log(8 \cdot \NodePotential(\InitialContribVecT{0}))$, it follows that
        \begin{equation*}
        \NodePotential(\InitialContribVecT{0}) \cdot \exp\left(-\,\frac{\BalancingSpeed \cdot \SpectralGap(\Laplacian(G))}{32} \cdot t\right) \leq \frac{1}{8}.
        \end{equation*}
    From $t \geq \max\{\gamma \log(n), \log(\NodePotential(\InitialContribVecT{0}))\} \cdot 8(\sigma + 1)$, it follows that
        \[\exp\left(-\,\frac{t}{8(\sigma + 1)}\right) \leq \min\left\{n^{-\gamma}, \frac{1}{8 \cdot \NodePotential(\InitialContribVecT{t})}\right\}.\]
    And so, for $t \geq t_0(\gamma)$, \cref{eq:initial_contrib_raw_tail_bound} entails
        \[\BigAutoProb{\NodePotential(\InitialContribVecT{t}) \leq \frac{1}{8}} \geq 1 - n^{-\gamma},\]
        which is the remaining claim for the high-probability statement.

    For the remaining claim (i.e., the statement concerning the expectation),
        note that for $t \geq t_0(0),$ the calculations above and \cref{eq:initial_contrib_raw_tail_bound} entail that
        \[\BigAutoProb{\NodePotential(\InitialContribVecT{t}) \leq \frac{1}{8}} \geq 1 - \frac{1}{8 \cdot \NodePotential(\InitialContribVecT{0})}.\]
    Hence, as $\NodePotential(\InitialContribVecT{\tau}) \leq \NodePotential(\InitialContribVecT{0})$ for all $\tau \in \N$,
        we have, for all $t \geq t_0(0)$,
        \begin{align*}
        \BigAutoExp{\NodePotential(\InitialContribVecT{t})}
           &\leq \frac{1}{8} \cdot \BigAutoProb{\NodePotential(\InitialContribVecT{t}) \leq \frac{1}{8}}
               + \NodePotential(\InitialContribVecT{0}) \cdot \BigAutoProb{\NodePotential(\InitialContribVecT{t}) > \frac{1}{8}}
        \\ &\leq \frac{1}{8} + \NodePotential(\InitialContribVecT{0}) \cdot \frac{1}{8 \cdot \NodePotential(\InitialContribVecT{0})}
            = \frac{1}{8} + \frac{1}{8} = \frac{1}{4},
        \end{align*}
        as claimed.
    \end{claimproof}
This concludes the proof of the lemma.
\end{proof}

\subsection{Proof of \cref{lem:rounding:errors:are:small}}

\restateRoundingErrorsAreSmall*

\label{p:lem:rounding:errors:are:small}
The proof is similar to the proof of \cite[Theorem 3.4]{DBLP:conf/focs/SauerwaldS12}.
\begin{proof}
    We show the concentration bound on $\discr(\RoundingContribVecT{t})$
        by proving concentration bounds on the absolute values $\abs{\NodeRoundingContribT{k}{t}}$ for each $k \in [n]$ and then applying a union bound over all $k$.
    To show the concentration bound on $\NodeRoundingContribT{k}{t}$ holds for any fixed sequence of matchings $\mixMatSeq{t} = (\mixMatBeta{\BalancingSpeed}(\tau))_{\tau=1}^t$;
        this implies a concentration bound on a random sequence of matchings
        by the law of total probability.

    So we fix $\mixMatSeq{t}$.
    Recall that
    \[\RoundingContribVecT{t} = \sum_{\tau=1}^t \mixMatTT{\tau+1}{t}_{k,\cdot} \cdot \RoundingErrVecT{\tau},\]
        where $\RoundingErrVecT{\tau} = (\RoundingErrT{k}{t})_{k \in [n]}$ is the vector of additive rounding errors incurred in round $\tau$:
        it is the difference between the load vector step~$t$,
        and what the load vector \emph{would} be after step $t$ if the balancing in this step were idealized.
    This additive rounding error stems from the constraint that only whole items can be transferred
        across the edges $\{i,j\}$ of the matching at time $\tau$.
    From the description of the protocol,
        it is immediate that the rounding errors at matched nodes sum to $0$,
        so that $\RoundingErrT{i}{\tau} = -\RoundingErrT{j}{\tau}$
            for all edges $\{i,j\} \in E(\mixMatT{\tau})$ matched in round $\tau$.
    Thus,
    \begin{align*}
        \NodeRoundingContribT{k}{t}
           &= \sum_{\tau=1}^t \mixMatTT{\tau+1}{t}_{k,\cdot} \cdot \RoundingErrVecT{\tau}
            = \sum_{\tau=1}^t \mixMatTT{\tau+1}{t}_{k,\cdot} \cdot \sum_{\{i,j\} \in E(\mixMatT{\tau})} (\RoundingErrT{i}{\tau} + \RoundingErrT{j}{\tau})
        \\ &= \sum_{\tau=1}^t \sum_{\{i,j\} \in E(\mixMatT{\tau})} \left(\mixMatTT{\tau+1}{t}_{k,i} - \mixMatTT{\tau+1}{t}_{k,j}\right) \cdot \RoundingErrT{i}{\tau}.
    \end{align*}
    We will derive the claimed tail bound on $\NodeRoundingContribT{k}{t}$
        by applying the Azuma-Hoeffding inequality (\cref{thm:azuma_hoeffding})
        to a sequence of partial sums as follows.
    We sequence the rounding actions with $\tau$ increasing
        and arbitrarily within rounds.
    If $i$ is the representative node of the $k$th edge in round $\tau$ (with $k \in [\lfloor n/2 \rfloor]$ and $\tau \in [t]$),
        for $l = (\tau - 1) \cdot \lfloor n/2 \rfloor + k$ let us write
        \[Y_l = \left(\mixMatTT{\tau+1}{t}_{k,i} - \mixMatTT{\tau+1}{t}_{k,j}\right) \cdot \RoundingErrT{i}{\tau},\]
        and let $Y_l = 0$ if there are fewer than $k$ edges are in the matching in round $\tau$.
    Se sequence of partial sums is then
        \(S_l \coloneqq \sum_{a \in [l]} Y_l\),
        which we consider with respect to the filtration $(\FilterT{l})_{l=0}^{t \cdot \lfloor n/2 \rfloor}$ in which $\FilterT{l-1}$ completely determines the state right before the rounding action corresponding to the term $Y_l$.
    Note that $S_{t \cdot \lfloor n/2 \rfloor} = \NodeRoundingContribT{k}{t}$.
    To apply \cref{thm:azuma_hoeffding},
        it is enough to show that the conditional expectation of the difference between successive terms is zero,
        and that we can bound the differences between terms.
    
    To check these preconditions,
        let us write $F_l$ for the fractional value of the load at node $i$ before the rounding action (i.e., the fractional value of the load $i$ if balancing were idealized and no rounding was necessary).
    Then the load will be rounded up with probability $F_l$, resulting in a positive rounding error of $\RoundingErrT{i}{\tau} = 1 - F_l$,
        or rounded down with probability $1 - F_l$, resulting in a negative rounding error of $\RoundingErrT{i}{\tau} = -F_l$.
    Hence,
        \[\AutoExpCond{\RoundingErrT{i}{\tau}}{\FilterT{l-1}} = F_l \cdot (1 - F_l) + (1 - F_l) \cdot (-F_l) = 0,\]
        so that, as required,
        \[\AutoExpCond{Y_l}{\FilterT{l-1}}
            = \BigAutoExpCond{\left(\mixMatTT{\tau+1}{t}_{k,i} - \mixMatTT{\tau+1}{t}_{k,j}\right) \cdot \RoundingErrT{i}{\tau}}{\FilterT{l-1}}
            = 0.\]
    From this description,
        it is also clear that writing $\delta_{i,j}(\tau) \coloneqq \mixMatTT{\tau+1}{t}_{k,i} - \mixMatTT{\tau+1}{t}_{k,j}$,
        the term $Y_l$ is bounded from above by $a_l \coloneqq \delta_{i,j}(\tau) (1 - \FractionalT{i}{\tau})$,
        and from below by $-b_l \coloneqq - \delta_{i,j}(\tau) \FractionalT{i}{\tau}$,
        so that $a_l + b_l = \delta_{i,j}(\tau)$.

    So we may apply \cref{thm:azuma_hoeffding};
        to use it we require (an upper bound on) the value of the sum $\sum_{l=1}^{\tau \cdot \lfloor n/2 \rfloor} (a_l + b_l)^2$,
        which we bound by applying \cref{obs:node_potential_change_exact} and collapsing the ensuing telescoping sum (analogously to the proof of Theorem 3.2 in~\cite{DBLP:conf/focs/SauerwaldS12}):
    \begin{align*}
        \sum_{l=1}^{\tau \cdot \lfloor n/2 \rfloor} (a_l &+ b_l)^2
            = \sum_{\tau=1}^t \underbrace{\sum_{\{i,j\} \in E(\mixMatT{\tau})} \left(\mixMatTT{\tau+1}{t}_{k,i} - \mixMatTT{\tau+1}{t}_{k,j}\right)^2}_{=\EdgePotential_{E(\mixMatT{\tau})}}
            \\& = \sum_{\tau=1}^t \frac{2}{1-(1-\BalancingSpeed)^2}\left(\NodePotential\left(\mixMatTT{\tau+1}{t}_{k,\cdot}\right) - \NodePotential\left(\mixMatTT{\tau+1}{t}_{k,\cdot} \cdot \mixMatT{\tau}\right)\right)    
    \\ & \overset{(a)}{\leq} \sum_{\tau=1}^t \frac{2}{\BalancingSpeed} \left(\NodePotential\left(\mixMatTT{\tau+1}{t}_{k,\cdot}\right) - \NodePotential\left(\mixMatTT{\tau+1}{t}_{k,\cdot} \cdot \mixMatT{\tau}\right)\right)
        \\ &= \frac{2}{\BalancingSpeed} \cdot \sum_{\tau=1}^t \left(\NodePotential\left(\mixMatTT{\tau+1}{t}_{k,\cdot}\right) - \NodePotential\left(\mixMatTT{\tau}{t}_{k,\cdot}\right)\right)
            = \frac{2}{\BalancingSpeed} \cdot \left(\NodePotential\left(\mixMatTT{t+1}{t}_{k,\cdot}\right) - \NodePotential\left(\mixMatTT{1}{t}_{k,\cdot}\right)\right)
        \\ &= \frac{2}{\BalancingSpeed} \cdot \left(\NodePotential\left(\mathbf{I}_{k,\cdot}\right) - \NodePotential\left(\mixMatTT{1}{t}_{k,\cdot}\right)\right)
            \leq \frac{2}{\BalancingSpeed} \cdot (1 - 0) = \frac{2}{\BalancingSpeed},
    \end{align*}
    where $(a)$ follows from the fact that $\BalancingSpeed\in(0,1]$ and therefore, $1-(1-\BalancingSpeed)^2\ge \BalancingSpeed$.
    So by \cref{thm:azuma_hoeffding} (with $\varepsilon = \sqrt{(\gamma+1) \log(n) / \BalancingSpeed}$ and $\E[\NodeRoundingContribT{k}{t}]=0$) we have
        \[\BigAutoProb{\abs{\NodeRoundingContribT{k}{t}} \geq \sqrt{\frac{(\gamma+1) \log(n)}{\BalancingSpeed}}} \leq 2 \exp\left(-\,\frac{2 \varepsilon^2}{2 / \BalancingSpeed}\right)
        \leq 2 \exp(- (\gamma+1) \log(n))
        = 2 n^{-\gamma-1}.\]
    Since $\discr(\RoundingContribVecT{t}) = \max_{k\in[n]}\NodeRoundingContribT{k}{t} - \min_{k\in[n]}\NodeRoundingContribT{k}{t}$, applying a union bound over all nodes $k \in [n]$ we see that 
       \[\BigAutoProb{\discr(\RoundingContribVecT{t})\ge 2\cdot \sqrt{\frac{(\gamma+1) \log(n)}{\BalancingSpeed}}} \le 2n^{-\gamma},
       \]
       which is the claimed concentration bound. 
       
    To show the bound on $\AutoExp{\discr(\RoundingContribVecT{t}}$,
        we apply \cref{lem:tail_bound_to_expectation_bound} with $X=\discr(\RoundingContribVecT{t})$, $c=2$ and $C=2\sqrt{ \log(n)/{\BalancingSpeed}}$ to see that,
       \[
       \AutoExp{\discr(\RoundingContribVecT{t})} \le 2\sqrt{ \frac{\log(n)}{{\BalancingSpeed}}}\cdot \left(1+\frac{2}{\log(n)}\right) = \Oh\left(\sqrt{\frac{\log(n)}{\BalancingSpeed}}\right). \qedhere
       \]
\end{proof}

\subsection{Omitted Proofs from \cref{sec:bound:contribution:dynamcally:allocated:balls}}
\label{sec:Omitted:Proofs:31}

\restateObsPotentialRelation*

\begin{proof}
    We assume w.l.o.g.\ that the entries of $\vec{x}$ sum to $0$,
        meaning that $\overline{x} = 0$, so that $\NodePotential(\vec{x}) = \sum_{i \in [n]} x_i^2$.
    As loads only change at matched nodes,
        let us investigate the potential change  at two matched nodes $i$ and $j$,
        where w.l.o.g.\ $x_i \geq x_j$.
    The amount of load transferred from $i$ to $j$ under idealized balancing (without rounding) is $(x_i - x_j) \cdot \BalancingSpeed / 2$.
    So with
        \[a \coloneqq \frac{x_i + x_j}{2},\quad b \coloneqq \frac{x_i - x_j}{2},\quad c \coloneqq (1 - \BalancingSpeed) \cdot \frac{x_i - x_j}{2},\]
        the loads before balancing are $x_i = a+b$ and $x_j = a-b$,
        and the loads after idealized balancing are $x_i' = a+c$ and $x_j' = a-c$.
    So the change of the potential contributions at $i$ and $v$ is
        \[(a+b)^2 + (a-b)^2 - ((a+c)^2 + (a-c)^2) = 2(a^2 + b^2) - 2(a^2 + c^2) = 2(b^2 - c^2),\]
        where we used $(x+y)^2 + (x-y)^2 = (x^2 + 2xy + y^2) + (x^2 - 2xy + y^2) = 2x^2 + 2y^2$.
    Now, \[2 (b^2 - c^2) = 2 (1^2 - (1 - \BalancingSpeed)^2) \left(\frac{x_i - x_j}{2}\right)^2 = \frac{1 - (1-\BalancingSpeed)^2}{2} (x_i - x_j)^2.\]
    Summing this over all edges in the matching gives, as claimed,
        \[\NodePotential(\vec{x}) - \NodePotential(\MixMatBeta{\BalancingSpeed} \cdot \vec{x}) = \frac{1 - (1-\BalancingSpeed)^2}{2} \sum_{\{i,j\} \in E(\MixMatBeta{\BalancingSpeed})} (x_i - x_j)^2
            = \frac{1 - (1-\BalancingSpeed)^2}{2} \cdot \EdgePotential_{E(\MixMatBeta{\BalancingSpeed})}(\vec{x}). \qedhere \]
\end{proof}

\restateLemGlobalDivergence*
\begin{proof}
    First recall that
        \[\left(\GlobalDivergence_k(\MixMatSeq{t})\right)^2 = \sum_{\tau=1}^{t} \norm*{\MixMatTT{\tau}{t}_{k,\cdot} - \frac{\vec{1}}{n}}_2^2.\]
    As the mixing matrices are doubly stochastic, each row is a stochastic vector $\vec{x}$.
    By definition of the node potential $\NodePotential$ we know
        \[\norm*{\MixMatTT{\tau}{t}_{k,\cdot} - \frac{\vec{1}}{n}}_2^2
            = \sum_{w=1}^n \left(\MixMatTT{\tau}{t}_{k,w} - \frac{1}{n}\right)^2
            = \NodePotential\left(\MixMatTT{\tau}{t}_{k,\cdot}\right)\]
    and hence \[\left(\GlobalDivergence_k(\MixMatSeq{t})\right)^2
        = \sum_{\tau=1}^t \NodePotential(\MixMatTT{\tau}{t}_{k,\cdot}).\]
To bound this sum we will apply the second statement of \cref{lem:drift} to the sequence of values $\NodePotential(\MixMatTT{\tau}{t})$ for $\tau = t, \dots, 1$. 
Since the matching matrices $\MixMatBeta{\BalancingSpeed}(1)\ldots, \MixMatBeta{\BalancingSpeed}(t)$ are symmetric we get
    \[\NodePotential(\MixMatTT{\tau}{t}_{k,\cdot})
        = \NodePotential\left(\MixMatTT{\tau+1}{t}_{k,\cdot} \cdot \MixMatBeta{\BalancingSpeed}(\tau)\right)
        = \NodePotential\left(\MixMatBeta{\BalancingSpeed}(\tau) \cdot \MixMatTT{\tau+1}{t}_{k,\cdot}\right).\]
By \cref{obs:node_potential_change_exact} with
$S=E(\MixMatBeta{\BalancingSpeed}(\tau))$ defined as the edges of $\MixMatBeta{\BalancingSpeed}(\tau)$ we get
\begin{equation}\label{eqn:potential_change_exact_beta}
\NodePotential(\MixMatTT{\tau+1}{t}_{k,\cdot}) - \NodePotential(\MixMatTT{\tau}{t}_{k,\cdot})
    = \frac{1 - (1 - \BalancingSpeed)^2}{2} \cdot \EdgePotential_{S}(\MixMatTT{\tau+1}{t}_{k,\cdot}) \geq 0.
\end{equation}
This shows that   $\NodePotential(\MixMatTT{\tau}{t}_{k,\cdot}) \leq \NodePotential(\MixMatTT{\tau+1}{t}_{k,\cdot})$ for all $\tau$.
Expressing \cref{eqn:potential_change_exact_beta} with Balancing Parameter $1$ and, for the ease of presentation, setting $\vec{V} \coloneqq \MixMatTT{\tau+1}{t}_{k,\cdot}$ gives us
        \begin{equation*}
         \NodePotential(\MixMatTT{\tau+1}{t}_{k,\cdot})-
         \NodePotential(\MixMatTT{\tau}{t}_{k,\cdot})=
          \NodePotential(\vec{V}) - \NodePotential(\MixMatBeta{\BalancingSpeed}(\tau) \cdot \vec{V})
            = (1 - (1-\BalancingSpeed)^2) \cdot \left(\NodePotential(\vec{V}) - \NodePotential(\MixMatBeta{1}(\tau) \cdot \vec{V})\right).
        \end{equation*}
    Since $\BalancingSpeed \leq 1 - (1 - \BalancingSpeed)^2 \leq 2 \BalancingSpeed$ for $\BalancingSpeed \in (0, 1]$ we get
        \begin{align}
            \AutoExp{\NodePotential(\vec{V}) - \NodePotential(\MixMatBeta{\BalancingSpeed}(\tau) \cdot \vec{V})}
                &\geq \BalancingSpeed \cdot \AutoExp{\NodePotential(\vec{V}) - \NodePotential(\MixMatBeta{1}(\tau) \cdot \vec{V})}, \label{eq:seven}
        \\  \AutoVar{\NodePotential(\vec{V}) - \NodePotential(\MixMatBeta{\BalancingSpeed}(\tau) \cdot \vec{V})}
                &\leq 4 \BalancingSpeed^2 \cdot \AutoVar{\NodePotential(\vec{V}) - \NodePotential(\MixMatBeta{1}(\tau) \cdot \vec{V})}. \label{eq:eight}
        \end{align}
   As $\MatchDistr(G)$ is $(g, \sigma^2)$-good,
     for any stochastic vector $\vec{v} \in \R^n$
        we have
   \(\AutoExp{\NodePotential(\vec{v}) - \NodePotential(\MixMatBeta{1}(\tau)  \cdot \vec{v})} \geq g(\NodePotential(\vec{v})).\)
    Combining this with \cref{eq:seven} gives
   \[\AutoExp{\NodePotential(\vec{v}) - \NodePotential(\MixMatBeta{\BalancingSpeed}(\tau) \cdot \vec{v})} \ge \beta\cdot g(\NodePotential(\vec{v})).\]
    And thus,
\begin{align*}
        \BigAutoExpCond{\NodePotential(\MixMatTT{\tau}{t}_{k,\cdot})}{\NodePotential(\vec{V})=\varphi}
           &= \BigAutoExpCond{\NodePotential(\MixMatBeta{\BalancingSpeed}(\tau) \cdot \vec{V})}{\NodePotential(\vec{V})=\varphi} 
            \leq \varphi - \BalancingSpeed \cdot g\left(\varphi\right).
    \end{align*}
    Similarly, as $\MatchDistr(G)$ is $(g, \sigma^2)$-good,
        for any stochastic vector $\vec{v} \in \R^n$
        we have  \item \(\AutoVar{\NodePotential(\MixMatBeta{1} \cdot \vec{v})} \leq (\sigma^2 - 1) \cdot \left(\NodePotential(\vec{v}) -  \AutoExp{\NodePotential(\MixMatBeta{1} \cdot \vec{v})}\right)^2.\)
    Combining this with \cref{eq:eight} gives us
        \[\AutoVar{\NodePotential(\MixMatBeta{\BalancingSpeed} \cdot \vec{v})}
            \leq 4\BalancingSpeed^2 (\sigma^2 - 1) \left(\NodePotential(\vec{v}) -  \AutoExp{\NodePotential(\MixMatBeta{1} \cdot \vec{v})}\right)^2,\]
        and thus
    \begin{align*}
    \BigAutoVarCond{\NodePotential(\MixMatTT{\tau}{t}_{k,\cdot})}{\NodePotential( \vec{V}) = \varphi}
       &= \BigAutoVarCond{\NodePotential(\MixMatBeta{\BalancingSpeed}(\tau) \cdot \vec{V})}{\NodePotential(\vec{V}) = \varphi}
    \\ &\leq 4 \BalancingSpeed^2 \cdot (\sigma^2 - 1) \cdot \left(\varphi - \BigAutoExpCond{\NodePotential\left(\MixMatBeta{1}(\tau) \cdot \vec{V}\right)}{\NodePotential(\vec{V}) = \varphi}\right)^2
    \\ &= 4 (\sigma^2 - 1) \cdot \left(\BalancingSpeed \cdot \BigAutoExpCond{\varphi - \NodePotential\left(\MixMatBeta{1}(\tau) \cdot \vec{V}\right)}{\NodePotential(\vec{V}) = \varphi}\right)^2
    \\ &\overset{(\ref{eq:seven})}{\leq} 4 (\sigma^2 - 1) \cdot \left(\BigAutoExpCond{\varphi - \NodePotential\left(\MixMatBeta{\BalancingSpeed}(\tau) \cdot \vec{V}\right)}{\NodePotential(\vec{V}) = \varphi}\right)^2
    \\ &= 4 (\sigma^2 - 1) \cdot \left( \BigAutoExpCond{\NodePotential\left(\MixMatBeta{\BalancingSpeed}(\tau) \cdot \vec{V}\right)}{\NodePotential(\vec{V}) = \varphi} -  \varphi\right)^2.
    \end{align*}
    
    We apply the second statement of \cref{lem:drift} with $p=n^{-\gamma}$, $\delta = 0.5$, and $h(x) \coloneqq \BalancingSpeed \cdot g(x)$, which is an increasing function as $g$ is increasing by the definition of $(g,\sigma^2)$-good, and get
    \[\BigAutoProb{\sum_{\tau=1}^{t-t_0} \NodePotential(\MixMatTT{\tau}{t}_{k,\cdot}) \leq 2 \cdot \int_0^1 \frac{x}{\BalancingSpeed \cdot g(x)}\,\dx}
        \geq 1 - n^{-\gamma},\]
        where $t_0 = 8 \sigma^2 (\gamma\log(n) + \log(8 \sigma^2))$.
From this follows that with probability at least $1-n^{-\gamma}$
\begin{align*}
    \left(\GlobalDivergence_k(\MixMatSeq{t})\right)^2 =  \sum_{\tau=1}^{t-t_0} \NodePotential(\MixMatTT{\tau}{t}_{k,\cdot}) + \sum_{\tau=t-t_0+1}^t \NodePotential(\MixMatTT{\tau}{t}_{k,\cdot}) \overset{(a)}{\le} 2 \cdot \int_0^1 \frac{x}{\BalancingSpeed \cdot g(x)}\,\dx + t_0,
\end{align*}
where $(a)$ follows from the fact that 
$\NodePotential(\MixMat_{k,\cdot}) <1$ for $k$-th row of any stochastic matrix $\MixMat$. 
The lemma follows applying the definition of $t_0$.
\end{proof}

\restateLemNodePotentialChangeStatistics*
\begin{proof}
    By \cref{obs:node_potential_change_exact},
        we have
        \[\NodePotential(\vec{x}) - \NodePotential(\MixMatBeta{1} \cdot \vec{x})
            = \frac{1 - (1-1)^2}{2} \cdot \EdgePotential_{E(\MixMatBeta{1})}(\vec{x}).\]
    Rearranging this lower bound
        into \[\NodePotential(\MixMatBeta{1} \cdot \vec{x}) = \NodePotential(\vec{x}) - \frac{1}{2} \cdot \EdgePotential_{E(\MixMatBeta{1})},\]
        and expanding the definition of $\EdgePotential_{E(\MixMatBeta{1})}$ we have by linearity of expectation
    \begin{align*}
    \AutoExp{\NodePotential(\MixMatBeta{1} \cdot \vec{x})}
       &= \BigAutoExp{\NodePotential(\vec{x}) - \frac{1}{2} \cdot \sum_{\{i,j\} \in E(\MixMatBeta{1})} (x_i - x_j)^2}
    \\ &= \NodePotential(\vec{x}) - \frac{1}{2} \cdot \sum_{\{i,j\} \in E(G)} \AutoExp{\1_{\{i,j\} \in E(\MixMatBeta{1})} \cdot (x_i - x_j)^2}
    \\ &= \NodePotential(\vec{x}) - \frac{1}{2} \cdot \sum_{\{i,j\} \in E(G)} \AutoProb{\{i,j\} \in E(\MixMatBeta{1})} \cdot (x_i - x_j)^2
    \\ &\leq \NodePotential(\vec{x}) - \frac{1}{2} \cdot \sum_{\{i,j\} \in E(G)} 
\frac{1}{8d} \cdot (x_i - x_j)^2
    \\ &= \NodePotential(\vec{x}) - \frac{1}{16d} \cdot \EdgePotential_G(\vec{x}),
    \end{align*}
        where the inequality used that, for $\MixMatBeta{1} \sim \RMDistr(G)$ and all edges $e \in E(G)$, it holds that $\AutoProb{e \in E(\MixMatBeta{1})} \geq 1/(8d)$ \cite[Lemma 2]{DBLP:journals/jcss/GhoshM96}.
It finishes the proof of the first statement.

\medskip

For the second statement observe that by \cref{obs:node_potential_change_exact} we have 
    \[\NodePotential(\vec{x}) - \NodePotential(\MixMatBeta{1} \cdot \vec{x})
            = \frac{1}{2} \cdot \EdgePotential_{E(\MixMatBeta{1})}(\vec{x})\]
Then, as $\NodePotential(\vec{x})$ is constant for a given $\vec{x}$,
    \begin{equation}\label{eqn:var_change_leq_beta_sq_var_edgepot}
        \AutoVar{\NodePotential(\MixMatBeta{1} \cdot \vec{x})} = \AutoVar{\NodePotential(\vec{x}) - \NodePotential(\MixMatBeta{1} \cdot \vec{x})}
            = \BigAutoVar{\frac{1}{2} \cdot \EdgePotential_{\MixMatBeta{1}}(\vec{x})}
            = \frac{1}{4} \AutoVar{\EdgePotential_{\MixMatBeta{1}}(\vec{x})}.\end{equation}

Recall that the matching distribution $\RMDistr(G)$ is obtained as follows.
First, generate a random edge set $S$ as follows.
For each $e \in E(G)$, $e\in S$ with probability $\pmax \coloneqq \AutoProb{e \in S} = 1/(4d) - 1/(64d^2) \leq 1/(4d)$, independently of all other edges.
Then, some edges of $S$ are deleted to create a proper matching, resulting in  $E(\MixMatBeta{1}) \subseteq S$.
Hence 
\[
0 \leq \EdgePotential_{E(\MixMatBeta{1})}(\vec{x}) = \sum_{\{i,j\} \in E(\MixMatBeta{1})} (x_i - x_j)^2 \leq \sum_{\{i,j\} \in S} (x_i - x_j)^2 = \EdgePotential_S(\vec{x}),
\]
and
\begin{equation}\label{eqn:var_pot_matching_leq_var_pot_s_etc}
\AutoVar{\EdgePotential_{E(\MixMatBeta{1})}(\vec{x})} \leq \AutoExp{(\EdgePotential_{E(\MixMatBeta{1})}(\vec{x}))^2}
\leq \AutoExp{(\EdgePotential_S(\vec{x}))^2}
= \AutoVar{\EdgePotential_S(\vec{x})} + (\AutoExp{\EdgePotential_S(\vec{x})})^2.
\end{equation}
Observe that $\EdgePotential_S(\vec{x})$ can be expressed as $\EdgePotential_S(\vec{x})=\sum_{\{i,j\} \in E(G)} (x_i - x_j)^2 \1_{\{i,j\} \in S}$ with $\AutoProb{\1_{\{i,j\} \in S} = 1}=\pmax$.
Thus,
\begin{align*}
\AutoExp{\EdgePotential_S(\vec{x})}
   &= \sum_{\{i,j\} \in E} (x_i - x_j)^2 \cdot \AutoExp{\1_{\{i,j\} \in S}}
    = \pmax \cdot \sum_{\{i,j\} \in E} (x_i - x_j)^2
    = \pmax \cdot \EdgePotential_G(\vec{x}); \\
\AutoVar{\EdgePotential_S(\vec{x})}
   &= \sum_{\{i,j\} \in E} (x_i - x_j)^4 \cdot \AutoVar{\1_{\{i,j\} \in S}}
    = \sum_{\{i,j\} \in E} (x_i - x_j)^4 \cdot \pmax(1-\pmax) 
\\ &= \pmax(1-\pmax) \cdot \sum_{\{i,j\} \in E} (x_i - x_j)^4
\\ &\leq \pmax \cdot \sum_{\{i,j\} \in E} (x_i - x_j)^2 \cdot \max_{\{k,l\} \in E} (x_k - x_l)^2
\\ &\leq \pmax \cdot \EdgePotential_G(\vec{x}) \cdot \max_{\{k,l\} \in E} (x_k - x_l)^2.
\end{align*}
By using \cref{lem:edge_potential_bounds}(3) and then \cref{claim:hitting_time_resistance_relation}(1) we get that
\[
\max_{\{k,l\} \in E} (x_i - x_j)^2
\leq \MaxEdgeResistance \cdot \EdgePotential_G(\vec{x})
\leq \frac{\EdgeHittingTime}{\abs{E}} \cdot \EdgePotential_G(\vec{x}).
\]
Hence,
\begin{align}
\AutoVar{\EdgePotential_{E(\MixMatBeta{1})}(\vec{x})}
&\overset{{\!\!\!\!(\ref{eqn:var_pot_matching_leq_var_pot_s_etc})\!\!\!\!}}{\leq} \AutoVar{\EdgePotential_S(\vec{x})} + (\AutoExp{\EdgePotential_S(\vec{x})})^2
\notag\\&\leq \pmax \cdot \EdgePotential_G(\vec{x}) \cdot \max_{\{k,l\} \in E} (x_k - x_l)^2 + (\pmax \cdot \EdgePotential_G(\vec{x}))^2
\notag\\&\leq \pmax \cdot \EdgePotential_G(\vec{x}) \cdot \frac{\EdgeHittingTime}{\abs{E}} \cdot \EdgePotential_G(\vec{x}) + \pmax^2 \cdot (\EdgePotential_G(\vec{x}))^2
\notag\\ &\leq \frac{1}{4d} \cdot \frac{\EdgeHittingTime}{dn/2} \cdot (\EdgePotential_G(\vec{x}))^2 + \frac{1}{16d^2} \cdot (\EdgePotential_G(\vec{x}))^2
\notag\\ &= \frac{1}{2d^2} \cdot \left(\frac{\EdgeHittingTime}{n} + \frac{1}{8}\right) \cdot \EdgePotential_G(\vec{x})^2. \label{eqn:var_edgepot_leq_pmax_etc}
\end{align}
Applying the first statement of this lemma we get

\begin{equation}\label{eqn:edgepot_bound_nodepot_dev_from_mean}
\EdgePotential_G(\vec{x})
\leq 16d \cdot (\NodePotential(\vec{x}) - \AutoExp{\NodePotential(\MixMatBeta{1} \cdot \vec{x})}).
\end{equation}
Putting everything together the second statement follows from
\begin{align*}
\AutoVar{\NodePotential(\MixMatBeta{1} \cdot \vec{x})}
&\overset{(\ref{eqn:var_change_leq_beta_sq_var_edgepot})}{\leq} \frac{1}{4} \cdot \AutoVar{\EdgePotential_{\MixMatBeta{1}}(\vec{x})}
\overset{(\ref{eqn:var_edgepot_leq_pmax_etc})}{\leq} \frac{1}{4} \cdot \frac{1}{2d^2} \cdot \left(\frac{\EdgeHittingTime}{s} + \frac{1}{8}\right) \cdot (\EdgePotential_G(\vec{x}))^2
\\ &\overset{(\ref{eqn:edgepot_bound_nodepot_dev_from_mean})}{\leq} \frac{1}{8d^2} \cdot \left(\frac{\EdgeHittingTime}{n} + \frac{1}{8}\right) \cdot \left(16d \cdot (\NodePotential(\vec{x}) - \AutoExp{\NodePotential(\MixMatBeta{1} \cdot \vec{x})}\right)^2
\\ &= 32 \cdot \left(\frac{\EdgeHittingTime}{n} + \frac{1}{8}\right) \cdot \left(\NodePotential(\vec{x}) - \AutoExp{\NodePotential(\MixMatBeta{1} \cdot \vec{x})}\right)^2
\\ &\leq \left(32 \cdot \frac{\EdgeHittingTime}{n} + 4\right) \cdot \left(\NodePotential(\vec{x}) - \AutoExp{\NodePotential(\MixMatBeta{1} \cdot\vec{x})}\right)^2 \qedhere
\end{align*}
\end{proof}

\restateEdgePotentialBounds*

\begin{proof}
First note that for all $\vec{x} \in \R^n$, $a, b \in \R$, and $S \subseteq E(G)$,
\begin{equation}\label{obs:edge_potential_quadratic}\begin{aligned}
\EdgePotential_S(a \cdot \vec{x} + b)
   &= \sum_{\{i, j\} \in S} ((a \cdot x_i + b) - (a \cdot x_j + b))^2
    = \sum_{\{i, j\} \in S} (a \cdot x_i + b - a \cdot x_j - b)^2
\\ &= \sum_{\{i, j\} \in S} a^2 (x_i - x_j)^2
    = a^2 \EdgePotential(\vec{x}).
\end{aligned}\end{equation}

The proof of the first part is similar to that of Theorem 2.6 in~\cite{DBLP:conf/focs/SauerwaldS12}.
First, see that
    \begin{align*}
        \EdgePotential_G(\vec{x})
           &= \sum_{\{i,j\} \in E(G)} (x_i - x_j)^2
            = \sum_{\{i,j\} \in E(G)} (x_i^2 - 2x_i x_j + x_j^2)
        \\ &= \sum_{i \in [n]} d \cdot x_i^2 - \sum_{i,j \in [n]} \AdjacencyMat_{i,j} x_i x_j
            = d \cdot \inp{\vec{x}}{\vec{x}} - \sum_{i \in [n]} x_i \left(\sum_{j \in [n]} \AdjacencyMat_{i,j} x_j\right)
        \\ &= d \cdot \inp{\vec{x}}{\IdentityMat \vec{x}} - \inp{\vec{x}}{\AdjacencyMat \vec{x}}
            = d \cdot \inp{\vec{x}}{(\IdentityMat - \AdjacencyMat / d) \vec{x}}
            = d \cdot \inp{\vec{x}}{\Laplacian(G) \vec{x}}.
\end{align*}
As $\EdgePotential_G(\vec{x} - b) = \EdgePotential_G(\vec{x})$ by \cref{obs:edge_potential_quadratic},
    we may assume w.l.o.g.\ that $\inp{\vec{x}}{\vec{1}} = 0$ by subtracting $b \coloneqq \inp{\vec{x}}{\vec{1}} / n$ from every coordinate of $\vec{x}$.
For such a vector we have $\NodePotential(\vec{x}) = \norm{x}^2_2 = \inp{\vec{x}}{\vec{x}}$, and 
\begin{align*}
    \EdgePotential_G(\vec{x})
       &= d \inp{\vec{x}}{\Laplacian(G) \vec{x}}
        = d \cdot \frac{\inp{\vec{x}}{\Laplacian(G) \vec{x}}}{\inp{\vec{x}}{\vec{x}}} \cdot \NodePotential(\vec{x})
      \geq d \cdot \NodePotential(\vec{x}) \cdot \min_{\substack{\vec{a} \in \R^n \setminus \{\vec{0}\} \\ \inp{\vec{a}}{\vec{1}} = 0}} \frac{\inp{\vec{a}}{\Laplacian(G) \vec{a}}}{\inp{\vec{a}}{\vec{a}}}
      \\&  = d \cdot \SpectralGap(\Laplacian(G)) \cdot \NodePotential(\vec{x}),
\end{align*}
    where the final equality is due to the min-max theorem and the fact that the smallest eigenvalue of $\Laplacian(G)$ is $0$,
        with its associated eigenvector being $\vec{1}$.

For the second part,
    let $i, j \in [n]$ be two distinct nodes of the graph with $x_i \neq x_j$.
Then
    \begin{equation}\label{eqn:general_resistance_bound}
    \EdgePotential_G(\vec{x})
        = (x_i - x_j)^2 \cdot \EdgePotential_G\left(\frac{\vec{x} - x_j}{x_i - x_j}\right)
        \geq (x_i - x_j)^2 \cdot \min_{\substack{\vec{a} \in \R^n \\ a_i = 1 \\ a_j = 0}} \EdgePotential_G(\vec{a})
        = \frac{(x_i - x_j)^2}{\ResistiveDistance{i}{j}},
    \end{equation}
    where the first equality uses \cref{obs:edge_potential_quadratic},
        the central inequality holds because the argument of $\EdgePotential_G$ is a vector $\vec{a} \in \R^n$ with $a_i = 1$ and $a_j = 0$,
        and the final equality is by Dirichlet's principle (\cref{thm:dirichlet_principle}).
Note that the bound also holds when $x_i = x_j$.

Given \cref{eqn:general_resistance_bound},
    we now show that $\EdgePotential_G(\vec{x})$ is larger than the first, resp. second, term inside the maximum of the second part's statement.
For the first term,
    we choose $i$ and $j$ such that $x_i - x_j = \discr(\vec{x})$,
    and recall that $\ResistiveDistance{i}{j} \leq \ResistiveDiameter$ for all $i,j \in [n]$.
Then, \cref{eqn:general_resistance_bound} states that
    \(\EdgePotential_G(\vec{x}) \geq \discr(\vec{x})^2 / \ResistiveDiameter,\)
and it remains to bound $\discr(\vec{x})$ from below by $\NodePotential(\vec{x})$.
To that end, as the vector $\vec{x}$ is stochastic by assumption,
    the sum over all its entries is 1,
    and there is at least one $k \in [n]$ with $x_k \leq 1/n$.
Hence, $\discr(\vec{x}) \geq \max_{k \in [n]} (x_k - 1/n)$, and so
\begin{align*}\discr(\vec{x}) &\geq \discr(\vec{x}) \cdot \underbrace{\sum_{k \in [n]} x_k}_{= 1}
                  \geq \sum_{k \in [n]} \underbrace{\vphantom{\sum_{k \in [n]}}\left(x_k - \frac{1}{n}\right)}_{\leq \discr(\vec{x})} x_k - \underbrace{\frac{1}{n} \cdot \sum_{k \in [n]} \left(x_k - \frac{1}{n}\right)}_{=\frac{1}{n} \cdot 0 = 0} = \sum_{k \in [n]} \left(x_k - \frac{1}{n}\right)^2 \\ & = \NodePotential(\vec{x}),\end{align*}
        as needed to complete the bound for the first term.

For the second term,
    we choose $i$ and $j$ such that $x_i = \max_{k \in [n]} x_k$, $x_j \leq x_i - 2/3 \cdot \discr(\vec{x})$
    with the distance $D$ between $i$ and $j$ being minimal.
As $x_i \geq \discr(\vec{x})$,
    each of the entries of $\vec{x}$ for the $D-1$ non-terminal nodes on a shortest path between $i$ and $j$
        is at least $\discr(\vec{x}) / 3$.
As $\vec{x}$ is stochastic by assumption, the sum of all loads is at most $1$, and we have
    \[\discr(\vec{x}) + (D - 1) \cdot \frac{\discr(\vec{x})}{3} = \frac{D+2}{3} \cdot \discr(\vec{x}) \leq 1,\]
    which implies
    \(D \leq 3 / \discr(\vec{x}).\)
Since $\ResistiveDistance{i}{j}$ is bounded by the standard distance between $i$ and $j$ (see \cref{lem:res_dist_leq_dist}),
    and $x_i - x_j \geq 2/3 \cdot \discr(\vec{x}),$
    we thus have, by \cref{eqn:general_resistance_bound},
    \[\EdgePotential_G(\vec{x})
        \geq \frac{(x_j - x_i)^2}{\ResistiveDistance{i}{j}}
        \geq \frac{(2/3 \cdot \discr(\vec{x}))^2}{3 / \discr(\vec{x})}
        = \frac{4 \cdot \discr(\vec{x})^3}{27}
        \geq \frac{4 \cdot \NodePotential(\vec{x})}{27},\]
    where the final inequality uses $\discr(\vec{x}) \geq \NodePotential(\vec{x})$ as shown above.

For the third statement we first rearrange \cref{eqn:general_resistance_bound} to see that, for all $i \neq j$,
    \[(x_i - x_j)^2 \leq \EdgePotential_G(\vec{x}) \cdot \ResistiveDistance{i}{j}.\]
Taking the maximum over all $\{i, j\} \in E(G)$ on both sides gives us
    \[\max_{\{i, j\} \in E(G)} (x_j - x_i)^2 \leq \EdgePotential_G(\vec{x}) \cdot \max_{\{i, j\} \in E(G)} \ResistiveDistance{i}{j}
        = \EdgePotential_G(\vec{x}) \cdot \MaxEdgeResistance,\]
    as claimed, where the final equality is by definition of $\MaxEdgeResistance$.
\end{proof}

The following lemma is well-known, we state it for completeness. 
It relates the hitting time of a graph $G$ to its resistive diameter and the edge hitting time of $G$ to the $\MaxEdgeResistance$.
\begin{lemma}
[label=claim:hitting_time_resistance_relation,restate=restateHittingTimeResistanceRelation]
    For any graph $G = (V, E)$
    \begin{enumerate}
        \item $\MaxEdgeResistance \cdot \abs{E} \leq \EdgeHittingTime \leq 2 \cdot \MaxEdgeResistance \cdot \abs{E}$,
                and
        \item $\ResistiveDiameter \cdot \abs{E} \leq \HittingTime \leq 2 \cdot \ResistiveDiameter \cdot \abs{E}$.
    \end{enumerate}
\end{lemma}

\begin{proof}
Recall that \[\EdgeHittingTime \coloneqq \max_{i,j \in V, \{i,j\} \in E} H(i,j),\]
and that \[\MaxEdgeResistance \coloneqq \max_{i,j \in V, \{i,j\} \in E} \ResistiveDistance{i}{j}.\]
For the first inequality, let $i,j \in V$ be adjacent nodes for which $\ResistiveDistance{i}{j} = \MaxEdgeResistance$.
Then, by \cref{cor:hitting_time_resistance_relation_simple},
\[2 \cdot \abs{E} \cdot \MaxEdgeResistance
\leq 2 \cdot \abs{E} \cdot \ResistiveDistance{i}{j}
\leq 2 \cdot \max{H(i,j), H(j,i)}
\leq 2 \cdot \EdgeHittingTime,\]
which becomes the first inequality after dividing by 2 on both sides.
For the second inequality, let $i,j \in V$ be adjacent nodes for which $\EdgeHittingTime = H(i,j)$.
Then, again by \cref{cor:hitting_time_resistance_relation_simple},
\[\EdgeHittingTime = H(i,j) \leq 2 \cdot \abs{E} \cdot \ResistiveDistance{i}{j}
\leq 2 \cdot \abs{E} \cdot \MaxEdgeResistance.\]

The second statement is entirely analogous, except that the $i,j \in V$ are no longer required to be adjacent, and that they are chosen such that $\ResistiveDistance{i}{j} = \ResistiveDiameter$ for the first inequality, or, for the second inequality, that $H(i,j) = \HittingTime$.
\end{proof}

\subsection{Omitted Details from the Proof of \cref{lem:disc:dyn}}

\label{apx:proof_discr_dyn}

\begin{proof}[Proof of \cref{claim:integral_bound}]
First, expanding the definition of $g_G(x)$, pulling out constant factors, and simplifying fractions
    results in
\[\int_0^1 \frac{x}{g_G(x)}\,\dx =
    16d \cdot \int_0^1 \min\left\{\frac{1}{d \cdot \SpectralGap(\Laplacian(G))}, \frac{\ResistiveDiameter}{x}, \frac{27}{4x^2}\right\}\,\dx,\]
    and we write $f_1(x)$, $f_2(x)$, and $f_3(x)$ for the first, second, and third argument of the minimum.
For $x \geq 0$, the indefinite integrals of these functions are
\begin{align*}
    \int f_1(x)\,\dx = \frac{x}{d \cdot \SpectralGap(\Laplacian(G))};\quad
    \int f_2(x)\,\dx  = \ResistiveDiameter \cdot \log(x);\quad
    \int f_3(x)\,\dx  = -\,\frac{27}{4} x^{-1}.
\end{align*}

First, we show that $\int_0^1 x / g_G(x)\,\dx = \Oh(1 / \SpectralGap(\Laplacian(G)))$:
As \(\min\{f_1(x), f_2(x), f_3(x)\} \leq f_1(x)\),
    we bound the integral in question as
\begin{align*}
    \int_0^1 \frac{x}{g_G(x)}\,\dx
        \leq 16d \cdot \int_0^1 \frac{1}{d \cdot \SpectralGap(\Laplacian(G))}\,\dx
        = 16d \cdot \frac{1}{d \cdot \SpectralGap(\Laplacian(G))}
        = \Oh\left(\frac{1}{\SpectralGap(\Laplacian(G))}\right).
\end{align*}

Next, we show that $\int_0^1 x / g_G(x)\,\dx = \Oh(\sqrt{d / \SpectralGap(\Laplacian(G))})$:
Let $x_{1,3} \coloneqq \sqrt{\frac{27}{4} d \SpectralGap(\Laplacian(G))}$
        be the $x$ such that $f_1(x) = f_3(x)$.
If $x_{1,3} \leq 1$,
    then
\begin{align*}
    \int_0^1 \frac{x}{g_G(x)}\,\dx
       &\leq 16d \cdot \left(\int_0^{x_{1,3}} f_1(x)\,\dx
    + \int_{x_{1,3}}^1 f_3(x)\,\dx\right)
    \\ &= 16d \cdot \left(\frac{x_{1,3}}{d \cdot \SpectralGap(\Laplacian(G))} + \frac{27}{4} \cdot \left(-1 + x_{1,3}^{-1}\right)\right)
    \\ &= 16d \cdot \left(\sqrt{\frac{27}{4 \cdot d \cdot \SpectralGap(\Laplacian(G))}} + \sqrt{\frac{27}{4 \cdot d \cdot \SpectralGap(\Laplacian(G))}} - \frac{27}{4}\right)
        = \Oh\left(\sqrt{\frac{d}{\SpectralGap(\Laplacian(G))}} \right).
\end{align*}
But if $x_{1,3} > 1$,
    the same bound also holds:
    we showed above that the integral in question is bounded by $\Oh(1/\SpectralGap(\Laplacian(G)))$,
    so that if $x_{1,3} > 1$,
        we have an upper bound of
        \[\int_0^1 \frac{x}{g_G(x)}\,\dx
            = \Oh\left(\frac{1}{\SpectralGap(\Laplacian(G))}\right)
            = \Oh\left(\frac{x_{1,3}}{\SpectralGap(\Laplacian(G))}\right)
            = \Oh\left(\sqrt{\frac{d}{\SpectralGap(\Laplacian(G))}}\right).\]

Last, we show that $\int_0^1 x / g_G(x)\,\dx = \Oh(\HittingTime / n \cdot \log(n))$:
Let $x_{1,2} \coloneqq d \cdot \SpectralGap(\Laplacian(G)) \cdot \ResistiveDiameter$
    be the $x$ such that $f_1(x) = f_2(x)$.
If $x_{1,2} \leq 1$, then
\begin{align*}
    \int_0^1 \frac{x}{g_G(x)}\,\dx
       &\leq 16d \cdot \left(\int_0^{x_{1,2}} f_1(x)\,\dx
    + \int_{x_{1,2}}^1 f_2(x)\,\dx\right)
    \\ &= 16d \cdot \left(\frac{x_{1,2}}{d \cdot \SpectralGap(\Laplacian(G))} + \ResistiveDiameter \cdot (\log(1) - \log(x_{1,2}))\right)
    \\ &= 16d \cdot \left(\ResistiveDiameter + \ResistiveDiameter \cdot \log\left(
    \frac{1}{d \cdot \SpectralGap(\Laplacian(G)) \cdot \ResistiveDiameter}\right)\right)
    \\ &= \Oh\left(d \cdot \ResistiveDiameter \cdot \log\left(
    \frac{1}{d \cdot \SpectralGap(\Laplacian(G)) \cdot \ResistiveDiameter}\right)\right)
    \\ &= \Oh\left(\frac{\HittingTime}{n} \cdot \log\left(\frac{1}{\SpectralGap(\Laplacian(G))} \cdot \frac{n}{\HittingTime}\right)\right)
        = \Oh\left(\frac{\HittingTime}{n} \log(n)\right),
\end{align*}
    where the penultimate bound uses the fact that $\ResistiveDiameter \cdot \abs{E(G)} = \ResistiveDiameter \cdot dn/2 \leq \HittingTime$  (\cref{claim:hitting_time_resistance_relation}),
    and the final bound uses the fact that the inverse spectral gap of the normalized Laplacian $1 / \SpectralGap(\Laplacian(G))$ is bounded from above by $\Oh(n^3)$ (cf.\ \cite{DBLP:journals/aam/AksoyCTT18}),
    and that $\HittingTime \geq 1$,
    so that the argument of the logarithm is polynomial in $n$.

Otherwise, if $x_{1,2} > 1$,
    the same bound also holds:
    we show above that the integral is bounded by $\Oh(1/\SpectralGap(\Laplacian(G)))$,
    so that if $x_{1,2} > 1$
        we have an upper bound of
        \[\int_0^1 \frac{x}{g_G(x)}\,\dx
            = \Oh\left(\frac{1}{\SpectralGap(\Laplacian(G))}\right)
            = \Oh\left(\frac{x_{1,2}}{\SpectralGap(\Laplacian(G))}\right)
            = \Oh\left(d \cdot \ResistiveDiameter\right)
            = \Oh\left(\frac{\HittingTime}{n} \cdot \log(n)\right).\]

Combining the three bounds,
    we have, as claimed,
    \begin{equation*}\int_0^1 \frac{x}{g_G(x)}\,\dx = \Oh\left(\min \left\{\frac{1}{\SpectralGap(\Laplacian(G))}, \sqrt{\frac{d}{\SpectralGap(\Laplacian(G))}}, \frac{\HittingTime}{n} \cdot \log(n) \right\}\right) = \Oh(T(G)).\qedhere\end{equation*}
\end{proof}

\begin{proof}[Proof of \cref{claim:edge_hitting_time_lower_bound}]
    By the first inequality of Corollary 3.3 in~\cite{lovasz1993random} it holds
        for any nodes $i,j \in V(G)$ that
        \[H(i,j) + H(j,i) \geq \abs{E(G)} \cdot \left(\frac{1}{d(i)} + \frac{1}{d(j)}\right).\]
    As $G$ is regular we have $d(i) = d(j) = d$ and $\abs{E(G)} = dn/2$, and since the statement holds in particular for any pair of nodes that is adjacent
        this entails
        \[2 \EdgeHittingTime \geq \frac{dn}{2} \cdot \frac{2}{d} = n,\]
and the claim follows.
\end{proof}

\subsection{Bounds for Specific Graph Classes}
\label{apx:hitting-time_spectral-gap}

In this appendix  we show bounds on the discrepancy for specific graph classes. Note that we assume that initially the system is empty.

\begin{corollary}
Let $\LoadVecT{t}$ be the state of process $\SyncProc{\RMDistr(G)}{\BalancingSpeed}{m}$
where $\LoadVecT{0} = \vec{0}$.
For an arbitrary $t$ it holds w.h.p.\ and in expectation
\begin{itemize}
    \item $\discr(\LoadVecT{t}) = \Oh(\sqrt{m} \log(n))$ for any regular graph.
    \item $\discr(\LoadVecT{t}) = \Oh(\log(n) + \sqrt{m\log(n)})$ for cycles and constant-degree regular graphs.
    \item $\discr(\LoadVecT{t}) = \Oh(\log(n) + \sqrt{m/n} \cdot \log^{3/2}(n))$ for the two-dimensional torus graphs.
    \item $\discr(\LoadVecT{t}) = \Oh((1 + \sqrt{m/n}) \cdot \log(n))$ for torus graphs with $\ge 3$ dimensions,  the hypercube, and all $d$-regular graphs with $d \geq \lfloor n/2 \rfloor$.
\end{itemize}
\end{corollary}

To show the above corollary we require bounds on $T(G)$ (\cref{lem:bound-on-T(G)}) and bounds on $\EdgeHittingTime$ (\cref{lem:bound-edge-hitting-time}). Then the corollary immediately follows from 
\cref{thm:main_sync_random}.

In the following lemma we provide some bounds on $T(G)$ for several specific graph classes.
\begin{lemma}
\label{lem:bound-on-T(G)}
    Assume $G$ is a graph with $n$ nodes.
    \begin{itemize}
        \item For constant-degree regular graphs $G$ we have $T(G) = \Oh(n)$.
        \item For a two-dimensional $k \times k$ toroidal mesh  $G$ we have $T(G) = \Oh(\log^2(n))$.
        \item For a $r$-dimensional $k \times \cdots \times k$ toroidal mesh  (with $r \geq 3$) we have $T(G) = \Oh(\log(n))$.
        \item For a $r$-dimensional hypercube $G$ we have $T(G) = \Oh(\log(n))$.
        \item For a $d$-regular graph $G$ with $d \ge \lfloor \frac{n}{2} \rfloor$ we have $T(G) = \Oh(\log(n))$.
        \item For an arbitrary $d$-regular graph $G$ we have $T(G) = \Oh(n \log(n))$.
    \end{itemize}
\end{lemma}

\begin{proof}
    Recall that $T(G) = \min\left\{1/\SpectralGap(\Laplacian(G)), \sqrt{d/\SpectralGap(\Laplacian(G))}, (\HittingTime / n) \cdot \log(n)\right\}$,
        and that
        $\HittingTime \leq 2 \cdot \ResistiveDiameter \cdot \abs{E}$ (\cref{claim:hitting_time_resistance_relation}),
        so that $\HittingTime / n = \Oh(d \cdot \ResistiveDiameter)$.

    For $d$-regular graphs with $d$ being constant,
        $1 / \SpectralGap(\Laplacian(G)) = \Oh(n \cdot d \cdot (\mathrm{diam}(G) + 1))$ by \cite{LANDAU19815}, where $\mathrm{diam}(G)$ diameter of $G$.
    As $\mathrm{diam}(G) \leq n$ and $d$ is constant,
        $1 / \SpectralGap(\Laplacian(G)) = \Oh(n^2)$,
        so that $T(G) = \Oh(\sqrt{d/\SpectralGap(\Laplacian(G))}) = \Oh(n)$.

    For the two-dimensional $k \times k$ toroidal mesh,
        $d \leq 4$ and $\ResistiveDiameter = \Oh(\log(n))$ by \cite[Theorem 6.1]{DBLP:journals/cc/ChandraRRST97},
        so that $T(G) = \Oh((\HittingTime / n) \cdot \log(n)) = \Oh(\log^2(n))$.

    For a $r$-dimensional $k \times \cdots k$ toroidal mesh with $r \geq 3$, as well as the $r$-dimensional hypercube,
        $d \leq 2r$ and
        $\ResistiveDiameter = \Oh(r^{-1})$ by \cite[Theorem 6.1]{DBLP:journals/cc/ChandraRRST97},
        so that $T(G) = \Oh((\HittingTime / n) \cdot \log(n)) = \Oh((d \cdot \ResistiveDiameter) \log(n)) = \Oh(r \cdot r^{-1} \cdot \log(n)) = \Oh(\log(n))$.

    For a $d$-regular graph $G$ with $d \geq \lfloor \frac{n}{2} \rfloor$,
        $\ResistiveDiameter = \Oh(d^{-1})$ by \cite[Theorem 3.3]{DBLP:journals/cc/ChandraRRST97},
        so that $T(G) = \Oh((\HittingTime / n) \cdot \log(n)) = \Oh((d \cdot \ResistiveDiameter) \log(n)) = \Oh(d \cdot d^{-1} \cdot \log(n)) = \Oh(\log(n))$.

    For general $d$-regular graphs $G$,
        $\HittingTime\le 3n^2-nd$ by \cite[Proposition 10.16]{LevinPeresbook},
        so that $T(G) = \Oh((\HittingTime / n) \cdot \log(n)) = \Oh((n^2 / n) \log(n)) = \Oh(n \log(n))$.
\end{proof}

To bound $\EdgeHittingTime$ for many specific graph classes
    we use the following.
\begin{theorem}[Theorem 2.10 of \cite{lovasz1993random}, citing \cite{Keilson1979}]\label{thm:expected_hitting_time_to_neighbors}
    Let $G$ be a graph and $i\in [n]$ be one of its nodes.
    Then if $J \in [n]$ is chosen uniformly at random from the neighbors of $i$ in $G$, $\AutoExp{H(i,J)}=2\abs{E}/d(i) -1$, where $d(i)$ is the degree of $i$ in $G$.
\end{theorem}
This gives us the following bounds.
\begin{lemma}{\label{lem:bound-edge-hitting-time}}
    Assume $G$ is a graph with $n$ nodes.\begin{itemize}
        \item For $G$ being a toroidal mesh (including cycles and hypercubes),
        or being a $d$-regular graph with $d \geq \lfloor n/2 \rfloor$,
        we have $\EdgeHittingTime = \Oh(n)$
        \item For an arbitrary $d$-regular graph $G$ we have
        $\EdgeHittingTime \leq dn$.
    \end{itemize}
\end{lemma}

\begin{proof}
Recall that $\EdgeHittingTime \coloneqq \max_{i,j \in V, \{i,j\} \in E} H(i,j)$.
Toroidal meshes are \emph{symmetric} or \emph{arc-transitive} graphs:
    for every two ordered pairs of adjacent nodes $(i_1,j_1)$ and $(i_2,j_2)$ there is a graph automorphism $f$ such that $f(i_1)=i_2$ and $f(j_1)=j_2$.
Hence, for every such two ordered pairs, $H(i_1, i_1) = H(i_2, j_2)$,
    and thus $\EdgeHittingTime = H(i,j)$ for any pair of adjacent nodes $i,j$.
So applying \cref{thm:expected_hitting_time_to_neighbors}
    shows that $\EdgeHittingTime = 2\abs{E}/d -1$.
As $\abs{E} = dn/2$ for $d$-regular graphs, $\EdgeHittingTime = 2(dn/2)/d - 1 = n - 1 = \Oh(n),$
    as claimed.

For dense graphs
    we bound $\EdgeHittingTime$ as $\EdgeHittingTime\le \HittingTime\le 2 \cdot \ResistiveDiameter \cdot \abs{E}$ (see \cref{claim:hitting_time_resistance_relation}).
    As $\ResistiveDiameter = \Oh(1/d)$ by \cite[Theorem 3.3]{DBLP:journals/cc/ChandraRRST97}, we get
        since $\abs{E} = dn/2$ that $\EdgeHittingTime = \Oh(dn / d) = \Oh(n)$.

For arbitrary $d$-regular graphs,
    $\EdgeHittingTime \leq 2 \cdot\MaxEdgeResistance \cdot \abs{E}$
        by the first statement of \cref{claim:hitting_time_resistance_relation}.
As $\abs{E} = dn/2$ for a $d$-regular graph,
    and as $\MaxEdgeResistance \leq 1$ (by definition of $\MaxEdgeResistance$ and \cref{lem:res_dist_leq_dist}),
    we thus have $\EdgeHittingTime \leq 2 \cdot 1 \cdot dn/2 = dn$.
\end{proof}

\section{Balancing Circuit Model}\label{apx:analysis_balancing_circuit}

In this appendix we prove \cref{thm:main_sync:lower}. 
The proof is similar to Theorem 1.2 in \cite{DBLP:conf/icalp/CaiS17}.

\begin{proof}[Proof of \cref{thm:main_sync:lower}]

First we show a lower bound on $\NodeDynamicContribT{k}{t}$.
     The idea is to decompose $\NodeDynamicContribT{k}{t}$ into sum of independent $Y_\ell$ random variable which have expected value zero. It then remains to show that $\sum_{\ell} \E\left[\left|Y_\ell^3\right|\right]$ is properly bounded. It allows us to apply a concentration inequality to the sum. To do so, we define several intermediate random variables similar to the proof of \cref{lem:mixing_well_means_balancing_well}.
     
Fix round $t$ and consider node $k\in [n]$ such that $\GlobalDivergence_k(\mixMatSeq{t})=\GlobalDivergence(\mixMatSeq{t})$. Recall that,
\[
 \NodeDynamicContribT{k}{t}
        =\sum_{\tau=1}^t\sum_{w\in[n]} \mixMatTT{\tau}{t}_{k,w}\cdot \NodeAllocT{w}{\tau}.
\]

We define indicator random variables $\BallsRoundNode{\tau}{j}{w}$ for $\tau\in[t]$, $j\in[m]$ and $w\in [n]$ as follows. 
\[\BallsRoundNode{\tau}{j}{w}\coloneqq 
    \left\{ \begin{array}{cc}
       1,  & \mbox{ if $j$-th load item of step $\tau$ goes to node $w$}, \\
       0,  &  \mbox{otherwise.}
    \end{array}
     \right.
     \]
Note that for fixed $j$ and $\tau$, $\sum_{w\in [n]} \BallsRoundNode{\tau}{j}{w} =1$ and $\Pr\left[\BallsRoundNode{\tau}{j}{w}=1\right]=1/n$.
Recall that $\NodeAllocT{w}{\tau}$ can be expressed as $\sum_{j\in [m]} \BallsRoundNode{\tau}{j}{w}$.
It then follows that
\[
    \NodeDynamicContribT{k}{t}=\sum_{\tau=1}^t \sum_{j\in[m] }\sum_{w\in [n]} \left(  \mixMatTT{\tau}{t}_{k,w}\cdot \BallsRoundNode{\tau}{j}{w}\right).
\]
We define the derivative from the average for $\NodeDynamicContribT{k}{t}$ as 
\[\widetilde{D}_k(t)\coloneqq \sum_{\tau =1}^t \sum_{k \in [m] }\underbrace{\sum_{w\in [n]} \left(  \mixMatTT{\tau}{t}_{k,w}\cdot\BallsRoundNode{\tau}{j}{w}-\frac{1}{n^2}\right)}_{\BallRoundContr{k}{\tau}{j}}.\]

It immediately follows that
$\widetilde{D}_k(t)=\NodeDynamicContribT{k}{t}- t\cdot m/n.$
We call 
\[\BallRoundContr{k}{\tau}{j} \coloneqq \sum_{w\in[n]}\left( \mixMatTT{\tau}{t}_{k,w}\cdot \BallsRoundNode{\tau}{j}{w}-1/n^2\right)\]
the contribution of the $j$-th load item (of step $\tau$) to $\widetilde{D}_k(t)$. 
For a fixed $\tau$ and $j$, from the linearity of expectation, it follows that
\begin{equation*}
    \E\left[\BallRoundContr{k}{\tau}{j}\right] = \sum_{w\in [n]} \E\left[ \mixMatTT{\tau}{t}_{k,w}\cdot \BallsRoundNode{\tau}{j}{w}- \frac{1}{n^2} \right]= \left(\sum_{w\in [n]}  \mixMatTT{\tau}{t}_{k,w}\cdot\frac{1}{n} \right)-\frac{1}{n} = 0,
\end{equation*}
where the last inequality follows since $\mixMatTT{\tau}{t}$ is a doubly stochastic matrix. 

Here for $\ell=(\tau-1)\cdot m + j$ such that $\tau\in[t]$ and $j\in[m]$ we define $Y_{\ell} \coloneqq\BallRoundContr{k}{\tau}{j}$ and it follows $\widetilde{D}_k(t)=\sum_{\ell=1}^{t\cdot m} Y_{\ell}$.
Note that $Y_{\ell}$'s are independent. We want to apply the Berry-Esseen Theorem~\cite{Berry90,esseen1942liapounoff}  (see \cref{thm:BerryEssen} in \cref{apx:known-results-probability-theory}). To do so, we need to compute $\Var[Y_{\ell}]$ and $\E[|Y_{\ell}|^3]$.
Then we get {\dense\begin{align*}
    &\Var\left[Y_{\ell}\right] =\BigAutoExp{{\Big( \BallRoundContr{k}{\tau}{j}-\underbrace{\E\left[\BallRoundContr{k}{\tau}{j}\right]}_{=0}\Big)^2}}= \BigAutoExp{\left(\sum_{w\in [n]}\left(\mixMatTT{\tau}{t}_{k,w}\cdot \BallsRoundNode{\tau}{j}{w}- \frac{1}{n^2}\right)\right)^2}
    \\&  =\BigAutoExp{\left(\left(\sum_{w\in[n]}\mixMatTT{\tau}{t}_{k,w}\cdot \BallsRoundNode{\tau}{j}{w}\right)-\frac{1}{n}\right)^2}
     =\frac{1}{n}  \sum_{w' \in [n]} \left(\mixMatTT{\tau}{t}_{k,w'} - \frac{1}{n}\right)^2
    = \frac{1}{n} \cdot \norm*{\mixMatTT{\tau}{t}_{k,\cdot} - \frac{\vec{1}}{n}}_2^2,
\end{align*}}
where in the second last equality we used the fact that for each $\tau$ and each $j$ exactly one of the $\BallsRoundNode{\tau}{j}{w}$ is one and all others are zero,
    and that each of the $n$ possible cases has uniform probability. Similarly we have
\begin{align*}
   \BigAutoExp{|Y_{\ell}|^3}&= \BigAutoExp{\left| \BallRoundContr{k}{\tau}{j}\right|^3}
    = \E\left[\left|\sum_{w\in [n]}\left(\mixMatTT{\tau}{t}_{k,w}\cdot \BallsRoundNode{\tau}{j}{w}- \frac{1}{n^2}\right)\right|^3\right] 
    \\&
    \overset{(a)}{=} \!\!\!\sum_{w'\in [n]}\!\!\! \BigAutoExpCond{\left|\sum_{w\in [n]}\left(\mixMatTT{\tau}{t}_{k,w}\cdot \BallsRoundNode{\tau}{j}{w} - \frac{1}{n^2}\right)\right|^3}{\BallsRoundNode{\tau}{j}{w'}=1}\cdot \Pr\left[\BallsRoundNode{\tau}{j}{w'}=1\right]
    \\& = 
    \frac{1}{n} \cdot \sum_{w'\in [n]}\left| \mixMatTT{\tau}{t}_{k,w'}-\frac{1}{n}\right|^3 \overset{(b)}{\le} \frac{1}{n} \cdot \sum_{w'\in [n]}\left( \mixMatTT{\tau}{t}_{k,w'}-\frac{1}{n}\right)^2 \le \frac{1}{n} \cdot \norm*{\mixMatTT{\tau}{t}_{k,\cdot} - \frac{\vec{1}}{n}}_2^2,
\end{align*}
where $(a)$ follows form the law to total expectation, $(b)$ from the fact that for any $w'\in [n]$, $|\mixMatTT{\tau}{t}_{k,w'}-1/n| <1$.

Recall that $\|\mixMatTT{\tau}{t}_{k,\cdot}-\frac{\vec{1}}{n}\|_2^2=\NodePotential{(\mixMatTT{\tau}{t}_{k,\cdot})}$.
By defining $F_{t\cdot m}(x)$ as the distribution of $\widetilde{D}_k(t)/\sqrt{\sum_{\ell=1}^{t\cdot m} \Var[Y_{\ell}] }$,
from \cref{thm:BerryEssen} it follows that,
\[
    \left|F_{t\cdot m}(x)-\Phi_N(x)\right| \le C_0\cdot \frac{\sum_{\ell=1}^{t\cdot m} \BigAutoExp{|Y_{\ell}|^3}}{\left(\sum_{\ell=1}^{t\cdot m} \Var[Y_{\ell}]\right)^{3/2}} \le C_0 \cdot \frac{\frac{m}{n}\cdot \sum_{\tau=1}^t \NodePotential{(\mixMatTT{\tau}{t}_{k,\cdot})}}{\left(\frac{m}{n}\cdot \sum_{\tau = 1}^t \NodePotential{(\mixMatTT{\tau}{t}_{k,\cdot})}\right)^{3/2}} =o(1),
\]
		    in which the last inequality follows from the assumption, $m\ge 4n\log(n)/\sum_{\tau=1}^t\NodePotential{(\mixMatTT{\tau}{t}_{k,\cdot})}$, and $C_0$ is some constant. Note that $\Phi_N(x)$ is the standard normal distribution.
		    Therefore it holds that,
		    \[
		    F_{t\cdot m}(x) \ge \Phi_N(x) - o(1) \ge \frac{1}{\sqrt{\pi}(x+\sqrt{x^2+2})e^{x^2}} - o(1)
		    \]
		    where the last inequality follows from [\cite{abramowitz+stegun}, Formula 7.1.13] which states 
		    \[
		    \frac{1}{\sqrt{\pi}(x+\sqrt{x^2+2})e^{x^2}} \le \Phi_N(x) \le \frac{1}{\sqrt{\pi}(x+\sqrt{x^2+4/\pi})e^{x^2}}.
		    \]
		    Hence with $x=1$ we have
		    \[
		    F_{t\cdot m}(1) \ge \frac{1}{\sqrt{\pi}(1+\sqrt{3})e} - o(1) \ge \frac{1}{16}.
		    \]
		    Therefore by replacing the definition of $F_{t\cdot m}(1)$ we get that
		    \[
		    \Pr \left[\frac{\widetilde{D}_k(t)}{\sqrt{\frac{m}{n}\cdot\sum_{\tau=1}^t \NodePotential{(\mixMatTT{\tau}{t}_{k,\cdot})}}} \ge 1\right]=\Pr\left[\widetilde{D}_k(t) \ge \sqrt{\frac{m}{n}\sum_{\tau=1}^t \NodePotential{(\mixMatTT{\tau}{t}_{k,\cdot})}}\right]\ge \frac{1}{16}.
		    \]
		    Recall that $\widetilde{D}_k(t)=\NodeDynamicContribT{k}{t}-\E\left[\NodeDynamicContribT{k}{t}\right]$, then it follows that
		    \[
		    \Pr\left[\NodeDynamicContribT{k}{t} \ge \E\left[\NodeDynamicContribT{k}{t}\right] + \sqrt{\frac{m}{n}\cdot \sum_{\tau=1}^t \NodePotential{(\mixMatTT{\tau}{t}_{k,\cdot})}}\right]\ge \frac{1}{16}.
		    \]

      Moreover, when node $k$ receives more than expectation from the allocated load items, there is (at least) one node  $w$ receiving less than expectation. Hence, 
		    \[
		    \Pr\left[\NodeDynamicContribT{k}{t} - \NodeDynamicContribT{w}{t} \ge \sqrt{\frac{m}{n}\cdot \sum_{\tau=1}^t \NodePotential{(\mixMatTT{\tau}{t}_{k,\cdot})}} \right] \ge \frac{1}{16}\cdot 1.
		    \]

		 Since $\LoadVecT{0}=\Vec{0}$, then $ \NodeInitialContribT{k}{t}= \NodeInitialContribT{w}{t}=0$. From \cref{lem:rounding:errors:are:small} it follows that $\abs{\NodeRoundingContribT{k}{t}-\NodeRoundingContribT{w}{t}} \le \sqrt{\log n}$ with probability $1-o(1)$. Since $m\ge 4n\cdot\log(n)/ \sum_{\tau=1}^t \NodePotential{(\mixMatTT{\tau}{t}_{k,\cdot})}$ and $\NodeLoadT{k}{t}=\NodeInitialContribT{k}{t}+\NodeDynamicContribT{k}{t} +\NodeRoundingContribT{k}{t}$, then it follows
   \[
		    \Pr\left[\NodeLoadT{k}{t} - \NodeLoadT{w}{t}\ge \frac{1}{2}\cdot\sqrt{\frac{m}{n}\cdot \sum_{\tau=1}^t \NodePotential{(\mixMatTT{\tau}{t}_{k,\cdot})}} \right] \ge \frac{1}{16}\cdot (1-o(1))\ge \frac{1}{17}.
		    \]
      
\end{proof}
\cref{thm:main_sync:lower} states that for a sequence of matchings $\mixMatSeq{t}$ as long as $m\ge 4n\cdot \log n/\GlobalDivergence_k(\mixMatSeq{t})$, then the load derivation of node $k$ from the expectation at round $t$ normalized by its standard deviation follows a standard normal distributed variable.

\subsection{Bounds for Specific Graph Classes}
\label{apx:bounds-specific-graphs-C}

In the following we drive some bounds on the discrepancy for specific graph classes. Note that we assume that initially the system is empty. The first corollary gives some upper bounds and the second one lower bounds. 
\Cref{cor:main_sync_circuit} and \cref{cor:main_sync:lower} are summarized in \cref{table:disc:upperbound} (in \cref{sec:conclusions}) and \cref{table:disc:lowerbound} (below), respectively.

\begin{corollary}\label{cor:main_sync_circuit}
\let\left\relax
\let\right\relax
Let $\LoadVecT{t}$ be the state of process $\SyncProc{\BCDistr(G)}{1}{m}$ at time $t$ with $\LoadVecT{0} = \vec{0}$ and assume $G$ has $n$ nodes.
For an arbitrary $t$ it holds w.h.p.\ and in expectation
\begin{itemize}
\item \(\discr(\LoadVecT{t}) =\Oh\left(\log(n)+\sqrt{(\CircuitSize\cdot m)/(n\cdot \SpectralGap{(\cMixMat}))}\cdot\sqrt{\log(n)}\right)\) for arbitrary graphs with round matrix $\RoundMat$.
\item \(\discr(\LoadVecT{t}) = \Oh\left(\log(n)+\sqrt{ m}\cdot \sqrt{\log (n)}\right)\) for cycle and regular graphs with constant $\CircuitSize$.
\item \(\discr(\LoadVecT{t}) =\Oh\left((1+\sqrt{{m}/{n}})\cdot \log(n)\right)\) for the two-dimensional torus or hypercube graphs. 
\item  \(\discr(\LoadVecT{t}) =\Oh\left(\log(n)+\sqrt{{ m}/{n}}\cdot \sqrt{\log(n)}\right)\) for constant three or more-dimensional torus.
\end{itemize}
\end{corollary}
\begin{proof}
The bounds follow from a straight-forward combination of the upper bounds on the local divergence from \cref{lem:bound-on-global-diverg} with \cref{thm:main_sync_circuit}.
\end{proof}

\begin{corollary}\label{cor:main_sync:lower}
Let $\LoadVecT{t}$ be the state of process $\SyncProc{\BCDistr(G)}{1}{m}$ at time $t$ with $\LoadVecT{0} = \vec{0}$.
It holds with constant probability that
\begin{itemize}   
\item $\discr(\LoadVecT{t}) = \Omega\left(\sqrt{m}\right)$,  for cycle, constant $d$-regular graphs, $t= \Omega(n^2)$ and $m\ge 4\log(n)$. 
\item  $\discr(\LoadVecT{t}) = \Omega\left(\sqrt{\frac{m}{n}\cdot \log(n)}\right)$ for two-dimensional torus, $t = \Omega(n)$, and $m\ge 4n$. 
\item $\discr(\LoadVecT{t}) = \Omega\left(\sqrt{\frac{m}{n}}\right)$, for constant $r\ge 3$-dimensional torus, hypercube graphs, $t \in \mathbb{N}$, and $m\ge 4n\cdot \log(n)$. 
 \end{itemize}
\end{corollary}

\begin{proof}
The bounds follow from a straight-forward combination of the bounds on the local divergence from \cref{lem:bound-on-global-diverg} to \cref{thm:main_sync:lower}.
\end{proof}

\begin{table}[t]
\caption {Asymptotic lower bounds on the discrepancy in specific graph classes.}
\label{table:disc:lowerbound}
\def\hx#1{#1&}
\def\arraystretch{1.5}
\centering
\begin{tabularx}{0.5\textwidth}{ Xc }
\toprule
Graph
& $\SyncProc{\BCDistr(G)}{1}{m}$
\\[-1ex]
& \cref{cor:main_sync:lower} \\

\midrule
 
\hx{$d$-regular graph\newline\small (const. $d$)}
$\sqrt{m}$ \\

\hx{cycle $C_n$}
$\sqrt{m}$ \\

\hx{2-D torus}
$\sqrt{(m/n)\cdot\log(n)}$ \\

\hx{$r$-D torus\newline\small (const. $r \geq 3$)}
$\sqrt{m/n}$ \\

\hx{hypercube}
$\sqrt{m/n}$ \\
\bottomrule

\end{tabularx}
\end{table}

The two corollaries above show that our bounds are almost tight for cycle graphs, constant $d$-regular graphs, $r$-dimensional torus graphs with constant $r$ and hypercube graphs. For instance, consider a cycle constructed by Odd-Even scheme and assume $m\ge \log(n)$. \cref{cor:main_sync_circuit} states that the discrepancy is, w.h.p., $\Oh(\sqrt{m \cdot \log n})$ while \cref{cor:main_sync:lower} implies that, with constant probability, the discrepancy is $\Omega(\sqrt{m})$.

\medskip

We now compute the global divergence for following concrete graphs and circuits:
For cycles of even length,
    we consider the ``Odd-Even'' scheme in which
        the first matching $\mixMatT{1}$ consists of all edges $\{j, (j+1)\pmod{n}\}$ for any odd $j \in [n]$,
        and the second matching $\mixMatT{2}$ consists of all edges $\{j,(j+1)\pmod{n}\}$ for any even $j \in [n]$.
More generally,
    for $r$-dimensional torus with node set $[n^{1/r}]^r$,
    the balancing circuit consists of $2r$ matchings in total,
        two matchings for each dimension $i$, analogously to  the cycle.
For the hypercube, the canonical choice is the dimension exchange circuit
    consisting of $\log_2(n)$ matchings,
    where nodes $u$ and $v$ are matched in $\mixMatT{i}$ if and only if
    their binary representations differ in bit $i$ only (see, e.g.,~\cite{DBLP:conf/icalp/CaiS17}).

 Recall that $\NodePotential{(\mixMatTT{\tau}{t}_{k,\cdot} )}=\norm{\mixMatTT{\tau}{t}_{k,\cdot}-\frac{\vec{1}}{n}}_2^2$ and 
 $\cMixMat \coloneqq \mixMatTT{1}{\CircuitSize}$. 
The next lemma is about the global divergence of some specific graphs for the distribution $\BCDistr(G)$.
\begin{lemma}[Global Divergence]\label{lem:bound-on-global-diverg}
   Let $G$ be a graph and consider $\BCDistr(G)$ constructed by Odd-Even scheme such that it produces the round matrix $\RoundMat$. 
            \begin{enumerate}
                \item For each $t\in N$ it holds $(\GlobalDivergence(\MixMatSeq{t}))^2=\Oh\left(\CircuitSize/  \SpectralGap(\RoundMat)\right)$.
                \item For a constant $\CircuitSize$ and each $t\in \N$ it holds $(\GlobalDivergence(\MixMatSeq{t}))^2=
                \Oh\left(n\right)$. It also holds for any $t=\Omega(n^2)$, $(\GlobalDivergence(\MixMatSeq{t}))^2=\Omega(n)$.
                \item For two-dimensional torus $G$ and for each $t\in \N$ it holds $(\GlobalDivergence(\MixMatSeq{t}))^2=\Oh\left(\log(n)\right)$. It also holds for any $t = \Omega(n)$, $(\GlobalDivergence(\MixMatSeq{t}))^2=\Omega(\log n)$.
                \item  For constant $r\ge 3$-dimensional torus $G$ and each $t\in \N$ it holds $(\GlobalDivergence(\MixMatSeq{t}))^2
                =\Oh\left(r\right)$. It also holds for any $t\in \N$, $(\GlobalDivergence(\MixMatSeq{t}))^2= \Omega(1)$.
                \item For hypercube graphs $G$ and each $t\in \N$ it holds $(\GlobalDivergence(\MixMatSeq{t}))^2=\Oh\left(\log(n)\right)$. It also holds for any $t$, $(\GlobalDivergence(\MixMatSeq{t}))^2= \Omega(1)$.
            \end{enumerate}
\end{lemma}

\begin{proof}
  Recall that the sequence of matching matrices $\mixMatSeq{t}$ has global divergence $\GlobalDivergence(\mixMatSeq{t})$, if
        \[\forall k \in [n], \sum_{\tau=1}^{t} \norm*{\mixMatTT{\tau}{t}_{k,\cdot} - \frac{\vec{1}}{n}}_2^2 \leq \left(\GlobalDivergence(\mixMatSeq{t})\right)^2.\]
Since the matchings are fixed we have $\left(\GlobalDivergence(\mixMatSeq{t})\right)^2=\max_{w\in [n]} \sum_{\tau=1}^{t}\|\mixMatTT{\tau}{t}_{w,\cdot}-\frac{\vec{1}}{n}\|_2^2$.
Consider a node $k\in [n]$ such that $\GlobalDivergence_k(\mixMatSeq{t}) = \GlobalDivergence(\mixMatSeq{t})$.
We have seen that
\begin{equation*}
\left(\GlobalDivergence_k(\mixMatSeq{t})\right)^2=\sum_{\tau=1}^{t}\norm*{\mixMatTT{\tau}{t}_{k,\cdot}-\frac{\vec{1}}{n}}_2^2=\sum_{\tau=1}^{t} \NodePotential{(\mixMatTT{\tau}{t}_{k,\cdot} )}.
\end{equation*}
Since $\NodePotential{(\cMixMatTT{1}{\tau}_{k,\cdot} )}$ is non increasing in $\tau \in \N$ and $\cMixMat \coloneqq \mixMatTT{1}{\CircuitSize}$, then
\[
\left(\GlobalDivergence_k(\mixMatSeq{t})\right)^2\le \sum_{\tau=1}^{\infty} \NodePotential{(\mixMatTT{1}{\tau}_{k,\cdot})}
\le \CircuitSize\cdot \sum_{\tau=1}^{\infty}\NodePotential{(\cMixMatTT{1}{\tau}_{k,\cdot})}.
\]
Hence, to bound $\left(\GlobalDivergence(\mixMatSeq{t})\right)^2$, it is enough to bound $\CircuitSize\cdot\sum_{\tau=1}^{\infty}\NodePotential{(\cMixMatTT{1}{\tau}_{k,\cdot})}$.

\subparagraph*{General case:} Here we get,
\begin{align*} 
     \CircuitSize \cdot  \sum_{\tau=1}^{\infty}\NodePotential{(\cMixMatTT{1}{\tau}_{k,\cdot})} \overset{(a)}{\le}  \CircuitSize\cdot\left( \sum_{\tau=0}^{\infty} (1-\SpectralGap{(\cMixMat)})^{2\tau}\right) \le \CircuitSize\cdot\left( \sum_{\tau=0}^{\infty} (1-\SpectralGap{(\cMixMat)})^{\tau}\right) =\Oh\left(  \frac{\CircuitSize}{\SpectralGap{(\cMixMat)}}\right),
\end{align*}
where $(a)$ follows from \cite[Lemma 2]{DBLP:conf/spaa/GhoshMS96}.
Note that $\NodePotential{(\cMixMatTT{1}{1}_{k,\cdot})}\le 1$.

\subparagraph*{Cycles:} Recall that in cycle $\CircuitSize=2$. It holds that
\begin{align*}
       &\CircuitSize\cdot  \sum_{\tau=1}^{\infty} \NodePotential{(\cMixMatTT{1}{\tau}_{k,\cdot})} 
       = \CircuitSize\cdot \left( \sum_{\tau=1}^{n^2}\NodePotential{(\cMixMatTT{1}{\tau}_{k,\cdot})} + \sum_{\tau=n^2+1}^{\infty}\NodePotential{(\cMixMatTT{1}{\tau}_{k,\cdot})} \right)
    \\&
     \overset{(b)}{\le} 
     \CircuitSize\cdot\left(  \sum_{\tau=1 }^{n^2} \Oh(\frac{1}{\sqrt{\tau}})+\sum_{\tau=n^2+1}^{\infty}\NodePotential{(\cMixMatTT{1}{\tau}_{k,\cdot})}  \right) 
     \\& \overset{(c)}{\le} 
     \CircuitSize\cdot\left(\Oh(\sqrt{n^2}) + \NodePotential{(\cMixMatTT{1}{n^2}_{k,\cdot})}\cdot \sum_{\tau=1}^{\infty} (1-\SpectralGap{(\cMixMat)})^{2\tau}\right)
     \\& 
     \overset{(d)}{\le} 
     \CircuitSize\cdot\left(\Oh(\sqrt{n^2}) + \Oh(\frac{1}{n})\cdot \sum_{\tau=1}^{\infty} (1-\SpectralGap{(\cMixMat)})^{2\tau}\right) \\
     &= \CircuitSize\cdot\Oh\left(n+ \frac{1}{n\cdot \SpectralGap{(\cMixMat)}}\right) 
     \overset{(e)}{\le} 
     \Oh(\CircuitSize\cdot n)=\Oh(2n),
\end{align*}
where $(b)$ and $(d)$ follows \cite{DBLP:conf/icalp/CaiS17}, $(c)$ from \cite[Lemma 2]{DBLP:conf/spaa/GhoshMS96}.
To see $(e)$, consider that the spectral gap of the round matrix corresponding to a cycle is $\Theta(1/n^2)$ \cite{DBLP:conf/focs/RabaniSW98}. Moreover, for $t= c n^2$ with some constant $c$, it follows from \cite{DBLP:conf/icalp/CaiS17} that
		    \[\sum_{\tau=1}^{t} \NodePotential{(\mixMatTT{\tau}{t}_{k,\cdot})}= c_1\cdot \sum_{\tau=1}^{t} \NodePotential{(\cMixMatTT{\tau}{t/2}_{k,\cdot})} =  \sum_{\tau=0}^{t/2-1} \Theta(\frac{1}{\sqrt{t/2-\tau}})= \Theta(\sqrt{t/2}) = \Omega(n),\]
      for $c_1\in[1,2]$.

\subparagraph*{Two-dimensional torus:}
Note that in $r$-dimensional torus graphs $\CircuitSize=2r=4$, and the spectral gap of the round matrix corresponding to a $r$-dimensional torus is $\Theta(1/n^{2/r})$ \cite{DBLP:conf/focs/RabaniSW98}. 
Hence,
\begin{align*}
\CircuitSize\cdot  \sum_{\tau=1}^{\infty} & \NodePotential{(\cMixMatTT{1}{\tau}_{k,\cdot})}  =  \CircuitSize\cdot \left( \sum_{\tau=1}^{n^2} \NodePotential{(\cMixMatTT{1}{\tau}_{k,\cdot})}+ \sum_{\tau=n^2+1}^{\infty}\NodePotential{(\cMixMatTT{1}{\tau}_{k,\cdot})} \right)
\\&
\overset{(f)}{\le} 
\CircuitSize\cdot\left( \sum_{\tau=1}^{n^2} \Oh(\frac{1}{\tau})+\sum_{\tau=n^2+1}^{\infty}\NodePotential{(\cMixMatTT{1}{\tau}_{k,\cdot})}  \right) \\
& \overset{(g)}{\le} 
\CircuitSize\cdot\left( \sum_{\tau=1}^{n^2} \Oh(\frac{1}{\tau})+ \NodePotential{(\cMixMatTT{1}{n^2}_{k,\cdot})}\cdot \sum_{\tau=1}^{\infty} (1-\SpectralGap{(\cMixMat)})^{2\tau}  \right)
\\& 
\overset{(h)}{\le} 
\CircuitSize\cdot\left(\Oh(\log(n))+ \Oh(\frac{1}{n^2})\cdot \sum_{\tau=1}^{\infty} (1-\SpectralGap{(\cMixMat)})^{2\tau}\right) 
\\ &=
\Oh\left(4\cdot\log(n)+\frac{ 4\cdot n^2}{n^2}\right)
=
\Oh\left(4\log(n)\right),
\end{align*}
where $(f)$ and $(h)$ follow from \cite{DBLP:conf/icalp/CaiS17}, $(g)$ from \cite[Lemma 2]{DBLP:conf/spaa/GhoshMS96}. Moreover, for $t= cn$ with some constant $c$, it follows from \cite{DBLP:conf/icalp/CaiS17} that
		    \[\sum_{\tau=1}^{t} \NodePotential{(\mixMatTT{\tau}{t}_{k,\cdot})}= c_1\cdot \sum_{\tau=1}^{t} \NodePotential{(\cMixMatTT{\tau}{t/4}_{k,\cdot})} =  \sum_{\tau=0}^{t/4-1} \Theta(\frac{1}{t/4-\tau})= \Theta(\log(t/4)) = \Omega(\log(n)),\]
      for $c_1\in[1,4]$. 

\subparagraph*{Constant three or more-dimensional torus:}
Let us assume $r=2(1+\epsilon)$ for some $\epsilon>0$ then 
\begin{align*}
\CircuitSize\cdot  \sum_{\tau=1}^{\infty} & \NodePotential{(\cMixMatTT{1}{\tau}_{k,\cdot})} \overset{(i)}{\le} \CircuitSize\cdot \sum_{\tau=1}^{\infty} \tau^{-(1+\epsilon)} \le \CircuitSize\cdot\left( 1 + \int_{1}^{\infty}\!\! x^{-(1+\epsilon)} \,dx \right) \le \CircuitSize\cdot( 1+ 1/\epsilon) =\Oh(2r),
\end{align*}
where $(i)$ follows form \cite{DBLP:conf/icalp/CaiS17}.

\subparagraph*{Hypercubes:} 
Similarly, it holds that
\begin{align*}
       \CircuitSize\cdot  \sum_{\tau=1}^{\infty} &\NodePotential{(\cMixMatTT{1}{\tau}_{k,\cdot})} \overset{(j)}{\le} \CircuitSize\cdot \left(\sum_{\tau=1}^{\infty} 2^{-\tau}\right) \le 2\cdot \CircuitSize=\Oh(2\log(n)),
\end{align*}
where $(j)$ follows from \cite{DBLP:conf/icalp/CaiS17}. Recall that in hypercube $\CircuitSize\le \log(n)$.

The lower bound of $1$ is trivial.
    \end{proof}

\section{Asynchronous Model}
\label{apx:asynchronous}

The following is the equivalent of \cref{lem:mixing_well_means_balancing_well} for the process $\AsyncProcName$:
\begin{lemma}\label{lem:mixing_well_means_balancing_well_async}
    Let $G$ be a regular graph, and let $t \in \N$.
    Then in $\AsyncProc{\ADistr(G)}{\BalancingSpeed}$,
        for all $k \in [n]$, $\gamma>0$,
        and for $\hat{\GlobalDivergence}_k > 0$ such that $\AutoProb{\GlobalDivergence_k(\MixMatSeq{t}) + 1 > \hat{\GlobalDivergence}_k} \leq n^{-\gamma}$,
        we have
        \begin{align*}
        \BigAutoProb{\abs*{\NodeDynamicContribT{k}{t} - t \cdot \frac{m}{n}} \geq \frac{\gamma\log(n)}{3} + \sqrt{\frac{\gamma\log(n)}{n}} \cdot \hat{\GlobalDivergence}_k}
            \leq 4 n^{-\gamma}.
    \end{align*}
\end{lemma}

\begin{proof}
    Let $\AllocVecT{\tau}$ be the vector of allocated loads in round $\tau$ and recall that we have
    \[\DynamicContribVecT{t} = \sum_{\tau=1}^t \MixMatTT{\tau}{t} \cdot \AllocVecT{\tau},\quad\textup{so that}\quad\NodeDynamicContribT{k}{t} = \sum_{\tau = 1}^t \MixMatTT{\tau}{t}_{k,\cdot} \cdot \AllocVecT{\tau}.\]
    Using $\MixMatTT{\tau}{t} = \MixMatTT{\tau+1}{t} \cdot \MixMatT{\tau}$,
        we can express the $k$th coordinate of $\DynamicContribVecT{t}$ as
        \[\NodeDynamicContribT{k}{t} = \sum_{\tau=1}^t C_k(\tau),\quad\textup{where}\ C_k(\tau) \coloneqq \MixMatTT{\tau}{t}_{k,\cdot} \cdot \AllocVecT{\tau} = \MixMatTT{\tau+1}{t}_{k,\cdot} \cdot \left(\MixMatT{\tau} \cdot \AllocVecT{\tau}\right)\]
        is the contribution of the load item allocated in round $\tau$ to $D_k(t)$.
    Note that in the second factorization of the $C_k(\tau)$,
        the two factors are independent as they concern disjoint rounds.

    Now consider the sequence $(Y(l))_{l=0}^t$ of partial sums
        $Y(l) = \sum_{\tau=t-l+1}^t (C_k(\tau) - 1/n)$
        with respect to the natural filtration $\Filter = (\FilterT{l})_{l=0}^t$ on the sequence of edges $(I(t-l), J(t-l))$.
    In particular, we have \[Y(0) = 0,\quad Y(l) - Y(l-1) = C_k(t-l) - 1/n,\quad\textup{and }Y(t) = D_k(t) - t/n,\]
        and $\FilterT{l}$ determines all edges used in rounds $t-l+1$ up to round $t$.
    To apply the martingale tail inequality \cref{corr:martingale_quad_characteristic} to $(Y(l))_{l=0}^t$,
        we need to check that $\AutoExpCond{Y(l) - Y(l-1)}{\FilterT{l-1}} = 0$
        and that $\abs{Y(l) - Y(l-1)} \leq 1$.

    For the first condition,
        note that both $\MixMatTT{\tau}{t}_{k,\cdot}$ and $\AllocVecT{\tau}$ are stochastic vectors (for the latter, this is because exactly one load item is allocated in each round in the asynchronous model).
    Thus, their inner product $C_k(\tau)$ has a value in the interval $[0;1],$
        so that $\abs{Y(l) - Y(l-1)} = \abs{C_k(t-l) - 1/n} \leq 1 - 1/n \leq 1$,
            as required.

    For the second condition, note that
    \begin{align*}
        \AutoExpCond{Y(l) - Y(l-1)}{\FilterT{l-1}}
           &= \BigAutoExpCond{C_k(t - l)}{((I(r), J(r)))_{r=t-l+1}^t} - 1/n,
    \end{align*}
        so that it is enough to show that the expected value of the $C_k(\tau)$ is $1/n$
        when conditioned on the matching choices in rounds $\tau+1$ to $t$.
    The bound given by \cref{corr:martingale_quad_characteristic} also involves the quantity 
    \begin{align*}
    \langle Y \rangle_{t}
        \coloneqq \sum_{l=1}^{t}\AutoExpCond{(Y(l) - Y(l-1))^2}{\FilterT{l-1}} 
        = \sum_{l=1}^t \AutoExpCond{(C_k(t-l) - 1/n)^2}{\FilterT{l-1}},
    \end{align*}
        so we will investigate $C_k(\tau)$ more thoroughly than would be required to compute only its conditional expectation.

    To this end, let us first make the dependence between $\MixMatT{\tau}$ and $\AllocVecT{\tau}$ more explicit.
    Let $(I(\tau), J(\tau))$ be the random orientation of the random edge selected in round $\tau$, so that the load item in round $\tau$ is allocated to $I(\tau)$, and then the load is balanced across the edge $\{I(\tau), J(\tau)\}$.
    Then
    \[(\MixMatT{\tau} \cdot \AllocVecT{\tau})_i = \begin{cases}
        1 - \BalancingSpeed / 2,\quad&\textup{if $i = I(\tau)$},\\
        \BalancingSpeed / 2,\quad&\textup{if $i = J(\tau)$},\\ 
        0,\quad&\textup{otherwise.}
    \end{cases}\]
    Using this, we may see that
    \begin{equation}\label{eqn:ckt_decomposition_async}\begin{aligned}
    C_k(\tau)
       &= \MixMatTT{\tau+1}{t}_{k,\cdot} \cdot \left(\MixMatT{\tau} \cdot \AllocVecT{\tau}\right)
        = \sum_{i \in [n]} \MixMatTT{\tau+1}{t}_{k,i} \cdot \left(\left(1 - \frac{\BalancingSpeed}{2}\right) \cdot \1_{i = I(\tau)} + \frac{\BalancingSpeed}{2} \cdot \1_{i=J(\tau)}\right)
    \\ &= \left(1 - \frac{\BalancingSpeed}{2}\right) \cdot \left(\sum_{i \in [n]} \MixMatTT{\tau+1}{t}_{k,i} \1_{i=I(\tau)}\right)
        +           \frac{\BalancingSpeed}{2}        \cdot \left(\sum_{i \in [n]} \MixMatTT{\tau+1}{t}_{k,i} \1_{i=J(\tau)}\right).
    \end{aligned}\end{equation}
    Now $\ADistr(G)$ is the uniform distribution over the edges of $G$,
        and the node to which load is allocated is a uniformly random endpoint of the chosen edge.
    Thus, $(I(\tau), J(\tau))$ is distributed uniformly over the oriented edges $\bigcup_{\{i,j\} \in E(G)} \{(i,j), (j,i)\}$.
    Since $G$ is $d$-regular,
        there are $2 \cdot \abs{E(G)} = 2 \cdot (dn / 2) = dn$ such oriented edges.
    Hence, for all $i \in [n]$,
        \begin{align*}\AutoProb{I(\tau) = i}
           &= \sum_{j \in [n]} \AutoProb{(I(\tau), J(\tau)) = (i,j)}
            = \sum_{j \in [n]} \frac{1}{dn} \cdot \1_{\{i,j\} \in E(G)}
        \\ &= \frac{1}{dn} \cdot \abs*{\{j \in [n] \mid \{i, j\} \in E(G)\}}
            = \frac{1}{dn} \cdot d
            = \frac{1}{n}.
        \end{align*}
    By an entirely analogous calculation, $\AutoProb{J(\tau) = i} = 1/n$ holds as well.
    So $I(\tau)$ and $J(\tau)$ are identically distributed (but not necessarily independent).
    Because of this,
        the two sums over $i \in [n]$ on the right-hand side of \cref{eqn:ckt_decomposition_async} are also identically distributed.
    
    We can now compute the conditional expectation of $C_k(\tau)$.
        Using \cref{eqn:ckt_decomposition_async} and linearity of expectation we see that
    \begin{equation*}\begin{aligned}
        \BigAutoExpCond{C_k(\tau)}{((I(l), J(l)))_{l=\tau+1}^t}
           &= \left(\left(1 - \frac{\BalancingSpeed}{2}\right) + \frac{\BalancingSpeed}{2}\right)
            \cdot \BigAutoExpCond{\sum_{i \in [n]} \MixMatTT{\tau+1}{t}_{k,i} \1_{i=I(\tau)}}{\MixMatTT{\tau+1}{t}}
        \\ &= 1 \cdot \sum_{i \in [n]} \AutoProb{I(\tau) = i} \cdot \MixMatTT{\tau+1}{t}_{k,i}
            = \frac{1}{n} \cdot \sum_{i \in [n]} \MixMatTT{\tau+1}{t}_{k,i} = \frac{1}{n}.
    \end{aligned}\end{equation*}
    So $\AutoExpCond{Y(l) - Y(l-1)}{\FilterT{l-1}} = 1/n - 1/n = 0,$
        as required for applying \cref{corr:martingale_quad_characteristic}.

    So all preconditions of \cref{corr:martingale_quad_characteristic} hold.
    Applying it with $\varepsilon = \gamma\log(n)$ and $\sigma = \hat{\GlobalDivergence}_k / \sqrt{n}$ yields
    \[\BigAutoProb{\abs{Y(t) - Y(0)} \geq \frac{\gamma \log(n)}{3} + \sqrt{2\gamma\log(n) / n} \cdot \hat{\GlobalDivergence}_k} \leq 2(n^{-\gamma} + \AutoProb{\langle Y \rangle_t > \hat{\GlobalDivergence}_k^2 / n}).\]
    We will now show that $\langle Y \rangle_t \leq 1/n \cdot (\GlobalDivergence_k(\MixMatSeq{t}) + 1)^2$,
        which finishes the proof after noting that then,
        \[\AutoProb{\langle Y \rangle_t > \hat{\GlobalDivergence}_k^2 / n}
            \leq \BigAutoProb{1/n \cdot (\GlobalDivergence_k(\MixMatSeq{t}) + 1)^2 > \hat{\GlobalDivergence}_k^2 / n}
            = \AutoProb{\GlobalDivergence_k(\MixMatSeq{t}) + 1 > \hat{\GlobalDivergence}_k}
            \leq n^{-\gamma},\]
        with the last inequality using the condition on $\hat{\GlobalDivergence}_k$ in the statement.

    So to bound $\langle Y \rangle_{t}$, recall that
        \[\langle Y \rangle_{t} \coloneqq \sum_{l=1}^{t}\AutoExpCond{(Y(l) - Y(l-1))^2}{\FilterT{l-1}} = \sum_{l=1}^{t}\AutoVarCond{Y(l) - Y(l-1)}{\FilterT{l-1}},\]
        with the latter equality using the fact the expected value of $(Y(l)-Y(l-1))$ conditioned on $\FilterT{l-1}$ is $0$.
    And since $Y(l) - Y(l-1) = C_k(t-l) - 1/n$ and $1/n$ is a constant,
    \[\langle Y \rangle_{t} = \sum_{l=1}^t \AutoVarCond{Y(l) - Y(l-1)}{\FilterT{l-1}}
           = \sum_{l=1}^t \BigAutoVarCond{C_k(t - l)}{((I(r), J(r)))_{r=t-l+1}^t}.\]
    By \cref{eqn:ckt_decomposition_async},
        and as for two identically distributed random variables $A$ and $B$, and $a, b \in \R^+$,
        we have $\AutoVar{aA + bB} = a^2\AutoVar{A} + 2ab\AutoCov{A, B} + b^2\AutoVar{B}
            \leq (a^2 + 2ab + b^2)\AutoVar{A} = (a + b)^2 \AutoVar{A}$:
    \begin{equation*}\begin{aligned}
        \BigAutoVarCond{C_k(\tau)}{((I(l), J(l)))_{l=\tau+1}^t}
           &\leq \left(1 - \frac{\BalancingSpeed}{2} + \frac{\BalancingSpeed}{2}\right)^2
              \cdot \BigAutoVarCond{\sum_{i \in [n]} \MixMatTT{\tau+1}{t}_{k,i} \1_{i=I(\tau)}}{\MixMatTT{\tau+1}{t}}
        \\ &= 1^2 \cdot \sum_{i \in [n]} \AutoProb{I(\tau) = i} \cdot \left(\MixMatTT{\tau+1}{t}_{k,i} - \frac{1}{n}\right)^2
        \\ &= \frac{1}{n} \cdot \norm*{\MixMatTT{\tau+1}{t}_{k,\cdot} - \frac{\vec{1}}{n}}_2^2.
    \end{aligned}\end{equation*}
    And hence we may bound $\langle Y \rangle_{t}$ from above using the global divergence:
        \begin{equation*}\begin{aligned}
        \langle Y \rangle_{t}
           &=\sum_{\tau=1}^{t} \BigAutoVarCond{C_k(\tau)}{((I(l), J(l)))_{l=\tau+1}^t}
            \leq \frac{1}{n} \cdot \sum_{\tau=1}^t \norm*{\MixMatTT{\tau+1}{t}_{k,\cdot} - \frac{\vec{1}}{n}}_2^2
        \\ &= \frac{1}{n} \left(\left(\GlobalDivergence_k(\MixMatSeq{t})\right)^2 - \norm*{\MixMatTT{1}{t}_{k,\cdot} - \frac{\vec{1}}{n}}_2^2 + \norm*{\MixMatTT{t+1}{t}_{k,\cdot} - \frac{\vec{1}}{n}}_2^2\right)
        \\ &\leq \frac{1}{n} \cdot \left(\left(\GlobalDivergence_k(\MixMatSeq{t})\right)^2 + 1\right)
            \leq \frac{1}{n} \cdot \left(\left(\GlobalDivergence_k(\MixMatSeq{t})\right) + 1\right)^2,
        \end{aligned}\end{equation*}
        which is all that remained to be shown.
\end{proof}

The next result is the analogue of \cref{lem:rmdistr_is_good}:
\begin{lemma}\label{lem:asdistr_is_good}
    Assume $G$ is an arbitrary $d$-regular graph.
    Then $\ADistr(G)$ is $(g_G, \sigma_G^2)$-good,
        where
        \[g_G(x) \coloneqq \frac{1}{dn} \cdot \max\left\{
            d \cdot \SpectralGap(\Laplacian(G)) \cdot x,
            \frac{1}{\ResistiveDiameter} \cdot x^2,
            \frac{4}{27} \cdot x^3\right\};\quad \sigma^2 = 2 \cdot \EdgeHittingTime.\]
\end{lemma}

The proof of \cref{lem:asdistr_is_good} is analogous to that of \cref{lem:rmdistr_is_good},
    except that we use \cref{prop:node_potential_change_statistics_async} stated below instead of \cref{prop:node_potential_change_statistics}.
\begin{lemma}\label{prop:node_potential_change_statistics_async}
    Let $G$ be a $d$-regular graph,
        let $\MixMat^1 \sim \ADistr(G)$,
        and let $\vec{x} \in \R^n$,
    Then
    \begin{enumerate}
    \item \(\NodePotential(\vec{x}) - \BigAutoExp{\NodePotential(\MixMatBeta{1} \cdot \vec{x})} = \frac{1}{dn} \cdot \EdgePotential_G(\vec{x}).\)
    \item \(\BigAutoVar{\NodePotential(\MixMatBeta{1} \cdot \vec{x})}
                \leq (2 \cdot \EdgeHittingTime - 1) \cdot\left(\NodePotential(\vec{x}) - \BigAutoExp{\NodePotential(\MixMatBeta{1} \cdot \vec{x})}\right)^2.\)
    \end{enumerate}
\end{lemma}

\begin{proof}
    For the first statement,
        we use \cref{obs:node_potential_change_exact}
        as well as the fact that $\ADistr(G)$ is the uniform distribution over the edges of $G$
        to see that, as claimed.
    \begin{align*}
    \NodePotential(\vec{x}) - \BigAutoExp{\NodePotential(\MixMatBeta{1} \cdot \vec{x})}
       &= \BigAutoExp{\NodePotential(\vec{x}) - \NodePotential(\MixMatBeta{1} \cdot \vec{x})}
        =  \BigAutoExp{\frac{1}{2} \cdot\EdgePotential_{\MixMat^1}(\vec{x})}
    \\ &= \frac{1}{2} \cdot \sum_{\{i,j\} \in E(G)} \frac{1}{\abs{E}} \cdot (x_i - x_j)^2
        = \frac{1}{2} \cdot \frac{1}{dn/2} \cdot \EdgePotential_G(\vec{x})
        = \frac{1}{dn} \cdot \EdgePotential_G(\vec{x}).
    \end{align*}

    For the second statement
        we first observe that $\NodePotential(\vec{x})$ is constant and by \cref{obs:node_potential_change_exact} we have
        \[\BigAutoVar{\NodePotential(\MixMatBeta{1} \cdot \vec{x})}
            = \BigAutoVar{\NodePotential(\vec{x}) - \NodePotential(\MixMatBeta{1} \cdot \vec{x})}
            = \BigAutoVar{\frac{1}{2} \cdot \EdgePotential_{\MixMatBeta{1}}(\vec{x})}.\]
    We bound this variance using the Bhatia-Davis inequality (see \cref{thm:bhatia-davis} in \cref{apx:known-results-probability-theory}).
    It states that, for a random variable $X$ taking values in $[m, M]$, and with $\mu \coloneqq \AutoExp{X}$, it is the case that \(\AutoVar{X} \leq (M - \mu)(\mu - m).\)
    Now from the definition of $\EdgePotential$,
        it is immediate that $\EdgePotential_{\MixMatBeta{1}}(\vec{x}) \geq 0$.
    For the upper bound on $\EdgePotential_{\MixMatBeta{1}}(\vec{x})$,
        recall that the matchings $\MixMatBeta{1} \sim \ADistr$ consist of just one edge,
        and so $\EdgePotential_{\MixMatBeta{1}} \leq \max_{\{i,j\} \in E(G)} (x_i - x_j)^2$.
    The latter is bounded from above by the third statement of \cref{lem:edge_potential_bounds},
        yielding
        \[\EdgePotential_{\MixMatBeta{1}}(\vec{x}) \leq \max_{\{i,j\} \in E(G)} (x_i - x_j)^2 \leq \MaxEdgeResistance \cdot \EdgePotential_G(\vec{x}).\]
    And so, by the Bhatia-Davis inequality (\cref{thm:bhatia-davis}),
        \begin{align*}
        \BigAutoVar{\frac{1}{2} \cdot \EdgePotential_{\MixMatBeta{1}}(\vec{x})}
           &\leq \left(\MaxEdgeResistance \cdot \EdgePotential_G(\vec{x}) - \frac{1}{dn} \cdot \EdgePotential_G(\vec{x})\right) \cdot \frac{1}{dn} \cdot \EdgePotential_G(\vec{x}),
        \\ &= \left(\MaxEdgeResistance \cdot dn - 1\right) \cdot \left(\frac{1}{dn} \cdot \EdgePotential_G(\vec{x})\right)^2
        \\ &\leq 2 \cdot \EdgeHittingTime \cdot \left(\NodePotential(\vec{x}) - \AutoExp{\NodePotential(\MixMatBeta{1} \cdot \vec{x}})\right)^2,
        \end{align*}
        where the last inequality used the fact that $\MaxEdgeResistance \cdot dn = 2 \cdot \MaxEdgeResistance \cdot \abs{E} \leq 2 \cdot \EdgeHittingTime$ by \cref{claim:hitting_time_resistance_relation}.
    \end{proof}

\subsection{Bounds for Specific Graph Classes}
\label{apx:async-bounds-graph-classes}

Again as in \cref{apx:hitting-time_spectral-gap} we consider specific graph classes and use the bounds on $T(G)$ and on the hitting time from \cref{apx:hitting-time_spectral-gap}.
When applied to \cref{thm:main_async} we get the following results w.h.p.\ and in expectation.

\begin{corollary}
    Let $\LoadVecT{t}$ be the state of process $\SyncProc{\RMDistr(G)}{\BalancingSpeed}{m}$
    where $\LoadVecT{0} = \vec{0}$.
    For an arbitrary $t$ it holds w.h.p.\ and in expectation
\begin{itemize}
    \item $\discr(\LoadVecT{t}) = \Oh(\sqrt{n} \log(n))$ for any regular graph. 
    \item $\discr(\LoadVecT{t}) = \Oh(\sqrt{n\log(n)})$ for cycle and constant-degree regular graphs.
    \item $\discr(\LoadVecT{t}) = \Oh(\log^{3/2}(n))$ for the two-dimensional torus graph.
    \item $\discr(\LoadVecT{t}) = \Oh(\log(n))$ for $r$-dimensional torus graphs with $r \ge 3$ dimensions, for the hypercube, and for all $d$-regular graphs with $d \geq \lfloor n/2 \rfloor$.
\end{itemize}
\end{corollary}

\section{Proof of the Drift Result}\label{apx:drift_proof}

In this appendix we give the full proof of our drift result from \cref{sec:drift}. 
We restate it for convenience.

\restateLemDrift*

\begin{proof}
    Throughout this proof we write \[f(x) \coloneqq \int_x^{x_0} \frac{1}{h(\varphi)} \dvarphi.\]

    We start by proving the first statement.
    Let $a, b \in \R^+$ with $a \leq b \leq x_0$ be two arbitrary numbers.
    Since $h$ is increasing we have $h(a) \leq h(b)$ and
        $1 / h(a) \geq 1 / h(b)$.
    Hence,
    \[f(a) - f(b)
        = \int_a^{x_0} \frac{1}{h(\varphi)} \dvarphi - \int_b^{x_0} \frac{1}{h(\varphi)} \dvarphi
        = \int_a^b \frac{1}{h(\varphi)} \dvarphi
        \geq \int_a^b \frac{1}{h(b)} \dvarphi = \frac{b - a}{h(b)}.\]
    From condition \ref{cond:drift:1} of the theorem it follows that $\AutoExpCond{X(t+1)}{X(t) = b} \leq b - h(b)$ and consequently $h(b) \leq b-\AutoExpCond{X(t+1)}{X(t) = b}$ giving us
    with $X(t)=b$
    \begin{equation}\label{eqn:new_drift_delta_f_bound}f(X(t+1)) - f(b) \geq \frac{X(t+1) - b}{\AutoExpCond{X(t+1) - b}{X(t) = b}}.\end{equation}

    We introduce a new sequence of random variables for which we will derive a lower tail bound, defined as 
        $(Y(t))_{t \in \N}$ given by $Y(0) \coloneqq 0$ and
        \[Y(t+1) \coloneqq Y(t) + \frac{X(t+1) - X(t)}{\E[X(t+1)-X(t)]}.\]
    Comparing this with \cref{eqn:new_drift_delta_f_bound} we see that regardless of the value of $X(t)$ it holds that \[f(X(t+1)) - f(X(t)) \ge \frac{X(t+1) - X(t)}{\E[X(t+1)-X(t)]}= Y(t+1)-Y(t).\]
By induction over $t$, and since $f(x_0) = \int_{x_0}^{x_0} (1 / h(\varphi))\,\dvarphi = 0$ and $Y(0) = 0$,
    we have for all $t$ \[f(X(t)) = f(X(t)) - f(x_0) \geq Y(t) - Y(0) = Y(t).\]
   From the definition of $(Y_t)_{t\ge 0}$ it follows assuming $X(t)=x$ that 
    \[\AutoExpCond{Y(t+1)-Y(t)}{X(t) = x}
        = E\left[\frac{X(t+1) - x}{\E[X(t+1)-x]}\right]
        = 1.\]
Then, from the law of total expectation we get that
\begin{align*}
\AutoExp{Y(t+1)-Y(t)} &= \sum_{x}\AutoExpCond{Y(t+1)-Y(t)}{X(t)=x}\cdot \Pr[X(t)=x]\\&= \sum_{x}1\cdot \Pr[X(t)=x] =1.
\end{align*} 
Since $Y(0)=0$ it immediately follows that $\E[Y(t)]=t$.
Furthermore, we may bound the variance of the change of $Y$ given $X(t) = x$ by
\begin{align*}
\AutoVarCond{Y(t+1)-Y(t)}{X(t) = x}
   &= 
   \BigAutoVar{\frac{X(t+1) -x}{\E[X(t+1)-x]}}
  =\frac{\BigAutoVar{X(t+1) - x}}{(\AutoExp{X(t+1) - x})^2}
\\ & \overset{(a)}{\leq} \frac{\sigma \cdot \left(\AutoExp{X(t+1)} - x\right)^2}{\left(\AutoExp{X(t+1) - x}\right)^2}
    = \sigma,
\end{align*}
where $(a)$ follows from Condition \ref{cond:drift:2} of the theorem.
The sequence $(Y(t) - \AutoExp{Y(t)})_{t\ge 0}$ is a martingale and hence fulfills the preconditions of \cref{thm:chung66} (Theorem 6.6 from \cite{DBLP:journals/im/ChungL06}) with $a_t \coloneqq 1$ and $\sigma^2_t \coloneqq \sigma$.
Note that $\E[Y(t)-\E[Y(t)]]=0$. Hence,  we obtain
        \[\BigAutoProb{Y(t)-\E[Y(t)] \leq 0 - \varepsilon} \leq \exp\left(-\,\frac{\varepsilon^2}{2 t (\sigma + 1)}\right).\]
    Recalling that $f(X(t)) \geq Y(t)$ and $\AutoExp{Y(t)} = t$
        and setting $\varepsilon = \delta t$ for some $\delta \in (0, 1)$
        we arrive at the first statement of the theorem;
        \[\AutoProb{f(X(t)) \leq (1 - \delta) t}
            \leq \exp\left(-\,\frac{\delta^2 t}{2 (\sigma + 1)}\right).\]

Next we prove the second statement and bound $\sum_{t=t_0+1}^{\infty} X(t)$.
Let $T(x) \coloneqq \min\{t \in \N \mid X(t) \leq x\}$ be a hitting time for the event that $X(t) \leq x$.
Using $\1_{ x < X(t)}$ as the indicator variable (which is one if $x < X(t)$ and zero otherwise) we can write $X(t)=\int_0^{x_0} \1_{X(t) > x}\,\dx$
because $x_0$ is fixed and $X(t)$ is non-increasing in $t$ resulting in  $X(t)\in[0,x_0]$.
As a consequence it holds that
\begin{align*}
\sum_{t=t_0+1}^{\infty} X(t)
   &= \sum_{t=t_0+1}^{\infty} \int_0^{x_0} \1_{X(t) > x}\,\dx
    = \int_0^{x_0} \left(\sum_{t=t_0+1}^{\infty} \1_{X(t) > x} \right)\,\dx
\\ &= \int_0^{x_0} \left(\sum_{t=t_0+1}^{\infty} \1_{t < T(x)}\right)\,\dx
    = \int_0^{x_0} \left(\sum_{t=t_0+1}^{T(x)-1} 1 \right)\,\dx
\\ &= \int_0^{x_0} \max\{0, T(x) - (t_0 + 1)\}\,\dx .
\end{align*}

We now proceed to bound the $T(x)$.
Using the first statement with a union bound over all $t > t_0 \coloneqq \frac{2(\sigma + 1)}{\delta^2} \cdot \left(-\log(p) + \log\left(\frac{2(\sigma + 1)}{\delta^2}\right)\right)$ gives us
\begin{align*}
\BigAutoProb{\bigvee_{t=t_0+1}^\infty f(X(t)) \leq (1 - \delta) t}
   &\leq \sum_{t=t_0+1}^\infty \exp\left(-\,\frac{\delta^2 t}{2 (\sigma + 1)}\right)
    \leq \int_{t_0}^\infty \exp\left(-\,\frac{\delta^2 t}{2 (\sigma + 1)}\right)\dt
\\ &=    \frac{2(\sigma + 1)}{\delta^2} \cdot \exp\left(-\,\frac{\delta^2 t_0}{2 (\sigma + 1)}\right)
    \eqqcolon p.
\end{align*}
As a consequence,
\[\BigAutoProb{\bigwedge_{t=t_0 + 1}^\infty t \leq \frac{f(X(t))}{1 - \delta}} \geq 1-p,\]
and 
\begin{equation}\label{eqn:prob_upper_bound_all_t}\BigAutoProb{\bigwedge_{t \in \N_0} t \leq \max\left\{t_0, \frac{f(X(t))}{1 - \delta}\right\}} \geq 1-p.\end{equation}

Recalling that $T(x) \coloneqq \min\{t \in \N \mid X(t) \leq x\}$    \cref{eqn:prob_upper_bound_all_t} implies that
\[\BigAutoProb{\bigwedge_{x < x_0} T(x) - 1 \leq \max\left\{t_0, \frac{f(X(T(x) - 1))}{1-\delta}\right\}} \geq 1 - p,\]
since $X(T(x) - 1) > x$ by the definition of $T(x)$ and $f$ is non-increasing it holds that $f(X(T(x) - 1)) \leq f(x)$.
It follows that 
\begin{equation*}
\BigAutoProb{\bigwedge_{x\le x_0} T(x) - 1 \leq \max\left\{t_0, \frac{f(x)}{1 - \delta}\right\}} \geq 1-p.\end{equation*}
        
As a consequence we get that with probability at least $1-p$ 
\[ \int_0^{x_0} \max\{0, T(x) - (t_0 + 1)\}\,\dx \leq \int_0^{x_0} \max\left\{0, \max\left\{t_0, \frac{f(x)}{1-\delta} \right\} + 1 - (t_0 + 1)\right\}\,\dx\]
Finally, we find that         
    \begin{align*}
    \MoveEqLeft \int_0^{x_0} \max\left\{0, \max\left\{t_0, \frac{f(x)}{1-\delta} \right\} + 1 - (t_0 + 1)\right\}\,\dx
    \\ &= \int_0^{x_0} \max\left\{0, \frac{f(x)}{1-\delta} - t_0\right\}\,\dx
        \leq \frac{1}{1-\delta} \int_0^{x_0} f(x)\,\dx
    \\ &= \frac{1}{1-\delta} \int_0^{x_0} \int_x^{x_0} \frac{1}{h(\varphi)}\,\dvarphi\,\dx
        = \frac{1}{1-\delta} \int_0^{x_0} \int_0^{x_0} \frac{\1_{\varphi \geq x}}{h(\varphi)} \,\dvarphi\,\dx
    \\ &= \frac{1}{1-\delta} \int_0^{x_0} \frac{1}{h(\varphi)} \int_0^{x_0} \1_{x \leq \varphi}\,\dx\,\dvarphi
        = \frac{1}{1-\delta} \cdot \int_0^{x_0} \frac{1}{h(\varphi)}  \cdot \varphi\,\dvarphi.
    \end{align*}
    
Putting everything together we see with probability at least $1-p$ that
\[
 \sum_{t=t_0+1}^{\infty} X(t) \le \frac{1}{1-\delta} \cdot \int_0^{x_0} \frac{\varphi}{h(\varphi)} \cdot \,\dvarphi. \qedhere
\]
\end{proof}

\end{document}